\documentclass[11pt,letterpaper]{article}
\usepackage[margin=1in]{geometry}

\usepackage[utf8]{inputenc}
\usepackage{amsmath,amssymb,amsthm}
\usepackage{thmtools}
\usepackage[affil-it]{authblk}
\usepackage{tabu}
\usepackage{enumitem}
\usepackage{dsfont}
\usepackage{braket}
\usepackage{cancel}
\usepackage{graphicx}
\usepackage{xcolor}
\usepackage{float}
\usepackage{array, makecell}
\usepackage{multirow}
\usepackage{esint}
\usepackage{setspace}
\usepackage{mdframed}
\usepackage[title]{appendix}

\usepackage[allcolors=blue,colorlinks=true,hyperindex=true]{hyperref}

\numberwithin{equation}{section}

\newcommand{\ketbra}[1]{\left|#1\right\rangle\left\langle #1\right|}

\DeclareMathOperator{\EV}{\mathbb{E}}

\newcommand{\vect}[1]{\boldsymbol{#1}}

\newcounter{thmcount}
\numberwithin{thmcount}{section}

\newtheorem{lemma}[thmcount]{Lemma}
\newtheorem{theorem}{Theorem}
\newtheorem*{theorem*}{Theorem}
\newtheorem{assumption}{Assumption}
\newtheorem{proposition}[thmcount]{Proposition}

\newtheorem{defn}[thmcount]{Definition}

\newtheorem*{proposition1}{Proposition 4.1 (Formal)}
\newtheorem*{theoremCLSS}{Theorem CLSS21}

\theoremstyle{definition}
\newtheorem{remark}[thmcount]{Remark}

\newcommand{\av}{{\vect{a}}}
\newcommand{\bv}{{\vect{b}}}
\newcommand{\cv}{{\vect{c}}}
\newcommand{\dv}{{\vect{d}}}

\newcommand{\zv}{{\vect{z}}}
\newcommand{\Zv}{{\vect{Z}}}
\newcommand{\xv}{{\vect{x}}}

\newcommand{\paramv}{\gammav,\betav}
\newcommand{\gammav}{{\vect{\gamma}}}
\newcommand{\Gammav}{\boldsymbol\Gamma}
\newcommand{\betav}{{\vect{\beta}}}

\newcommand{\cN}{\mathcal{N}}

\newcommand{\G}{\mathds{G}}
\newcommand{\ER}{Erd{\H{o}}s-R\'{e}nyi }
\newcommand{\Id}{\mathds{1}}

\newcommand{\M}{\mathrm{m}}
\renewcommand{\P}[1]{{(#1)}}

\newcommand{\mixed}{{\cv,\text{mixed}}}
\newcommand{\ERsub}{\textnormal{ER}}

\newcommand{\Gmixed}{\G_\mixed}

\newcommand{\OGP}{\textnormal{OGP}}
\newcommand{\OPT}{\textnormal{OPT}}

\newcommand{\barD}{\overline{D}}

\def\TT{{\mathbb T}}
\def\PP{{\mathbb P}}
\def\SS{{\mathbb S}}
\def\UU{{\mathbb U}}
\def\Z{{\mathbb Z}}
\def\C{{\mathbb C}}
\def\R{{\mathbb R}}
\def\E{{\mathbb E}}

\def\eps{{\varepsilon}}

\def\cC{{\mathcal C}}
\def\cE{{\mathcal E}}
\def\cR{{\mathcal R}}
\def\Poisson{{\rm Poisson}}
\def\Disk{{\mathbb D}}

\def\lin{{\rm lin}}
\def\pois{{\xi}}
\def\Poisson{{\rm Poisson}}
\def\Unif{{\rm Unif}}
\newcommand{\Ub}{\eta}

\newcommand{\qmax}{{q_{\rm max}}}
\newcommand{\lo}[1]{\hat{#1}}
\newcommand{\bgbbraket}[1]{\braket{\paramv|#1|\paramv}}

\newcommand{\LHS}{\textnormal{LHS}}

\def\avu{{\underline \av}}
\def\bvu{{\underline \bv}}
\def\cvu{{\underline \cv}}

\begin{document}

\title{Performance and limitations of the QAOA at constant levels\\
on large sparse hypergraphs and spin glass models}

\author[1]{Joao Basso\thanks{joao.basso@berkeley.edu}}
\author[2]{David Gamarnik\thanks{gamarnik@mit.edu}}
\author[3]{Song Mei\thanks{songmei@berkeley.edu}}
\author[4]{Leo Zhou\thanks{leozhou92@gmail.com (Corresponding author)}}

\affil[1]{\footnotesize Google Quantum AI, Venice, CA 90291}
\affil[2]{\footnotesize Operations Research Center and Sloan School of Management, MIT, Cambridge, MA 02139}
\affil[3]{\footnotesize Department of Statistics, University of California, Berkeley, CA 94720}
\affil[4]{\footnotesize Walter Burke Institute for Theoretical Physics, Caltech, Pasadena, CA 91125}

\date{September 28, 2022}

\maketitle

\begin{abstract}
\normalsize
\onehalfspacing
The Quantum Approximate Optimization Algorithm (QAOA) is a general purpose quantum algorithm designed for combinatorial optimization.
We analyze its expected performance and prove concentration properties at any constant level (number of layers) on ensembles of random combinatorial optimization problems in the infinite size limit.
These ensembles include mixed spin models and Max-$q$-XORSAT on sparse random hypergraphs.
Our analysis can be understood via a saddle-point approximation of a sum-over-paths integral.
This is made rigorous by proving a generalization of the multinomial theorem, which is a technical result of independent interest.
We then show that the performance of the QAOA at constant levels for the pure $q$-spin model matches asymptotically the ones for Max-$q$-XORSAT on random sparse \ER hypergraphs and every large-girth regular hypergraph. 
Through this correspondence, we establish that the average-case value produced by the QAOA at constant levels is bounded away from optimality for pure $q$-spin models when $q\ge 4$ and is even. 
This limitation gives a hardness of approximation result for quantum algorithms in a new regime where the whole graph is seen.
\end{abstract}

\thispagestyle{empty}

 \clearpage

\section{Introduction}
\label{sec:intro}

Quantum computers are widely believed to be more powerful than classical computers, in part due to Shor's seminal quantum algorithm for solving the classically intractable problem of integer factorization~\cite{shor1994algorithms}.
As quantum computers begin to come online, an important open question is whether we can harness their power to achieve a computational advantage on optimization problems with widespread real-world applications.
The Quantum Approximate Optimization Algorithm (QAOA) is a leading quantum algorithm designed to find approximate solutions of combinatorial optimization problems (COPs)~\cite{farhi2014quantum}.
The QAOA is computationally universal~\cite{lloyd2018universal}, and its generalizations can capture other powerful algorithms such as the quantum singular value transformation~\cite{Lloyd2021HamQSVT}.
Although the QAOA can find the optimum when its level\footnote{In some literature, $p$ is also referred to as the QAOA depth. Here we call $p$ as the QAOA's level to avoid confusion with the quantum circuit-depth or runtime, which for the QAOA scales roughly as $p\times (\text{max graph degree})$.} (number of layers) $p$ goes to infinity~\cite{farhi2014quantum}, we have a limited knowledge of its behavior at finite $p$ due to the challenges in analyzing quantum many-body dynamics on classical computers.
Even at the lowest level $p$\,=\,1, the QAOA has output distributions that cannot be efficiently simulated on any classical device under reasonable complexity-theoretic assumptions \cite{farhi2016supreme}, similar to algorithms implemented in recent ``quantum supremacy'' experiments~\cite{Arute2019supremacy,zhong2020supremacy}.
Experimental tests of quantum optimization algorithms have largely been restricted to only modest-sized problems and short runtimes due to noisy quantum hardware limitations~\cite{harrigan2021quantum}, 
although hints of a polynomial speedup over simulated annealing are recently observed in some cases~\cite{Ebadi2022quantum}.
Nevertheless, it is difficult to conclude from experimental observations whether there is a definitive quantum advantage for approximate optimization without a convincing picture of the quantum algorithms' asymptotic behavior at large problem sizes and long runtimes.

To address this question, we are in need of rigorous theoretical studies of quantum optimization algorithms.
Recent work have taken steps in this direction by analyzing the QAOA and obtaining provable performance guarantees.
Early results look at MaxCut, where it was shown the QAOA at $p=1$ has a guaranteed worst-case approximation ratio that beats random guessing but not the best known guarantee achieved by the classical semi-definite programming (SDP) algorithm~\cite{farhi2014quantum}.
Since the known methods for proving the worst-case guarantees require a computation that scales doubly exponentially with $p$ and are thus limited to extremely small $p$~\cite{wurtz2020bounds}, later results turn to the more tractable analysis of average-case performance.
It was shown in \cite{farhi2019quantum} that the QAOA at $p=11$ outperforms the standard SDP on typical instances of the Sherrington-Kirkpatrick spin glass problem.
This result was extended to MaxCut on large random regular graphs in \cite{basso2021quantum} where the QAOA beats the best unconditionally proven performance of any known classical algorithm.
Nevertheless, pending a widely believed conjecture that these two problems exhibit no overlap gap property (OGP), an approximate message passing (AMP) algorithm~\cite{montanari2019optimization,AMS2021MaxCut} can get arbitrarily close to the optimum.

Moreover, recent theoretical results show that the QAOA's level $p$ needs to grow at least logarithmically with problem size $n$ for some COPs on graphs exhibiting locally tree-like structures~\cite{bravyi2020obstacles,farhi2020quantumRandom,farhi2020quantumWorstCase,chou2021limitations}.
The practical relevance of this limitation on the QAOA is yet to be understood, and furthermore these results do not apply to models of graphs exhibiting full connectivity.
Many classical algorithms including AMP are similarly limited, and provably fail to reach optimality if the problem exhibits OGP even when there is full connectivity.
The possibility for a quantum advantage, however, even in the constant-$p$ regime, was thus left open prior to this work.

In this paper, we analyze the power and limitation of the QAOA applied to general models of random COPs in the infinite size limit.
An example we consider is the $q$-spin model, which describes ensembles of COPs with random all-to-all $q$-body couplings.
This problem is provably difficult for many classical algorithms, including AMP~\cite{gamarnik2021overlapAukosh}, low-degree polynomials~\cite{gamarnik2020lowFOCS} and Boolean circuits~\cite{gamarnik2021circuit},  when $q\ge 4$ and is even, because it exhibits the OGP.
On the other hand, the power of quantum algorithms for this model is not known. Our first contribution is a formalism to calculate the average performance of the QAOA at any fixed level $p$ in the infinite size limit for various ensembles with i.i.d. random couplings (Theorem~\ref{thm:moments}), using a strong generalization of the multinomial theorem  (Proposition~\ref{thm:gen-multinomial}).
Although the proof of the latter result is mathematically involved, the result itself has a simple interpretation as a saddle-point approximation.
We also establish concentration properties for these problems, where we show that measurement outputs from the QAOA applied to a typical instance concentrate at the calculated average.
We note that these concentration results do not follow from general concentration bounds applicable to classical algorithms, and instead we establish the concentration property by showing vanishing variance.
Our result substantially generalizes previous analyses, which were limited to either two-body couplings ($q$\,=\,2)~\cite{farhi2019quantum}, or the lowest QAOA level $p=1$ for arbitrary $q$-body couplings~\cite{ClaesDam2021,boulebnane2021predicting}.

For our second main result, we show that the performance of the QAOA on $q$-spin models  matches asymptotically the one for Max-$q$-XORSAT on sparse random hypergraphs. This asymptotic equivalence was first shown via explicit formulas for $q=2$ in \cite{basso2021quantum,boulebnane2021predicting}, and we generalize it to arbitrary $q$ in the current paper (Theorem~\ref{thm:dense_sparse_agreement} and \ref{thm:agree-regular}). 
Our proof method for the asymptotic equivalence of QAOA performance on these models differs from the approach usually employed in the classical context.
Classical results on the equivalence of dense and sparse models, such as \cite{dembo2017extremal, sen2018optimization}, use Lindeberg's argument or its variants to establish universality properties of free energy of random Hamiltonians. 
Unfortunately, Lindeberg's argument appears powerless in our setting, and we use other methods to establish this correspondence. We leave it as an interesting challenge to develop general purpose methods establishing Lindeberg-type universality in the quantum setting.

Lastly, we show that the QAOA at any constant $p$ cannot approximate arbitrarily well the ground state values of $q$-spin models in the average-case when $q\ge 4$ and is even (Theorem~\ref{thm:obstruction}).
Previously, this limitation was only shown for some COPs on sparse hypergraphs, via arguments exploiting the OGP and the locality of the QAOA that prevents it from seeing the whole graph at sufficiently low depth~\cite{farhi2020quantumRandom,chou2021limitations}.
Importantly, for the fully connected $q$-spin models we consider, the locality-based arguments do not apply, and no limitation of this kind was known.
Instead, our result is obtained from a ``dense-from-sparse'' reduction where we use the previous obstruction-by-OGP result on sparse random hypergraphs to prove limitation of the QAOA on dense spin glass models.
Since we have shown the QAOA's performance on the two types of models match asymptotically, its ability to find near optimums for $q$-spin models would contradict its failure to reach near optimality on sparse ones.
This proof idea is unusual and novel: many results for sparse random hypergraphs are obtained by establishing them for complete graphs ($q$-spin models) and then employing the asymptotical equivalence of the two graph structures~\cite{dembo2017extremal,sen2018optimization,panchenko2018k,chen2019suboptimality}, including the OGP statement established in~\cite{chen2019suboptimality}.
However, the usage of results for sparse random graphs in order to obtain ramifications for complete graphs is new.

Our work clarifies paths forward in understanding quantum advantages in approximate optimization.
Although it was previously known that local quantum algorithms such as the QAOA are limited in the low circuit-depth regime where they do not see the whole graph, our result shows for the first time that significant barriers remain even when the whole graph is seen.
One natural path forward is to compare the energy achieved by the constant-$p$ QAOA to that by the AMP algorithm \cite{montanari2019optimization, alaoui2020optimization}, which holds the current record on the $q$-spin models among polynomial-time classical algorithms~\cite{huang2021lipschitz}.
It would be very interesting to see whether the QAOA can achieve a better energy than AMP.
Moreover, there is currently no good methods for analyzing the QAOA when the level $p$ grows faster than say, $2\log n$, even for sparse random hypergraphs.
In this regime, none of the currently known limitations apply, and it remains an open question how fast $p$ needs to grow to achieve arbitrarily good approximation.

\section{Background: the QAOA, spin glasses, and the overlap gap property}
\label{sec:review}

\paragraph{Review of the QAOA}---
The QAOA is a quantum algorithm introduced by \cite{farhi2014quantum} for finding approximate solutions to combinatorial optimization problems.
The goal is to maximize a cost function, which counts the number (or total weight) of clauses satisfied by an input bit string.
Given a cost function $C(\zv)$ on bit strings $\zv \in \{+1,-1\}^n$,  we can define a corresponding quantum operator $C$, diagonal in the computational basis, as $C\ket{\zv} = C(\zv) \ket{\zv}$.
Moreover, define the operator $B = \sum_{j=1}^n X_j$, where $X_j$ is the Pauli $X$ operator acting on qubit $j$.
Given a set of parameters $\vect\gamma = (\gamma_1,\gamma_2,\ldots,\gamma_p) \in \R^p$ and $\vect\beta = (\beta_1,\beta_2,\ldots,\beta_p) \in \R^p$, 
the QAOA initializes the system of qubits in the state $\ket{s} = 2^{-n/2}\sum_{\zv}\ket{\zv}$ and applies $p$ alternating layers of unitary operations $e^{-i\gamma_k C}$ and $e^{-i\beta_k B}$ to prepare the state
\begin{align} \label{eq:wavefunction}
\ket{\paramv} = e^{-i\beta_p B} e^{-i\gamma_p C} \cdots e^{-i\beta_1 B} e^{-i\gamma_1 C} \ket{s}.
\end{align}
For a given cost function $C$, measuring $\ket{\paramv}$ in the computational basis enough times will yield a bit string $\zv$ whose value $C(\zv)$ is near the quantum expectation $\braket{\paramv|C|\paramv}$ or better.
Heuristic strategies have been proposed to optimize $\braket{\paramv|C|\paramv}$ with respect to parameters $(\paramv)$ using a good initial guess \cite{ZhouQAOA}.

\paragraph{Defining ensembles of random COPs}---
We consider a general combinatorial optimization problem (COP) on $n$ bits where each clause involves at most $\qmax$ bits.
This problem can be understood as maximizing a cost function over $\zv \in \{\pm 1\}^n$ that takes the form
\begin{equation} \label{eq:CJ}
C_J(\zv) = \sum_{q=1}^{\qmax} c_q \sum_{i_1,\ldots,i_q=1}^n J_{i_1,i_2,\ldots,i_q} z_{i_1} z_{i_2}\cdots z_{i_q} 
= c_1\sum_{i=1}^n J_i z_i + c_2\sum_{i,j=1}^n J_{i,j} z_i z_j + \cdots
\end{equation}
where each problem instance is specified by a choice of tensors $J=\{ \{ J_{i_1, \ldots, i_q}\}_{i_1, \ldots, i_q \in [n]} \}_{q \in [q_{\max}]}$.
We study the application of the QAOA to an ensemble of random COPs that takes the above form.
For example, we consider
\begin{itemize}
\item $\Gmixed(n)$ --- mixed spin model. This ensemble is defined by any sequence of $c_q\in \mathbb{R}$ and randomly chosen $J_{i_1,\ldots,i_q} \sim_{iid} \cN(0,1/n^{q-1})$ as normal distribution, for each $q \in [q_{\max}]$ and $i_1, \ldots, i_q \in [n]$.
\item $\G_q(n)$ --- pure $q$-spin model. 
    This is a special case of the mixed spin model where $c_q = 1$ and $c_{q'} = 0$ for all $q' \neq q$. Note the Sherrington-Kirkpatrick (SK) model is the pure $2$-spin model.
\item $\G_{d,q}^\ERsub(n)$ --- Max-$q$-XORSAT on a random \ER directed multi-hypergraph  \cite{chou2021limitations}. 
    Here, a random directed multi-hypergraph on $n$ vertices is obtained by first choosing the number of edges $m \sim \Poisson(dn)$, and then choosing hyperedges $e^1, e^2, \ldots, e^m$ i.i.d. uniformly at random from the set $[n]^q$ of all vertex $q$-tuples (some hyperedges could potentially be identical). Each hyperedge associates a random weight $w(e^j) \sim_{iid} \Unif(\{ \pm 1/\sqrt{d}\})$. The cost function is defined as $C_{\ERsub, q}(\zv) = \sum_{j = 1}^m w(e^j) z_{e^j_1} \cdots z_{e^j_q}$. This model has an equivalent description in the form (\ref{eq:CJ}). More specifically, the cost function $C_{\ERsub, q}(\zv)$ (as a function) has the same distribution as $C_J(\zv) = \sum_{i_1, \ldots, i_q = 1}^n J_{i_1, \ldots, i_q} z_{i_1} \cdots z_{i_q}$, where $J_{i_1, \ldots, i_q} \sim_{iid} ( J^+_{i_1, \ldots, i_q} - J^-_{i_1, \ldots, i_q}) / \sqrt{d} $ with $J^+_{i_1, \ldots, i_q}, J^-_{i_1, \ldots, i_q} \sim_{iid} \Poisson(d/(2n^{q-1}))$.
    See Appendix~\ref{apx:ER} for an explanation.
\end{itemize}

\noindent
We denote generically by $\G(n)$ any of the distribution over $J$ for a fixed problem size $n$ specified by the ensemble description, as above.
Later we will drop the $n$ and denote $\G$ as the ensemble that describes the distribution of problems at all sizes.

These ensembles of COPs are studied as spin glasses in the statistical physics community.
We now review a few known facts about the typical-case behavior of these random COP ensembles that will set the stage for our results.

\paragraph{Ground energy density of random COPs}---
We begin with the fact that the optimization of mixed spin model admits a scaling limit of the following form: 
for any choice of coefficients $\cv=(c_1, c_2,\ldots, c_{\qmax})$, suppose that $J\sim\G_\mixed(n)$, then there exists a constant $\eta_{\rm OPT}(\Gmixed)$ which is the asymptotic optimum value of the associated optimization problem in the sense  
\begin{equation} \label{eq:eta-OPT}
     \lim_{n\to\infty} \frac{1}{n} \max_{\zv\in \{\pm 1\}^n}C_J(\zv)=\eta_{\rm OPT}(\Gmixed),~~~~ \text{in probability. }
\end{equation}
The existence of the limit \eqref{eq:eta-OPT} follows from a simple and clever subadditivity argument of Guerra and Toninelli~\cite{GuerraTon}, and the actual value of this limit was obtained as a result of a very impressive development starting from a non-rigorous physics-style argument by Parisi~\cite{parisi1980sequence}, and then rigorously verified by Talagrand~\cite{talagrand2006parisi}. See also Panchenko~\cite{panchenko2013sherrington} for a book reference for this and related results. 

For the special case of pure $q$-spin model $J \sim \G_q(n)$ (recall that this corresponds to $c_q=1$ and $c_{q'}=0$ for all $q'\ne q$ in $\G_\mixed$), we denote the right hand side of \eqref{eq:eta-OPT} as $\eta_\OPT(\G_q)$. It is known that $\eta_\OPT(\G_q)$ also describes the ground energy density of random sparse model $\G^{\ERsub}_{d,q}$ in the large degree limit: 
\begin{equation} \label{eq:GS-dense=sparse}
    \lim_{d\to\infty} \lim_{n\to\infty} \frac{1}{n}\EV_{J\sim \G^{\ERsub}_{d,q}(n)} \max_{\zv \in \{\pm1\}^n} C_J(\zv) 
    = \eta_\OPT(\G_q).
\end{equation}
This result was established first in~\cite{dembo2017extremal} for the case of graphs (that is $q=2$), and then extended to the case of \ER hypergraphs in~\cite{sen2018optimization, chen2019suboptimality}. 
While the results were restricted to certain types of distributions of $J$, the proof approach developed in~\cite{sen2018optimization}, which is based on the Lindeberg-type argument, reveals a universality property: the limit depends on the distribution only through the first and second moment, and furthermore applies to the setting of $J$ with non-zero mean upon centering.

\paragraph{The overlap gap property and algorithmic thresholds}---
While the above results give us a statistical prediction of the typical-case optimal energy density $\eta_\OPT$ for these random COPs, they are non-constructive and thus do not yield efficient algorithms to find near-optimal solutions $\zv$ such that $C_J(\zv)/n \approx \eta_\OPT$.
As it turns out, there is a provable obstacle preventing many algorithms to reach optimality in certain ensembles of problems, which is described as a property in the solution space geometry of the problem.
This is the overlap gap property (OGP), which roughly says that for certain choices of the disorder $J$, specifically in the case
of the pure $q$-spin model with $q\ge 4$ even, there is a gap in the set of possible pairwise overlaps of near-optimal solutions.

The use of OGP to show obstruction for quantum algorithms, specifically the QAOA, was initiated in \cite{farhi2020quantumRandom} and subsequently extended  in \cite{chou2021limitations}.
Both work prove limitation of local quantum algorithms when the COPs can be embedded on a sparse hypergraph.
We now describe the result in~\cite{chou2021limitations} formally. This result forms the basis of our negative result on the
limitation of the performance of the QAOA in the setting of the fully connected spin glass models.

\begin{theoremCLSS}[Modified version of \cite{chou2021limitations} Corollary 4.4]
\label{thm:CLSS}
Let $J \sim \G^\ERsub_{d,q}(n)$. For every  even $q\ge 4$, there exists $\eta_{\rm OGP}(\G_q)<\eta_{\rm OPT}(\G_q)$ and a sequence $\{\delta(d)\}_{d\ge 1}$
with the following property.
For every $\epsilon>0$ there exist sufficiently large $d_0$ such that 
for every $d\ge d_0$, every 
 $p\le \delta(d) \log n$ and an arbitrary choice of parameters $\paramv$, with probability converging to $1$ as $n \to \infty$, the performance of the QAOA  at level $p$ 
satisfies $\braket{\paramv|C_J/n|\paramv} \le \eta_{\rm OGP}(\G_q)+\epsilon$. 
\end{theoremCLSS}
The proof idea first introduced in~\cite{farhi2020quantumRandom} and then extended to the hypergraph setting in~\cite{chou2021limitations}
uses the effective locality of the QAOA at logarithmic depth (or level) which prevents it from overcoming the OGP barrier. 
The implementation of this idea in the context of classical algorithm was introduced in~\cite{gamarnik2014limits} and then 
extended to a broad class of other classical algorithms; see~\cite{gamarnik2021overlap} for a survey. 

When considering COPs embedded on dense hypergraphs, however, such as the case of the model $\Gmixed$ or $\G_q$,  
the techniques in~\cite{farhi2020quantumRandom} and~\cite{chou2021limitations} do not apply
since they rely crucially on the locality of the algorithm and the sparsity of the hypergraph, so that the whole graph is not seen by 
the QAOA at sufficiently low level.
In contrast, for these dense models the QAOA sees the whole graph at any level.
Thus, a new idea is needed to prove the obstruction in this non-local setting.
This is one of the main goals of  this paper, and it is 
achieved by a ``dense-from-sparse'' argument which uses an asymptotic equivalence of the algorithm's performance on the dense $\G_q$ and the sparse $\G^\ERsub_{d,q}$ models.
It is noteworthy that while there is a large literature on using the ``sparse-from-dense'' reduction
for the purposes of establishing the results on sparse graphs from known results on dense graphs, 
such as~\cite{dembo2017extremal,sen2018optimization, panchenko2018k,chen2019suboptimality}, the converse direction we undertake in this paper is novel.

\section{Main results}
\label{sec:results}

\subsection{Performance of the QAOA on random COPs}

Our first main result is a constructive method to evaluate the QAOA's performance applied to a generic ensemble of random COPs in the form of Eq.~\eqref{eq:CJ} that satisfies the following assumption:
\begin{assumption}
\label{assum:iidJ}
For every $q \in [\qmax]$, $J_{i_1,\ldots,i_q}$ are i.i.d. following some mean zero symmetric distribution with finite second moment. Assume that $\EV[e^{i \lambda J_{1,\ldots,q}}]$ is real positive for large $n$ and denote $g_{q,n}(\lambda) = n^{q-1} \log \EV[e^{i \lambda J_{1,\ldots,q}}]$. For any fixed $\lambda$, we assume that (1) $\lim_{n \to \infty} g_{q, n}''(\lambda) / n = 0$; (2) $\lim_{n\to \infty} g_{q,n}(\lambda) \equiv g_q(\lambda)$ exists and $g_q$ is differentiable; (3) $\lim_{n \to \infty} g_{q,n}'(\lambda) = g_q'(\lambda)$.
\end{assumption}
Note the ensembles $\G_\mixed$, $\G_q$, $\G_{d,q}^\ERsub$ defined earlier satisfy the assumption above.
In the theorem below, we establish the limiting performance of QAOA of at any constant level $p$ for any ensemble satisfying Assumption \ref{assum:iidJ}.
The full proof of this theorem is given in Appendix \ref{sec:QAOA-derivation}.
\begin{theorem} \label{thm:moments}
Suppose $C_J$ is a random COP of form \eqref{eq:CJ} drawn from an ensemble $\G$ that satisfies Assumption~\ref{assum:iidJ}.
Then for any $p$ and any parameters $(\paramv) \in \R^{2p}$, we have
\begin{equation}\label{eqn:first-moment-limit-thm}
\lim_{n\to\infty}\EV_{J\sim \G(n)}\Big[\braket{\paramv|C_J/n|\paramv}\Big] = V_p(\G,\paramv)
\end{equation}
and
\begin{equation}\label{eqn:second-moment-limit-thm}
\lim_{n\to\infty}\EV_{J\sim \G(n)}\Big[\braket{\paramv|(C_J/n)^2|\paramv}\Big] = \big[V_p(\G,\paramv)\big]^2,
\end{equation}
where the limit $V_p(\G,\paramv)$ has a formula that we define explicitly in Eq.~\eqref{eq:Vp-formula}.
\end{theorem}

A corollary of the above theorem is that the value produced by the QAOA satisfy concentration properties.
Specifically,  it means that with probability tending to 1 as $n\to\infty$, measurements of the QAOA applied to a typical instance of the random COP ensemble will yield a string $\zv$ whose value $C_J(\zv)/n$ concentrate at the quantum expectation $\braket{\paramv|C_J/n|\paramv}$, which itself is close to the ensemble-averaged value $V_p(\paramv)$.
This follows from the fact that the second moment is equal to the first moment squared in the $n\to\infty$ limit. 
To see this, note that
\begin{equation}
\EV_J[\braket{\paramv|(C_J/n)^2|\paramv}] - \EV_J^2[\braket{\paramv|C_J/n|\paramv}] = {\rm Var}({\rm instance}) + {\rm Var}({\rm measurement})
\end{equation}
is the combined variance over instances and measurements, where
\begin{align}
{\rm Var}({\rm instance}) &= \EV_J[\braket{\paramv|C_J/n|\paramv}^2] - \EV_J^2[\braket{\paramv|C_J/n|\paramv}], \\
{\rm Var}({\rm measurement}) &=  \EV_J\hspace{-2pt}\Big[\braket{\paramv|(C_J/n)^2|\paramv} - \braket{\paramv|C_J/n|\paramv}^2\Big].
\end{align}
Since both are non-negative, Theorem~\ref{thm:moments} implies both ${\rm Var}({\rm instance})$ and  ${\rm Var}({\rm measurement})\to0$ as $n\to\infty$.
By Chebyshev's inequality, this means the QAOA is concentrated over both instances and measurements.

\paragraph{Formula for $V_p(\G, \paramv)$}---
We now describe our formula of $V_p$. First, we denote 
\begin{equation}\label{eq:set_A_def}
A := \big\{(a_1, a_2,\ldots, a_p, a_{-p}, \ldots, a_{-2}, a_{-1}) : a_{\pm j} \in \{ \pm 1 \}, \forall 1\le j \le p\big\}
\end{equation}
as the set of $2p$-bit strings.
Given QAOA parameters $(\paramv) \in \R^{2p}$, we define for any $\av\in A$
\begin{eqnarray}
Q_{{\av}} &:=& 
{\textstyle
    \prod_{r=1}^p
	(\cos\beta_r)^{1 + ({a}_r + {a}_{-r})/2}
	(\sin\beta_r)^{1 -({a}_r + {a}_{-r})/2}
	(i)^{({a}_{-r} - {a}_r) /2} },
	\label{eq:Qdef} \\
\Phi_\av &:=& 
{\textstyle\sum_{r=1}^p \gamma_r \big(a_r a_{r+1} \cdots a_p ~-~ a_{-p} \cdots a_{-r-1} a_{-r} \big)}. \label{eq:Phi_def} 
\end{eqnarray}
We also denote $\av\bv \in A$ as the bit-wise product of $\av, \bv \in A$.
Then we define $\{ W_\av\}_{\av \in A}$ to be the unique solution to the following self-consistent equation:
\begin{equation}\label{eq:Ws_self_consistent}
W_\av = Q_\av \exp\Big[\sum_{q=1}^\qmax q \sum_{\bv_1,\ldots,\bv_{q-1} \in A} g_{q}\big(c_q \Phi_{\av \bv_1\cdots\bv_{q-1}}\big) W_{\bv_1} \cdots W_{\bv_{q-1}} \Big], \qquad \forall \av \in A, 
\end{equation}
where $g_q(\lambda) = \lim_{n \to \infty} n^{q-1}\log \EV_{J\sim\G(n)}[e^{i J_{1,2,\ldots,q} \lambda}]$.
We will establish the existence and uniqueness of solution to the above equation in Lemma~\ref{lem:unique_solution_SCE}.
In general, the solution $\{W_\av\}$ can be obtained sequentially in some order of the $4^p$ elements of $A$, using an $O(4^{p\qmax})$-time iterative procedure (see Lemma~\ref{lem:efficiency-4pqmax}); it can also be as efficient as $O(p^2 4^p)$ (see Theorem~\ref{thm:agree-regular}).
Finally, $V_p$ is defined by
\begin{equation} \label{eq:Vp-formula}
V_p(\G, \paramv) = {-}\sum_{q=1}^\qmax ic_q \sum_{\av_1,\ldots,\av_q \in A} g_{q}'(c_q\Phi_{\av_1\cdots \av_q}) W_{\av_1} \cdots W_{\av_q}.
\end{equation}

\begin{remark}
The performance of QAOA for the SK model at any constant level $p$ has been previously derived in \cite{farhi2019quantum}. Subsequently, the QAOA's performance at level $p=1$  has been derived for the mixed spin model in \cite{ClaesDam2021} and sparse random hypergraph model in \cite{boulebnane2021predicting} respectively.
Our Theorem~\ref{thm:moments} encompasses these results by deriving the performance for any such ensemble of COPs with i.i.d. couplings at any level $p$.
In particular it can be verified that, when $c_2 = 1/\sqrt{2}$, $c_q = 0$ for all $q \neq 2$, and $g_{2} = - x^2 /2$, our formula (\ref{eq:Vp-formula}) coincides with the formula of the QAOA's performance for the SK model in \cite{farhi2019quantum}.
Our result also coincides with the ones in \cite{ClaesDam2021,boulebnane2021predicting} for various models at $p=1$.
\end{remark}

\subsection{Equivalence of the performance of QAOA on dense and sparse graphs}

As Eq.~\eqref{eq:GS-dense=sparse} states, 
the global optimum of the dense model $\G_q$ is asymptotically identical to the global optimum of the sparse model $\G_{d,q}^\ERsub$.
This was established using Lindeberg's universality type arguments in \cite{dembo2017extremal, sen2018optimization}.
It is conceivable that the performance of the QAOA will also be similar on these two models of disorder.
We establish that a more general version of this universality is true at constant levels in the following theorem.%

\begin{theorem}[Universality]\label{thm:dense_sparse_agreement}
Let $\G_{d, q}(n)$ be a generic ensemble of COPs with only $q$-body couplings satisfying Assumption~\ref{assum:iidJ}, with characteristic function $g_{q,n}^{(d)}(\lambda) = n^{q-1} \log \EV[e^{i \lambda J^{(d)}_{1,2,\ldots, q}}]$ for $J^{(d)} \sim \G_{d,q}(n)$. Moreover, suppose that
\begin{align}\label{eq:universality_assumption}
    \lim_{d \to \infty} \lim_{n \to \infty} \Big( g_{q,n}^{(d)}(\lambda), g_{q,n}^{(d)\prime}(\lambda)\Big) = \Big({-}\frac{\lambda^2}{2}, - \lambda\Big), \qquad \forall \lambda \in \R. 
\end{align}
Then, the asymptotic performance of the level-$p$ QAOA on $\G_{d, q}$ is the same as on $\G_q$ (the pure $q$-spin model):
\begin{align}
    V_p(\G_q, \paramv) = \lim_{d\to \infty}V_p(\G_{d, q},\paramv),
\end{align}
at any parameters $(\paramv)$. In particular, this applies to the sparse model $\G_{d,q}^\ERsub$, i.e.,
\begin{equation}\label{eqn:Vp_dense_sparse}
    V_p(\G_q, \paramv) = \lim_{d\to\infty}V_p(\G_{d,q}^\ERsub,\paramv).
\end{equation}
\end{theorem}

We remark that this identity is by explicit computation via the formula \eqref{eq:Vp-formula} above (see proof in Appendix~\ref{apx:ER}). Unfortunately, the Lindeberg type argument appears powerless in this setting. 

Moreover, the performance of the QAOA for Max-$q$-XORSAT on large-girth $d$-regular hypergraphs was derived in \cite{basso2021quantum} using an approach more direct than our Theorem~\ref{thm:moments}. There, the authors computed the expected performance of QAOA which was shown to be identical among all $d$-regular $q$-uniform hypergraphs with girth $>2p+1$.
They gave an explicit formula for the performance and denoted it as $\nu_p^{[q]}(d,\paramv)$.
It was also shown in \cite{basso2021quantum} that the formula of $\lim_{d\to\infty}\nu_p^{[2]}(d,\paramv)$ for large-girth regular graphs matches $V_p(\G_2,\paramv)$, the analogous formula for the SK model.
We generalize this correspondence to arbitrary $q$ in the following theorem, which shows that the QAOA's performance for the $q$-spin model $\G_q$ is also equivalent to that for Max-$q$-XORSAT on any large girth $d$-regular hypergraphs in the $d\to\infty$ limit. 

\begin{theorem}
\label{thm:agree-regular}
Let $\nu_p^{[q]}(d, \paramv)$ be the performance of QAOA on any instance of Max-$q$-XORSAT on any $d$-regular $q$-uniform hypergraph with girth $>2p+1$ given in \cite{basso2021quantum}.
Then for any $p$ and any parameters $(\paramv)$, we have
\begin{equation} \label{eq:reg-eq}
    V_p(\G_q, \paramv)  = \sqrt{2} \lim_{d\to\infty} \nu_p^{[q]}(d,\sqrt{q}\paramv).
\end{equation}
\end{theorem}

The proof is given in Appendix~\ref{apx:regular}.
We note that \cite{basso2021quantum} provides a more succinct formula for $\lim_{d\to\infty} \nu_p^{[q]}(d,\paramv)$, and has evaluated it up to $p\le20$ with an $O(p^2 4^p)$-time iteration on a classical computer.
By the equality \eqref{eq:reg-eq}, this also gives a faster procedure to evaluate $V_p(\G_q,\paramv)$ than the $O(4^{p \qmax})$-time procedure for the more generic case described in Theorem~\ref{thm:moments}. 

\begin{remark}
Although $d$-regular hypergraphs are similar to \ER hypergraphs, our Theorem \ref{thm:moments} does not apply to $d$-regular hypergraphs since we need independence structure of the tensor $J$ in the cost function $C_J$ (c.f. Eq. (\ref{eq:CJ})). On the other hand, the technique in \cite{basso2021quantum} is algebraic and is specifically for $d$-regular hypergraphs, and their technique does not apply to $q$-spin models and \ER hypergraphs.
\end{remark}

\subsection{Limitation of the QAOA on dense hypergraphs}

We now turn to our last main result, where we show that the QAOA's performance is obstructed even in a regime when the whole graph is seen.
This is in contrast to all known proven limitations of the QAOA that apply to sparse graphs when the QAOA does not see the whole graph \cite{bravyi2020obstacles, farhi2020quantumRandom, farhi2020quantumWorstCase, chou2021limitations}.
Recall the value $\eta_{\OGP}(\G_q) < \eta_{\OPT}(\G_q)$ from Theorem~\hyperref[thm:CLSS]{CLSS21}.

\begin{theorem}
\label{thm:obstruction}
For any fixed $p$, parameters $(\paramv)$,
and any even $q \ge 4$, we have
\begin{align}
V_p(\G_q,\paramv) = \lim_{n\to\infty} \EV_{J\sim \G_q(n)} [\braket{\paramv|C_J/n|\paramv}]
\le \eta_\OGP(\G_q).
\end{align}
This implies that constant-$p$ QAOA is not able to find a near-global optimizer of the q-spin model when $q \ge 4$ and is even.
\end{theorem}

The proof of this obstruction theorem exploits the equivalence of the QAOA's performance on dense and sparse hypergraphs established in Theorem~\ref{thm:dense_sparse_agreement} above, together with Theorem \hyperref[thm:CLSS]{CLSS21} established in \cite{chou2021limitations}. 
See Section~\ref{sec:proof_sketch} where we give a short proof.

While the theorem statement here is for the ensemble average, we remark that it also applies to typical instances due to the concentration property implied by Theorem~\ref{thm:moments}.
Furthermore, note the constant-$p$ QAOA for the $q$-spin model $\G_q$ has a circuit-depth that grows polynomially with the graph size $n$, and the entire graph is seen by the algorithm at any level $p$.
This is in sharp contrast to the QAOA applied to sparse models such as $\G_{d,q}^\ERsub$.

\section{Technical overview}

\subsection{A generalized multinomial theorem motivated by the QAOA}

We now explain the key technical idea behind this paper, where we provide a mathematical framework to study the performance of the QAOA for a general ensemble of random COPs. 
The goal is to evaluate the quantum expectation of the operator $C_J$ which yields the average value produced by the algorithm.
Using techniques introduced in \cite{farhi2019quantum}, we insert complete sets of $Z$-basis states between unitary operations in Eq.~\eqref{eq:wavefunction} tracking the path of every qubit, and write this expectation explicitly for any ensemble $\G$ satisfying Assumption~\ref{assum:iidJ} as the following sum over paths:
\begin{align}
&\EV_{J\sim\G(n)} \hspace{-2pt}
\Big[\braket{\paramv| \frac{C_J}{n} |\paramv}\Big]
= \sum_{\{n_\av\}} \binom{n}{\{n_\av\}} \prod_{\av\in A } Q_\av^{n_\av}
\exp\Big[
		n \sum_{q=1}^\qmax \sum_{\av_1,\ldots,\av_q \in A} g_{q,n}\big(c_q \Phi_{\av_1\cdots\av_q}\big)\frac{n_{\av_1} \cdots n_{\av_q}}{n^q}
	\Big] \nonumber \\
&\qquad \qquad\qquad\qquad\qquad\qquad\qquad\qquad
	\times\Big(-\sum_{q=1}^\qmax ic_q \sum_{\bv_1,\ldots,\bv_q \in A} g_{q,n}'(c_q\Phi_{\bv_1\cdots \bv_q})\frac{n_{\bv_1} \cdots n_{\bv_q}}{n^q}\Big),
	\label{eq:Cn-explicit}
\end{align}
where $A = \{\pm 1\}^{2p}$, $Q_\av$ and $\Phi_\av$ are defined earlier in Eqs.~\eqref{eq:Qdef} and \eqref{eq:Phi_def},
and  the sum is over all sets of non-negative integers $\{n_\av: \av\in A\}$ that add up to $n$.
Here, each $n_\av$ counts the number of qubits whose path matches a given bit string $\av$.
(See Lemma \ref{lem:characteristic_function_derivative} in Appendix~\ref{sec:QAOA-derivation} for the precise statement.)
We may also consider higher powers $(C_J/n)^k$ to obtain concentration properties of the algorithm, but we will focus on $k=1$ here to explain the essentials.

While the above expression can be evaluated explicitly by summing over all $O(n^{|A|})=O(n^{4^p})$ terms, this double exponential scaling quickly becomes intractable even for the modest case of $p=2$.
The fact that we have a polynomial of $\{n_\av\}$ inside the exponential in Eq.~\eqref{eq:Cn-explicit} also prevents us from applying the multinomial theorem when its degree $\qmax>1$.
In the $n\to\infty$ limit, one may be tempted to treat the $\{n_\av\}$ as random variables from a multinomial distribution with $\{Q_\av\}$ as probabilities, so that we can approximate the sum as a Gaussian integral and then apply Laplace's method.
However, since the $Q_\av$'s are generally complex numbers, this approach does not apply.

To overcome this difficulty, our main technical contribution in this paper is a generalized multinomial theorem that enables evaluation of sums like the one in Eq.~\eqref{eq:Cn-explicit} in the $n\to\infty$ limit. We state the informal theorem here and defer its formal version to Proposition \hyperref[thm:gen-multinomial-restate]{4.1 (Formal)} in Appendix~\ref{apx:gen-multinomial}, where we also provide the full proof.
\begin{proposition}[Informal]
\label{thm:gen-multinomial}
Suppose we are given a finite set $A$ and a set of complex numbers $\{Q_\av\}_{\av\in A}$ where $\sum_{\av\in A} Q_\av = 1$.
Also suppose $A$ has a subset $A_0$ such that $\{ Q_\av \}_{\av \in A_0} \subseteq [0, 1]$, and the remaining elements in $A \setminus A_0$ can be decomposed into pairs of $(\av,\bar\av)$ such that $Q_\av+Q_{\bar\av}=0$.
Then for any sequence of bounded-degree polynomials $f_n(\{\omega_\av\}_{\av \in A})$ and ``well-played'' (defined later in Definition~\ref{def:well-played-polynomial-redefine}) polynomials $P_n(\{\omega_\av\}_{\av \in A})$ with $\lim_{n\to\infty} (f_n,P_n) = (f, P)$, we have
\begin{equation} \label{eq:limiting-value}
\lim_{n\to\infty} 
\sum_{\{n_\av\}} \binom{n}{\{n_\av\}}
\Big(\prod_{\bv\in A} Q_\bv^{n_\bv} \Big)
	\exp\Big[n P_n(\{n_\av/n\}) \Big]
f_n(\{n_\av/n\})
= f(\{W_\av\}) \,,
\end{equation}
where $\{W_\av\}_{\av \in A}$ is given as the unique solution to
\begin{equation} \label{eq:Ws-in-gen-multi}
W_\av = Q_\av \exp\bigg[ \frac{\partial P(\{ W_\bv\}_{\bv \in A})}{\partial W_\av} \bigg],~~~ \forall \av \in A.
\end{equation}
\end{proposition}
Although motivated by the desire to analyze the QAOA, this generalized multinomial theorem may be of independent interest for other endeavors.
The proof is rather cumbersome, but the result is surprisingly consistent with the answer obtained from a simple but non-rigorous application of the saddle-point method.
Specifically, in the $n\to\infty$ limit, one may define continuous variables $\omega_\av := n_\av/n$, and approximate the sum over paths in Eq.~\eqref{eq:limiting-value} as an integral:
\begin{equation} \label{eq:saddle}
    \int  \Big({\textstyle \prod_{\av\in A} d{\omega_\av}} \Big) e^{n S(\{\omega_\av\})} f(\{\omega_\av\}) \,,
\end{equation}
where $S(\{\omega_\av\}) = -\sum_{\av\in A} \omega_\av \log (\omega_\av / Q_\av) + P(\{\omega_\av\})$.
Then Eq.~\eqref{eq:limiting-value} may be understood as a saddle-point approximation of the above integral as $n\to\infty$, where it is dominated by the saddle point of $S(\{\omega_\av\})$ subject to the constraint that $\sum_\av \omega_\av = 1$. This saddle point turns out to be the unique solution to Eq.~\eqref{eq:Ws-in-gen-multi}.
See Appendix~\ref{sec:nonrigorous} for more details.

However, it is challenging to make this saddle-point approximation rigorous directly.
Instead, we prove Proposition~\ref{thm:gen-multinomial} by making use of the combinatorial structure that emerges in the summation when the coefficients of $P_n$ satisfy a property that we call ``well-played.''
This ``well-played'' property manifests after pairing up the variables $(n_\av, n_{\bar\av})$ associated with the cancelling pairs of complex numbers $(Q_\av,Q_{\bar\av})$, and then transforming the polynomial $P_n(\{n_\av/n\})$ into a ``canonical representation'' of the dual variables $t_\av = n_\av + n_{\bar\av}$ and $d_\av = n_\av - n_{\bar\av}$.
In this canonical representation, we find that the limit as $n\to\infty$ exists if all the terms of $P_n$ are at least linear in $t_\av$, along with some additional constraints.
This property also enables Eq.~\eqref{eq:Ws-in-gen-multi} to be exactly solved with an iterative procedure and allows for an explicit evaluation of the limiting value \eqref{eq:limiting-value}.

\subsection{Proof (sketches) of main theorems}\label{sec:proof_sketch}

\paragraph{Proof sketch of Theorem~\ref{thm:moments} (the QAOA's performance on i.i.d. ensembles)}---
We start by deriving Eq.~\eqref{eq:Cn-explicit} using similar techniques as in \cite{farhi2019quantum}.
In order to evaluate the more general cases considered in this paper, we show that the polynomial in the exponential of \eqref{eq:Cn-explicit} satisfies the ``well-played'' property whenever $g_{q,n}(\lambda)$ is an even function.
Then applying Proposition~\ref{thm:gen-multinomial}, we get the expected performance of the QAOA as in Theorem~\ref{thm:moments}.
The second moment is obtained similarly.
In order for the proofs to go through easily, we require some simple technical conditions on $g_{q,n}$ (that is, on the distribution of $J$), as stated in Assumption~\ref{assum:iidJ}.

\paragraph{Proof sketch of Theorem~\ref{thm:dense_sparse_agreement} and Theorem~\ref{thm:agree-regular} (dense-sparse correspondence)}---
We apply the formula $V_p(\G,\paramv)$ given in Theorem~\ref{thm:moments} to the pure $q$-spin model $\G_q$ and a generic ensemble $\G_{d,q}$ satisfying the stated assumptions.
In particular, we show that the aforementioned \ER ensemble $\G_{d,q}^\ERsub$ can be transformed to an equivalent description in the form of $\G_{d,q}$ using the Poisson splitting trick (c.f. Lemma \ref{lem:Poisson-splitting}).
The $V_p$ formulas for these ensembles are then shown to match asymptotically, yielding Theorem~\ref{thm:dense_sparse_agreement}.

With the formula for $V_p(\G_q,\paramv)$ in hand, Theorem~\ref{thm:agree-regular} is then straightforwardly proved by algebraically transforming the formula for $\lim_{d\to\infty} \nu_p^{[q]}(d,\paramv)$ given in~\cite{basso2021quantum} using the notations of this paper for all $q$.
We explicitly show the two formulas match, similar to a proof in \cite{basso2021quantum} which had obtained the analogous result at $q=2$.

\paragraph{Proof of Theorem~\ref{thm:obstruction} (limitation of the QAOA on $\G_q$)}---
Here we give the short but complete proof of Theorem \ref{thm:obstruction}, which is easily implied by Theorem \ref{thm:dense_sparse_agreement} and Theorem \hyperref[thm:CLSS]{CLSS21}. 
Indeed, for any fixed $p \in \Z_{> 0}$ and fixed $(\paramv)$, Theorem \hyperref[thm:CLSS]{CLSS21} implies that
\begin{equation}\label{eqn:OGP-obstruction-ER}
\lim_{d\to\infty}V_p(\G_{d,q}^\ERsub,\paramv) = \lim_{d \to \infty} \lim_{n \to \infty} \E_{J \sim \G_{d,q}^\ERsub(n)}[\langle \paramv | C_J / n | \paramv \rangle] \le \eta_{\OGP}(\G_q). 
\end{equation}
In the equation above, we use the fact that  $\langle \paramv | C_J / n | \paramv \rangle$ on $\G_{d,q}^\ERsub$
concentrates around its expectation which is implied by Theorem \ref{thm:moments}, and thus the high probability bound in Theorem \hyperref[thm:CLSS]{CLSS21} extends to the expectation bound. 

Furthermore, Eq.~\eqref{eqn:Vp_dense_sparse} in Theorem \ref{thm:dense_sparse_agreement} gives
\[
V_p(\G_{q},\paramv) = \lim_{d\to\infty}V_p(\G_{d,q}^\ERsub,\paramv). 
\]
Combining these two equations the proof is completed.

\section{Discussion and outlook}

In this paper we have considered the performance of the QAOA for the problem of finding a near ground state of spin glass models when the algorithm is applied at a level (number of layers) that does not grow with problem size.
We have derived an analytic formula of the value produced by the quantum algorithm as a function of its parameters in the limit as the number of spins diverges to infinity. 
Using this formula we have established that this value is asymptotically the same for the pure $q$-spin model and for Max-$q$-XORSAT on a sparse random hypergraph model.
This extends recent results for the case of 2-spin models at any level~\cite{farhi2019quantum, basso2021quantum, boulebnane2021predicting} and for the case of $q$-spin models at level $1$~\cite{ClaesDam2021,boulebnane2021predicting}.
We have also established a concentration result showing that this value is concentrated around the instance-independent average with high probability as the system size diverges to infinity. 

Using this correspondence, we prove that the value produced by the QAOA is bounded away from optimality by a multiplicative constant for the case of $q$-spin models with $q\ge 4$ and even.
This is obtained as a corollary of a recent result~\cite{chou2021limitations} that the value of the QAOA is bounded away from optimality when the algorithm is implemented on sparse random hypergraphs.
The latter result relied on locality of the algorithm and was restricted to sparse hypergraphs, much like prior negative results \cite{bravyi2020obstacles,farhi2020quantumRandom,farhi2020quantumWorstCase} for the QAOA in regimes where it does not see the whole graph.
In this paper, we extend the limitation to the $q$-spin models, where the QAOA sees the whole graph at any level.

Our proof approach for this limitation uses a novel idea of ``dense-from-sparse'' reduction.
While many results in the past have used the ``sparse-from-dense'' reduction where properties of sparse random hypergraphs are established from the corresponding properties of the $q$-spin model, the reversed direction implemented in this paper is new.

There is a large scope of problems which remain open.
Our method of proof for the concentration
result, which follows \cite{farhi2019quantum}, is rather unconventional and is based on explicitly computing the 
second moment of the value produced by the algorithm.
This contrasts sharply
with approaches in classical settings where concentration bounds
follow rather directly by application of standard techniques such as McDiarmid's or Azuma's inequalities.
We note these concentration inequalities give stronger (exponential) convergence than what can be obtained from our explicit calculation.
The quantum setting considered in this paper prevents the implementation 
of the more standard methods, and in general the concentration properties in quantum
systems represent a general scope of rather interesting open problems. 

Similarly, the ``dense-from-sparse'' reduction in our paper is obtained from a rather bulky explicit calculation of the asymptotic performance of the QAOA. In the classical settings such equivalence results follow from a broader universality type argument based on Lindeberg's approach. The direct application of Lindeberg's argument to the quantum setting appears to fail, and finding a workable quantum counterpart for such universality argument is an interesting open problem. 

It is surprising to us that the result of our complicated calculation can be understood simply as a saddle-point approximation.
The latter is a tool commonly used in physics calculations, often non-rigorously, e.g., in Parisi's formula of the SK model~\cite{parisi1980sequence}.
Nevertheless, rigorous verification of the saddle-point approximation's predictions can sometimes require indirect and sophisticated methods, e.g., in Talagrand's proof of the Parisi formula~\cite{talagrand2006parisi}.
Here, our generalized multinomial theorem serves as an indirect proof that the saddle-point approximation gives correct predictions of the QAOA's behavior for many spin glass models.
Following the appearance of this work, \cite{Boulebnane2022KSAT} is able to directly apply the saddle-point method to analyze the QAOA for random $k$-SAT in certain regimes.
It would be interesting to understand more broadly when the saddle-point method can be applied to yield simple and accurate analysis of quantum algorithms and many-body dynamics.

Although we have proven a limitation of the QAOA at any constant level $p$, our work still leaves open a few possibilities of a quantum advantage in this regime.
For example, it would be very interesting to compare the constant-$p$ QAOA's performance on the $q$-spin models to the state-of-the-art classical algorithm which is the AMP algorithm~\cite{montanari2019optimization,alaoui2020optimization}.
This algorithm provably finds $(1-\epsilon)$-approximate optimums when there is no OGP (conjectured for $q=2$) after $p_{\rm AMP}(\epsilon)$ number of iterations for any $\epsilon>0$.
Here $p_{\rm AMP}(\epsilon)$ is a function independent of problem size.
Nevertheless, AMP faces an algorithmic threshold bounded away from optimality when $q\ge 4$ is even~\cite{gamarnik2021overlapAukosh}, the setting where OGP is known to exist.
The maximum value achievable by AMP can be obtained numerically via an extended Parisi formula.
Whether the QAOA can match or possibly even beat the performance of the AMP algorithm remains an interesting
open challenge.

Another interesting challenge regards improving our analysis and obtaining explicit numerical values achieved by the constant-$p$ QAOA at large $p$.
Presently, we only know explicit values for the $q$-spin models up to $p\le 20$ from \cite{basso2021quantum} due to the $O(p^2 4^p)$-complexity of evaluating the current formula. 
Going beyond and obtaining these values at higher $p$ can shed light on the challenge of comparing the performance of the QAOA with the performance of the AMP algorithm, mentioned earlier.

Recently,  it was shown in~\cite{huang2021lipschitz}  that no algorithms satisfying an ``overlap concentration property'' can obtain a value better than AMP on the mixed  $q$-spin models.
This was done using  a variant of the overlap gap property, called the branching-OGP.
It would be interesting to see if this limitation extends to the QAOA at constant levels.
This would imply in particular that the QAOA at constant levels does not surpass the value achieved by the AMP algorithm.

Our proof method is
limited to the QAOA with a constant level $p$. It is of interest to extend 
it to the QAOA with $p$ that grows with problem size $n$. 
At the current stage we don't have 
the techniques to approach this.
Since the QAOA provably reaches optimality when no bound on $p$ is placed, it is in particular important to understand whether this can be achieved at $p$ which is only polynomially large, so that the QAOA remains within the class of polynomial-time algorithms.
This would provide a definitive evidence of a quantum advantage in optimization.

\section*{Acknowledgments}
We thank Sergio Boixo, Edward Farhi, Sam Gutmann, Jarrod R. McClean, and Benjamin Villalonga for helpful comments. 
D.G. is supported in part by NSF grant DMS-2015517. S.M. is supported in part by NSF grant DMS-2210827.

\clearpage

{
\tableofcontents
}

\clearpage
\appendix

\section{Some notations and conventions}
Before delving into the formal derivation and proof of our results in the appendices that follow, we first establish some notations and conventions that we use throughout the paper.

For any integer $n\ge1$, we denote $[n]=\{1,2,3,\ldots,n\}$. We denote $\Z$, $\Z_{\ge 0}$, $\R$, and $\C$ to be the set of integers, non-negative integers, real numbers, and complex numbers, respectively. For a set or tuple $S$, we denote $| S |$ as the cardinality (number of elements) in $S$. 

We denote
\begin{align}
    \{O_\av\}_{\av \in A} \equiv \{O_\av : \av \in A\}
\end{align}
to be a set of elements indexed by $A$. We will sometimes also write $\{O_\av\}$ if it is clear which set $\av$ is in. For a set of non-negative integers $\{n_\av\}_{\av \in A} \subseteq \Z_{\ge 0}$ such that $\sum_{\av\in A} n_\av = n$, we denote the multinomial coefficient as
\begin{equation}
    \binom{n}{\{n_\av\}} = \frac{n!}{\prod_{\av\in A} n_\av!}.
\end{equation}

Sometimes the set $\{n_\av\}_{\av \in A}$ is taken as an argument into a function. This argument is to be understood as an ordered tuple (or a vector of $n_\av$'s), i.e.,
\begin{align}
f(\{n_\av\}_{\av\in A}) = f(n_1, n_2, \ldots, n_{|A|}).
\end{align}

For any finite set $S$, we define the {\it Kleene star} $S^*$ of $S$ as the set of all possible ordered tuples (strings or words) formed from elements of $S$, i.e.
\begin{align}
    S^* = \bigcup_{k=0}^\infty S^k = S^0 \cup S \cup S^2 \cup S^3 \cup \cdots
\end{align}
where $S^k$ is the set of vectors of length $k$ with elements in $S$, and $S^0 = \{ \emptyset \}$ contains the empty string.
For a string $\avu \in S^*$, we denote $| \avu |$ to be the length of the string. We denote $\sum_{\av\in \avu} f(\av)$ to mean that $\sum_{l = 1}^k f(\av_l)$ if the string gives $\avu=(\av_1,\av_2,\ldots,\av_k)$ (the notations $\prod_{\av \in \avu}$ and $\max_{\av \in \avu}$ have similar interpretation).

For any two sets $A$ and $B$, we let $A \cup B$ be the union of $A$ and $B$. We use notation $C = A \sqcup B$ to denote that $C$ is the disjoint union of $A$ and $B$. That is, $C = A \cup B$ and $A$ and $B$ are disjoint. 

Finally, we denote $\log x = \ln x$ in this paper.

\section{Proof of the generalized multinomial theorem (Proposition \ref{thm:gen-multinomial})}
\label{apx:gen-multinomial}

\subsection{Formal statement of Proposition \ref{thm:gen-multinomial}}

We will give a formal statement of Proposition \ref{thm:gen-multinomial} in this section. We first give some definitions that will be helpful for the formal statement of Proposition \ref{thm:gen-multinomial}. 

\begin{defn}[Proper set $A$]\label{def:set_A}
We say a finite set $A$ is a \emph{proper set} if it is the disjoint union of three finite sets $A_0$, $D$, and $\overline D$ with distinct elements (i.e., $A = A_0 \sqcup D \sqcup \barD$). Moreover, $D$ and $\barD$ have the same number of elements, and are equipped with a one-to-one mapping $\iota$ from $D$ to $\barD$. Finally, elements in $D \cup \barD$ are equipped with a bar operation: for $\av \in D$, we denote $\bar \av \equiv \iota(\av) \in \barD$; for $\av \in \barD$, we denote $\bar \av \equiv \iota^{-1}(\av) \in D$. 
\end{defn}

\begin{defn}[Proper complex numbers $\{ Q_\av \}_{\av \in A}$]\label{def:complex_Q}
Let $A = A_0 \sqcup D \sqcup \barD$ be a proper set. We say a set of complex numbers $\{ Q_\av \}_{\av \in A}$ indexed by elements in $A$ is \emph{proper} if $\{ Q_\av \}_{\av \in A_0}$ are real non-negative with $\sum_{\av \in A_0} Q_\av = 1$, and for any $\av \in D$, we have $Q_{\bar \av} = -Q_{\av}$. 
\end{defn}

\begin{defn}[Natural and canonical representation of functions of $\{ \omega_\av \}_{\av \in A}$]\label{def:canonical_natural_representation}
Let $A = A_0 \sqcup D \sqcup \barD$ be a proper set. Let $f(\{ \omega_\av \}_{\av \in A})$ be a complex function over complex variables $\{ \omega_\av \}_{\av \in A}$. Let $g$ be a complex function $g(\{ \tau_\av \}_{\av \in D}, \{ \eta_\bv \}_{\bv \in D}, \{ \nu_\cv \}_{\cv \in A_0})$ over complex variables $\{ \tau_\av \}_{\av \in D}$, $\{ \eta_\bv \}_{\bv \in D}$, and $\{ \nu_\cv \}_{\cv \in A_0}$. We say $g$ is a \emph{canonical representation} of $f$, and $f$ is a \emph{natural representation} of $g$, if for any complex variables $\{ \omega_\av \}_{\av \in A} \subseteq \C$, we have
\[
f(\{ \omega_\av \}_{\av \in A}) = g (\{\omega_\av + \omega_{\bar{\av}}\}_{\av \in D}, \{\omega_{\bv} - \omega_{\bar{\bv}} \}_{\bv \in D}, \{ \omega_\cv\}_{\cv\in A_0}). 
\]
By this equation, we can define the natural representation of any complex function $g$ over complex variables $\{ \tau_\av \}_{\av \in D}$, $\{ \eta_\bv \}_{\bv \in D}$, and $\{ \nu_\cv \}_{\cv \in A_0}$. Furthermore, this is a linear change of variable whose Jacobian is non-singular, so that any complex function $f$ over complex variables $\{ \omega_\av \}_{\av \in A}$ also has a canonical representation $g$ defined as
\[
g (\{\tau_\av\}_{\av \in D}, \{ \eta_\bv \}_{\bv \in D}, \{\nu_\cv \}_{\cv\in A_0}) = f(\{ (\tau_\av + \eta_\av) / 2 \}_{\av \in D} \cup \{ (\tau_\bv - \eta_\bv) / 2 \}_{\bv \in \barD} \cup \{ \nu_\cv \}_{\cv\in A_0}). 
\]
We use the operator $\cC$ to denote the transformation from natural representation to canonical representation. That is, we write $g = \cC[f]$ and $f = \cC^{-1}[g]$ if $g$ is the canonical representation and $f$ is the natural representation. 
\end{defn}
Intuitively, we think of the two representations as related via the following basis transformation:
\begin{alignat}{2}
    \tau_\av &= \omega_\av + \omega_{\bar\av}, \quad 
    &&\quad\forall (\av, \bar\av) \in D\times \barD,
    \nonumber\\
    \eta_\bv &= \omega_\bv - \omega_{\bar\bv},
    \quad
    &&\quad\forall (\bv, \bar\bv) \in D\times \barD,
    \\
    \nu_\cv &= \omega_\cv,
    \quad
    &&\quad\forall\cv \in A_0.
    \nonumber
\end{alignat}

\begin{defn}[Sequence of converging polynomials with uniformly bounded degree]\label{def:converging-polynomial}
Let $A$ be a finite set. Consider a sequence of complex polynomials $\{ f_n \}_{n \ge 1}$ over complex variables $\{ \omega_\av\}_{\av \in A}$ of form 
\begin{equation}
f_n(\{ \omega_\av \}_{\av \in A} ) = \sum_{\avu \in A^*} F_{\avu, n} \prod_{\av \in \avu} \omega_\av. 
\end{equation}
We say that $\{ f_n \}_{n \ge 1}$ is a \emph{sequence of converging polynomials with uniformly bounded degree}, if (1) there exists an integer $d_{\max}$ such that the degree of each $f_n$ is bounded by $d_{\max}$; (2) there exists a polynomial $f$ with degree bounded by $d_{\max}$ and of form 
\begin{equation}
f(\{ \omega_\av \}_{\av \in A} ) = \sum_{\avu \in A^*} F_{\avu} \prod_{\av \in \avu} \omega_\av,
\end{equation}
such that $\lim_{n \to \infty} F_{\avu, n} = F_{\avu}$ for any $\avu \in A^*$. We say $f$ is the limit of $\{ f_n\}_{n \ge 1}$. 
\end{defn}

\begin{defn}[Well-played polynomials]\label{def:well-played-polynomial-redefine}
Let $A = A_0 \sqcup D \sqcup \barD$ be a proper set (c.f. Definition \ref{def:set_A}) and assume that there is an ordering $\succ$ on the set $D$. Consider a complex polynomial $P$ over complex variables $\{ \omega_\av\}_{\av \in A}$ whose canonical representation (c.f. Definition \ref{def:canonical_natural_representation}) is given by
\begin{equation}\label{eq:poly_t_d_n_def_redefinition}
\cC[P](\{\tau_\av\}_{\av \in D}, \{\eta_\bv \}_{\bv \in D}, \{\nu_\cv\}_{\cv\in A_0}) 
    = \sum_{\avu \in D^*, \bvu \in D^*, \cvu \in A_0^*} \Psi_{\avu,\bvu, \cvu}
            \prod_{\av \in \avu} \tau_\av
            \prod_{\bv \in \bvu} \eta_\bv
            \prod_{\cv \in \cvu} \nu_\cv .
\end{equation}
We say $P$ is \emph{well-played} if for some constant $L \ge 1$, we have
\begin{equation}\label{eqn:Psi-well-played-condition}
    \Psi_{\avu, \bvu, \cvu} \neq 0 
    \qquad \textnormal{only if} \qquad
    1 \le |\avu| \le L, \quad
    0 \le |\bvu|, |\cvu| \le L, 
    \quad\textnormal{ and }\quad
    \max (\avu) \succ \max(\bvu).
\end{equation}
That is, its coefficient is nonzero only if the words $\avu,\bvu,\cvu$ have bounded lengths (i.e. the polynomial has bounded degree), and the word $\avu$ has length at least 1 and contains an element strictly larger than anything in the word $\bvu$. 

Furthermore, consider a sequence of general polynomials $\{ P_n \}_{n \ge 1}$ over complex variables $\{ \omega_\av\}_{\av \in A}$ whose canonical representation given as
\begin{equation}\label{eq:poly_t_d_n_def_sequence}
\cC[P_n](\{\tau_\av\}_{\av \in D}, \{\eta_\bv\}_{\bv \in D}, \{\nu_\cv\}_{\cv\in A_0}) 
    = \sum_{\avu \in D^*, \bvu \in D^*, \cvu \in A_0^*} \Psi_{\avu,\bvu, \cvu, n}
            \prod_{\av \in \avu} \tau_\av
            \prod_{\bv \in \bvu} \eta_\bv
            \prod_{\cv \in \cvu} \nu_\cv. 
\end{equation}
We say $\{ P_n \}_{n \ge 1}$ is a \emph{sequence of converging well-played polynomials with uniformly bounded degree} if (1) each $P_n$ is a well-played polynomial; (2) $\{ P_n\}_{n \ge 1}$ is a sequence of converging polynomials with uniformly bounded degree (c.f. Definition \ref{def:converging-polynomial}).
This implies that there exists a well-played polynomial $P$ of the form \eqref{eq:poly_t_d_n_def_redefinition} such that $\lim_{n \to \infty} \Psi_{\avu,\bvu, \cvu, n} = \Psi_{\avu, \bvu, \cvu}$ for every $\avu \in D^*, \bvu \in D^*, \cvu \in A_0^*$. 
\end{defn}

The following lemma shows that two particular self-consistent equations (SCEs) related to a well-played polynomial have unique solutions. The solutions to these two self-consistent equations are the same. This unique solution will be used in the generalized multinomial theorem that follows.

\begin{lemma}[Existence and uniqueness of SCE solution]\label{lem:unique_solution_SCE}
Let $A$ be a proper finite set (c.f. Definition \ref{def:set_A}) and $\{ Q_\av \}_{\av \in A} \subseteq \C$ be a set of proper complex numbers (c.f. Definition \ref{def:complex_Q}). Let $P$ be a well-played polynomial (c.f. Definition \ref{def:well-played-polynomial-redefine}). We have the following: 
\begin{itemize}
\item[(a)] Consider the following self-consistent equation upon variables $\{ W_\xv \}_{\xv \in D}$: 
\begin{align}\label{eqn:SCE_canonical}
    W_{\xv} =  Q_{\xv} \exp \Big[ \partial_{\tau_\xv} \cC[P](\{\tau_\av = 0\}_{\av \in D}, \{\eta_\bv = 2 W_\bv\}_{\bv \in D},\{\nu_\cv = Q_\cv\}_{\cv \in A_0}) \Big], ~~~ \forall \xv \in D.
\end{align}
This self-consistent equation has a unique solution. 
\item[(b)] Consider the following self-consistent equation upon variables $\{ W_\xv \}_{\xv \in A}$: 
\begin{equation}
\begin{aligned}\label{eqn:SCE_natural}
&~W_{\xv} =  Q_\xv \exp \Big[ \partial_{\omega_\xv} P( \{ \omega_\av =  W_\av \}_{\av \in A}) \Big], ~~~ &&\forall \xv \in A.\\
\end{aligned}
\end{equation}
This self-consistent equation has a unique solution. The solution satisfies $W_\av = Q_\av$ when $\av \in A_0$ and $W_{\bar \xv} + W_{\xv} = 0$ when $\xv \in D$. 
\item[(c)] Let $\{ \tilde W_\xv \}_{\xv \in D}$ be the unique solution of (\ref{eqn:SCE_canonical}), and $\{ \overline W_\xv \}_{\xv \in A}$ be the unique solution of (\ref{eqn:SCE_natural}). We have $\tilde W_\av = \overline W_\av$ for $\av \in D$.
\item[(d)] The self-consistent equation (\ref{eqn:SCE_canonical}) can be solved in $O(| A |^{d_{\max} + 1})$ time complexity on a classical computer, where $d_{\max}$ is the maximum degree of polynomial $P$. 
\end{itemize}
\end{lemma}

The proof of Lemma~\ref{lem:unique_solution_SCE} is given in Section~\ref{sec:unique_solution_SCE}. Now we are ready to state the formal version of Proposition \ref{thm:gen-multinomial}, the generalized multinomial theorem: 

\begin{proposition1}
\label{thm:gen-multinomial-restate}
Let $A$ be a proper finite set (c.f. Definition \ref{def:set_A}) and $\{ Q_\av \}_{\av \in A} \subseteq \C$ be a set of proper complex numbers (c.f. Definition \ref{def:complex_Q}). Then for a sequence of converging polynomials $f_n(\{\omega_\av\}_{\av \in A})$ with uniformly bounded-degree and with limit $f$ (c.f. Definition \ref{def:converging-polynomial}) and any sequence of converging well-played polynomials $P_n(\{\omega_\av\}_{\av \in A})$ with uniformly bounded degree and with limit $P$ (c.f. Definition \ref{def:well-played-polynomial-redefine}), we have
\begin{align}\label{eqn:gen-multi-restate-equation}
\lim_{n\to\infty} 
\sum_{\substack{\{ n_\av \ge 0 \}_{\av \in A}\\ \sum_{\av \in A} n_\av = n}} \binom{n}{\{n_\av\}_{\av \in A}}
\Big(\prod_{\bv\in A} Q_\bv^{n_\bv} \Big)
	\exp\Big[n \cdot P_n(\{n_\av/n\}_{\av \in A}) \Big]
f_n(\{n_\av/n\}_{\av \in A})
= f(\{W_\av\}_{\av \in A})
\end{align}
where $\{W_\av\}_{\av \in A}$ is given as the unique solution to Eq. \eqref{eqn:SCE_natural}.
\end{proposition1}

\subsection{Proof of Lemma \ref{lem:unique_solution_SCE}}
\label{sec:unique_solution_SCE}

\noindent
{\bf Proof of (a). }
By the canonical representation of polynomial $P$ as in Eq. (\ref{eq:poly_t_d_n_def_redefinition}), we have
\[
\begin{aligned}
&~\partial_{\tau_\xv} \cC[P](\{\tau_\av = 0\}_{\av \in D}, \{\eta_\bv = 2W_\bv\}_{\bv \in D},\{\nu_\cv = Q_\cv\}_{\cv \in A_0}) \Big] \\
=&~ \partial_{\tau_\xv} \Big[\sum_{\avu,\bvu \in D^*, \cvu \in A_0^*}
    \Psi_{\avu,\bvu, \cvu}
            \prod_{\av \in \avu} \tau_\av
            \prod_{\bv \in \bvu} \eta_\bv
            \prod_{\cv \in \cvu} \nu_\cv \Big] \Big |_{\{\tau_\av = 0\}_{\av \in D}, \{\eta_\bv = 2W_\bv\}_{\bv \in D},\{\nu_\cv = Q_\cv\}_{\cv \in A_0}},\\
            =&~\sum_{\bvu \in D^*, \cvu \in A_0^*} \Psi_{\xv, \bvu, \cvu} \prod_{\bv \in \bvu} (2W_\bv)
            \prod_{\cv \in \cvu} Q_\cv,
\end{aligned}
\]
By the fact that $P$ is a well-played polynomial, we have $\Psi_{\xv, \bvu, \cvu}$ is non-zero only if $\max\{ \bvu \} \prec \xv$. As a consequence, Eq. (\ref{eqn:SCE_canonical}) can be rewritten as 
\begin{align}\label{eqn:SCE_canonical_in_proof}
    W_{\xv} =  Q_{\xv} \exp \Big[ \sum_{\bvu\in D^*, \cvu\in A_0^*, \max(\bvu) \prec \xv} \Psi_{\xv, \bvu, \cvu} \prod_{\bv \in \bvu} (2W_\bv)
            \prod_{\cv \in \cvu} Q_\cv \Big], ~~~~ \forall \xv \in D. 
\end{align}
Note that the right hand side of the equation above depends on $\{ W_\bv \}_{\bv \in D}$ only through $\{ W_\bv \}_{\bv \prec \xv}$. This implies that, for any $\xv \in D$, if the values of $\{ W_\bv \}_{\bv \prec \xv}$ are determined, we can use Eq. (\ref{eqn:SCE_canonical_in_proof}) (equivalent to Eq. (\ref{eqn:SCE_canonical})) to determine $W_\xv$. So Eq. (\ref{eqn:SCE_canonical}) is essentially a recursive equation that can determine $\{ W_\bv \}_{\bv \in D}$ sequentially in the ascending order of $\prec$. This implies that there exists a unique solution to Eq. (\ref{eqn:SCE_canonical}).

\vspace{8pt}
\noindent
{\bf Proof of (b): solution is anti-symmetric. } First we show that Eq. (\ref{eqn:SCE_natural}) implies that $W_\av + W_{\bar \av} = 0$ for all $\av \in D$. Using the canonical representation, we write explicitly the natural representation of polynomial $P$ as
\begin{equation}
    P(\{\omega_\av\}_{\av \in A}) = \sum_{\avu,\bvu\in D^*,\cvu \in A_0^* } \Psi_{\avu,\bvu,\cvu} \prod_{\av \in\avu} (\omega_\av + \omega_{\bar\av}) \prod_{\bv \in \bvu} (\omega_\bv - \omega_{\bar\bv}) \prod_{\cv \in \cvu} \omega_\cv .
\end{equation}
Since $P$ is well-played, here $\Psi_{\avu,\bvu,\cvu} \neq 0$ only if $|\avu|\ge 1$ and $\max(\avu) \succ \max(\bvu)$.

We prove $W_{\av} + W_{\bar \av} = 0$ for any $\av \in D$ using an induction argument. Note that for any $\xv \in D$, by the symmetric property of $P$, we have
\begin{equation}
\begin{aligned}
\partial_{\omega_\xv} P(\{W_\av\}_{\av \in A}) =&~ P_{\xv, (1)}(\{ W_\av\}) + P_{\xv, (2)}(\{ W_\av\}),\\
\partial_{\omega_{\bar \xv}} P(\{W_\av\}_{\av \in A}) =&~ P_{\xv, (1)}(\{ W_\av\}) - P_{\xv, (2)}(\{ W_\av\}),
\end{aligned}
\end{equation}
where
\begin{equation}
\begin{aligned}
P_{\xv, (1)}(\{ W_\av\}) = &~ \sum_{\avu,\bvu\in D^*,\cvu \in A_0^* } \Psi_{\avu,\bvu,\cvu} \cdot \partial_{W_\xv}  \Big[\prod_{\av \in\avu} (W_\av + W_{\bar\av})\Big] \prod_{\bv \in \bvu} (W_\bv - W_{\bar\bv})  \prod_{\cv \in \cvu} W_\cv, \\
P_{\xv, (2)}(\{ W_\av\}) = &~\sum_{\avu,\bvu\in D^*,\cvu \in A_0^* } \Psi_{\avu,\bvu,\cvu}  \prod_{\av \in\avu} (W_\av + W_{\bar\av})\cdot \partial_{W_\xv} \Big[\prod_{\bv \in \bvu} (W_\bv - W_{\bar\bv}) \Big]  \prod_{\cv \in \cvu} W_\cv.
\end{aligned}
\end{equation}
First, let $\xv$ be the largest element (under order $\succ$) in $D$. By the well-played property of $P$, we have $\max\{ \bvu \} \prec \xv$ for every nonzero $\Psi_{\avu,\bvu,\cvu}$, which implies that $\partial_{W_\xv} [\prod_{\bv \in \bvu} (W_\bv - W_{\bar\bv}) ] = 0$ for every nonzero $\Psi_{\avu,\bvu,\cvu}$. This implies that $P_{\xv, (2)}(\{ W_\av\}) = 0$ so that $\partial_{\omega_\xv} P(\{W_\av\}) = \partial_{\omega_{\bar \xv}} P(\{W_\av\})$. Furthermore, by the property that $Q_\xv + Q_{\bar \xv} = 0$, using Eq. (\ref{eqn:SCE_natural}) we have $W_\xv + W_{\bar \xv} = 0$. 

Next, we assume the induction hypothesis that $\xv \in D$ is such that $W_{\av} + W_{\bar \av} = 0$ for every $\av \succ \xv$, and we will show that $W_\xv + W_{\bar \xv} = 0$. When the quantity $\partial_{W_\xv} [\prod_{\bv \in \bvu} (W_\bv - W_{\bar\bv}) ] \neq 0$, we know that $\xv \in \bvu$. So by the well-played property of $P$, for any non-zero $\Psi_{\avu,\bvu,\cvu}\times \partial_{W_\xv} [\prod_{\bv \in \bvu} (W_\bv - W_{\bar\bv}) ]$, we must have $\max\{ \avu\} \succ \xv$. As a consequence, by the inductive hypothesis, for any non-zero $\Psi_{\avu,\bvu,\cvu} \times \partial_{W_\xv} [\prod_{\bv \in \bvu} (W_\bv - W_{\bar\bv}) ]$, we must have $\prod_{\av \in\avu} (W_\av + W_{\bar\av}) = 0$. This implies that $P_{\xv, (2)}(\{ W_\av\}) = 0$ so that $\partial_{\omega_\xv} P(\{W_\av\}) = \partial_{\omega_{\bar \xv}} P(\{W_\av\})$. Furthermore, by the property that $Q_\xv + Q_{\bar \xv} = 0$, using Eq. (\ref{eqn:SCE_natural}) we have $W_\xv + W_{\bar \xv} = 0$. This proves the induction conclusion and finishes the induction argument.  

\noindent
{\bf Uniqueness part of (b). } First we consider Eq. (\ref{eqn:SCE_natural}) when $\xv\in A_0$. Then $\partial_\xv P(\{\omega_\av=W_\av\}) = 0$ because for every term in $P$ above in which $\Psi_{\avu,\bvu,\cvu} \neq 0$, $\avu$ is always non-empty, so we get a factor $ \prod_{\av\in\avu}(W_\av + W_{\bar\av}) = 0$ after plugging in $\{\omega_\av=W_\av\}$.
Then Eq. (\ref{eqn:SCE_natural}) implies that
\begin{equation}
    W_\xv = Q_\xv 
    \qquad \forall \xv\in A_0.
\end{equation}

Then we consider Eq. (\ref{eqn:SCE_natural}) when $\xv \in D$. 
Due to the constraint that $W_\av + W_{\bar\av}=0$, the only non-zero terms in $P$ after taking the derivative $\partial_{\omega_\xv}$ and plugging in $\{\omega_\av=W_\av\}$ are the ones where $\avu = \xv$. This gives
\begin{equation} \label{eq:dPdWx}
\partial_{\omega_\xv} P(\{W_\av\}) 
= \sum_{\bvu\in D^*, \cvu\in A_0^*} \Psi_{\xv,\bvu,\cvu}  \prod_{\bv \in \bvu} (W_\bv - W_{\bar\bv}) \prod_{\cv\in \cvu} W_\cv
= \sum_{\bvu\in D^*, \cvu\in A_0^*} \Psi_{\xv,\bvu,\cvu}  \prod_{\bv \in \bvu} (2W_\bv) \prod_{\cv\in \cvu} Q_\cv. 
\end{equation}
Note since $\Psi_{\xv,\bvu,\cvu}\neq0$ only if $\xv \succ \max(\bvu)$, the above expression only contains dependence of $\{W_\bv: \bv \prec \xv\}$.
So Eq. (\ref{eqn:SCE_natural}) implies that
\begin{equation}\label{eqn:recursive-in-proof}
W_\xv = Q_\xv \exp\Big[
\sum_{\bvu\in D^*, \cvu\in A_0^*, \max(\bvu) \prec \xv} \Psi_{\xv,\bvu,\cvu}  \prod_{\bv \in \bvu} (2W_\bv) \prod_{\cv\in \cvu} Q_\cv
    \Big],
    ~~~~ \forall \xv\in D.
\end{equation}
Manifestly, this can be solved to yield a unique solution for all $W_\xv$ in the increasing order of $\xv\in D$. This implies that, if Eq. (\ref{eqn:SCE_natural}) has a solution, then the solution is unique. 

\noindent
{\bf Existence part of (b). } It can be easily checked that $\{ W_\av \}_{\av \in A} = \{ Q_\cv \}_{\cv \in A_0} \cup \{ W_\xv \}_{\xv \in D} \cup \{ W_{\bar \xv} = - W_\xv \}_{\bar \xv \in \overline D}$, where $\{ W_\xv \}_{\xv \in D}$ is the solution of Eq. (\ref{eqn:recursive-in-proof}), satisfies Eq. (\ref{eqn:SCE_natural}). This shows that Eq. (\ref{eqn:SCE_natural}) has a solution. 

\vspace{8pt}

\noindent
{\bf Proof of (c). } Note that the solution $\{ \overline W_\xv \}_{\xv \in A}$ of Eq. (\ref{eqn:SCE_natural}) satisfies Eq. (\ref{eqn:recursive-in-proof}), and the solution $\{ \tilde W_\xv \}_{\xv \in D}$ of Eq. (\ref{eqn:SCE_canonical}) satisfies Eq. (\ref{eqn:SCE_canonical_in_proof}). Note that Eq. (\ref{eqn:recursive-in-proof}) coincides with Eq. (\ref{eqn:SCE_canonical}) which has a unique solution. This implies that $\tilde W_\xv = \overline W_\xv$, for all $\xv \in D$.  

\vspace{8pt}
\noindent
{\bf Proof of (d). } The algorithm to solve $\{ W_\xv\}_{\xv \in D}$ is by sequentially use Eq. (\ref{eqn:SCE_canonical_in_proof}) in the ascending order of $\prec$. Note that $ \sum_{\bvu\in D^*, \cvu\in A_0^*, \max(\bvu) \prec \xv} \Psi_{\xv, \bvu, \cvu} \prod_{\bv \in \bvu} (2W_\bv) \prod_{\cv \in \cvu} Q_\cv$ is the partial derivative of a polynomial of $| A |$ variables and of degree $d_{\max}$, so evaluating it has at most time complexity $O(| A |^{d_{\max}})$.
Since there are $O(|D|)=O(|A|)$ equations in Eq. \eqref{eqn:SCE_canonical_in_proof}, the total time complexity is $O(| A | \times | A |^{d_{\max}}) = O(| A |^{d_{\max} + 1})$. This completes the proof.

\subsection{Proposition \hyperref[thm:gen-multinomial-restate]{4.1 (Formal)} in canonical form (Statement of Proposition \ref{prop:well-played-again})}

In this section of the appendix, we state an equivalent version of Proposition \hyperref[thm:gen-multinomial-restate]{4.1 (Formal)}, which uses the canonical representation of polynomials.
This new version of the proposition allows us to exploit the combinatorial structure of well-played polynomials more directly for proofs.

Define the operator $\sqint_\bv$ (which depends on a non-negative integer $t_\bv \in \Z_{\ge 0}$ and the complex number $Q_\bv \in \C$) for any $\bv\in D$ acting on any function $f(\eta_\bv)$ as
\begin{align}\label{eqn:little_sum-definition}
\sqint_{\bv}f(d_\bv / n) :=
    \sum_{n_\bv, n_{\bar \bv} \ge 0, n_\bv + n_{\bar \bv} = t_\bv} \binom{t_\bv}{n_\bv, n_{\bar \bv}} Q_\bv^{n_\bv} (-Q_\bv)^{n_{\bar \bv}} f((n_\bv - n_{\bar \bv}) / n). 
\end{align}
An important observation is that the operators $\{ \sqint_{\bv} \}_{\bv \in D}$ are commutative: for different $\bv_1, \bv_2 \in D$, we have $\sqint_{\bv_1} \sqint_{\bv_2} f(d_{\bv_1} / n, d_{\bv_2} / n ) = \sqint_{\bv_2} \sqint_{\bv_1} f(d_{\bv_1} / n, d_{\bv_2} / n )$ for any function $f(\eta_{\bv_1}, \eta_{\bv_2})$. 

\vspace{2pt}
Now we give the alternative statement of Proposition \hyperref[thm:gen-multinomial-restate]{4.1 (Formal)} in canonical form:

\begin{proposition}
\label{prop:well-played-again}
Let $A$ be a proper set (c.f. Definition \ref{def:set_A}) and $\{ Q_\av \}_{\av \in A} \subseteq \C$ be a set of proper complex numbers (c.f. Definition \ref{def:complex_Q}). Then for any sequence of converging uniformly well-played polynomials $P_n(\{\omega_\av\}_{\av \in A})$ with limit $P$ (c.f. Definition \ref{def:well-played-polynomial-redefine}) and any fixed non-negative integers $\{r_\av\}_{\av \in D}, \{ s_\bv\}_{\bv \in D}, \{ m_\cv\}_{\cv \in A_0}$, we define 
\begin{equation}\label{eqn:In-definition}
\begin{aligned}
&~I_n(\{r_\av\}, \{ s_\bv\}, \{ m_\cv\}) \\
:=&~ \sum_{\{ n_\cv, t_\av \} \subseteq \Z_{\ge 0}, \sum_{\cv \in A_0} n_\cv + \sum_{\av \in D} t_\av = n} \binom{n}{\{n_\cv, t_\av\}_{\av \in D, \cv \in A_0}}
    \Big(\prod_{\cv 
    \in A_0 } Q_\cv^{n_\cv} \Big) \Big( \prod_{\av\in D}
     \sqint_\av
    \Big) \\
    &~ \times \Big\{ \exp\Big[ n\cdot \cC[P_n](\{t_\av/n\}_{\av\in D}, \{d_\bv/n\}_{\bv\in D}, \{n_\cv/n\}_{\cv\in A_0})\Big]
    \times \prod_{\av, \bv \in D, \cv\in A_0}
    \Big(\frac{t_\av }{n}\Big)^{r_\av}
    \Big(\frac{d_\bv }{n}\Big)^{s_\bv}
    \Big(\frac{n_\cv }{n}\Big)^{m_\cv} \Big\}.
\end{aligned}
\end{equation}
We further define 
\begin{equation}\label{eqn:I-definition}
I(\{r_\av\}, \{ s_\bv\}, \{ m_\cv\}) := \prod_{\av, \bv \in D, \cv\in A_0} 1(r_\av = 0) (2 W_{\bv})^{s_\bv}Q_\cv^{m_\cv}, 
\end{equation}
where $\{ W_\bv \}_{\bv \in D}$ is given as the unique solution to the following equation:
\begin{align}\label{eqn:def-W-in-proposition}
    W_{\xv} =  Q_{\xv} \exp \Big[ \partial_{\tau_\xv} \cC[P](\{\tau_\av = 0\}_{\av \in D}, \{\eta_\bv = 2W_\bv\}_{\bv \in D},\{\nu_\cv = Q_\cv\}_{\cv \in A_0}) \Big]. 
\end{align}
Then for any fixed integers $\{r_\av\}_{\av \in D}, \{ s_\bv\}_{\bv \in D}, \{ m_\cv\}_{\cv \in A_0}$, we have
\begin{align}\label{eq:well-played-outcome}
    \lim_{n\to\infty}I_n(\{r_\av\}, \{ s_\bv\}, \{ m_\cv\}) = I(\{r_\av\}, \{ s_\bv\}, \{ m_\cv\}). 
\end{align}
\end{proposition}

We will show that Proposition \ref{prop:well-played-again} implies Proposition \hyperref[thm:gen-multinomial-restate]{4.1 (Formal)} in the section that follows.
The proof of Proposition~\ref{prop:well-played-again} is given later in Section~\ref{sec:proof-well-played}.

\subsection{Proof of Proposition 4.1 (Formal) using Proposition \ref{prop:well-played-again}}

The fact that Proposition~\hyperref[thm:gen-multinomial-restate]{4.1 (Formal)} follows from Proposition~\ref{prop:well-played-again} can be understood rather straightforwardly via a transformation of variables.
Here we give the full detailed proof.

\vspace{10pt}

\noindent
{\bf Step 1. } We first show that the prelimit of the left hand side of Eq. (\ref{eqn:gen-multi-restate-equation}) coincides with $I_n(\{r_\av\}, \{ s_\bv\}, \{ m_\cv\})$ when taking $f_n = \Upsilon$ with
\begin{equation}\label{eqn:function-f-choice}
\Upsilon(\{\omega_\av\}_{\av \in A}; \{r_\av\}, \{ s_\bv\}, \{ m_\cv\}) = \prod_{\av, \bv \in D, \cv\in A_0}
    (\omega_\av + \omega_{\bar \av})^{r_\av}
    (\omega_\bv - \omega_{\bar \bv})^{s_\bv}
    \omega_\cv^{m_\cv}. 
\end{equation}
That is,
\begin{gather}
~\sum_{\substack{\{ n_\av \ge 0 \}_{\av \in A}, \sum_{\av \in A} n_\av = n}} \binom{n}{\{n_\av\}_{\av \in A}}
\Big(\prod_{\bv\in A} Q_\bv^{n_\bv} \Big)
	\exp\Big[n \cdot P_n(\{n_\av/n\}_{\av \in A}) \Big]
\Upsilon(\{n_\av/n\}_{\av \in A}) \nonumber \\
=~ I_n(\{r_\av\}, \{ s_\bv\}, \{ m_\cv\}). 
\label{eqn:I_n_E_choice}
\end{gather}
Note we have suppressed the dependence on $(\{r_\av\},\{s_\bv\}, \{m_\cv\})$ in $\Upsilon(\cdot)$ for simplicity of notation.
    
To prove Eq.~\eqref{eqn:I_n_E_choice}, we note for any function $g(\{ \omega_\av \}_{\av \in A})$, we have
{\small
\begin{equation}\label{eqn:reorganizing_sum_n_t_d}
\begin{aligned}
&~\sum_{\substack{\{ n_\av \ge 0 \}_{\av \in A}\\ \sum_{\av \in A} n_\av = n}} \binom{n}{\{n_\av\}_{\av \in A}}
\Big(\prod_{\bv\in A} Q_\bv^{n_\bv} \Big) g(\{ n_\av / n \}_{\av \in A})\\
=&~ \sum_{\substack{\{ n_\cv, t_\av \}_{\av \in D, \cv \in A_0} \subseteq \Z_{\ge 0}\\ \sum_{\cv \in A_0} n_\cv + \sum_{\av \in D} t_\av = n}} \binom{n}{\{n_\cv, t_\av\}} \Big[\prod_{\av \in A} 
\sum_{\substack{n_\av, n_{\bar \av} \ge 0 \\ n_\av + n_{\bar \av } = t_\av}} {t_\av \choose n_\av, n_{\bar \av}} \Big]
\Big(\prod_{\bv\in A} Q_\bv^{n_\bv} \Big)  g(\{ n_\av / n \}_{\av \in A})\\
=&~ \sum_{\substack{\{ n_\cv, t_\av \}_{\av \in D, \cv \in A_0} \subseteq \Z_{\ge 0}\\ \sum_{\cv \in A_0} n_\cv + \sum_{\av \in D} t_\av = n}} \binom{n}{\{n_\cv, t_\av\}} \Big(\prod_{\cv\in A_0} Q_\cv^{n_\cv} \Big)  \Big[ \prod_{\bv \in A} 
\sum_{\substack{n_\bv, n_{\bar \bv} \ge 0 \\ n_\bv + n_{\bar \bv } = t_\bv}} {t_\bv \choose n_\bv, n_{\bar \bv}} 
Q_\bv^{n_\bv} Q_{\bar \bv}^{n_{\bar \bv}} \Big] \\
&~\qquad\qquad\qquad\qquad\qquad\qquad \times \cC[g](\{ t_\av / n \}_{\av \in D}, \{ (n_\bv - n_{\bar \bv})/n\}_{\bv \in D}, \{ n_\cv / n \}_{\cv \in A_0}) \\
=&~ \sum_{\substack{\{ n_\cv, t_\av \}_{\av \in D, \cv \in A_0} \subseteq \Z_{\ge 0}\\ \sum_{\cv \in A_0} n_\cv + \sum_{\av \in D} t_\av = n}} \binom{n}{\{n_\cv, t_\av\}} \Big(\prod_{\cv\in A_0} Q_\cv^{n_\cv} \Big) \Big( \prod_{\bv \in A} \sqint_\bv \Big)   \cC[g](\{ t_\av / n \}_{\av \in D}, \{ d_\bv /n\}_{\bv \in D}, \{ n_\cv / n \}_{\cv \in A_0}). 
\end{aligned}
\end{equation}
}
where we used the definition of $\sqint_\bv$ as in Eq. \eqref{eqn:little_sum-definition}.

Now we choose $g(\{ \omega_\av\}_{\av \in A}) = \exp\{ n P_n(\{ \omega_\av\}_{\av \in A}) \} \Upsilon(\{ \omega_\av \}_{\av \in A})$. Note that the canonical form of $g$ gives
\begin{equation}
\begin{aligned}
&~\cC[g] = \cC\Big[\exp[n \cdot P_n(\{\omega_\av\}_{\av \in A}) ] 
\Upsilon(\{\omega_\av\}_{\av \in A}) \Big]\\
=&~ \exp[ n\cdot \cC[P_n](\{\tau_\av\}_{\av\in D}, \{\eta_\bv\}_{\bv\in D}, \{\nu_\cv\}_{\cv\in A_0})\Big]
    \times \prod_{\av, \bv \in D, \cv\in A_0}
    \tau_\av^{r_\av}
    \eta_\bv^{s_\bv}
    \nu_\cv^{m_\cv}. 
\end{aligned}
\end{equation}
Plugging this to Eq. (\ref{eqn:reorganizing_sum_n_t_d}) proves Eq. (\ref{eqn:I_n_E_choice}). 

\vspace{3pt}
\noindent
{\bf Step 2. } We then show that the right hand side of Eq. (\ref{eqn:gen-multi-restate-equation}) coincides with $I(\{r_\av\}, \{ s_\bv\}, \{ m_\cv\})$ when taking $f_n = \Upsilon$ as defined in Eq. (\ref{eqn:function-f-choice}).
That is, we will show that 
\begin{equation}
\begin{aligned}
\Upsilon(\{W_\av\}_{\av \in A}; \{r_\av\}, \{ s_\bv\}, \{ m_\cv\}) = I(\{r_\av\}, \{ s_\bv\}, \{ m_\cv\}),
\end{aligned}
\end{equation}
where $\{ W_\av\}_{\av \in A}$ is the unique solution of Eq. (\ref{eqn:SCE_natural}).

Indeed, by Lemma \ref{lem:unique_solution_SCE}, the set union $\{ W_\xv \}_{\xv \in D} \cup \{ - W_\xv \}_{\xv \in D} \cup \{ Q_\av \}_{\av \in A_0}$ used in Proposition \ref{prop:well-played-again} coincide with the $\{ W_\av \}_{\av \in A}$ used in Proposition \hyperref[thm:gen-multinomial-restate]{4.1 (Formal)}. Then using the fact that $W_\av + W_{\bar \av} = 0$ for all $\av \in D$ and the fact that $W_\av = Q_\av$ for all $\av \in A_0$, we have
\begin{equation}
\begin{aligned}
&~\Upsilon(\{W_\av\}_{\av \in A}; \{r_\av\}, \{ s_\bv\}, \{ m_\cv\}) = \prod_{\av, \bv \in D, \cv\in A_0}
    (W_\av + W_{\bar \av})^{r_\av}
    (W_\bv - W_{\bar \bv})^{s_\bv}
    W_\cv^{m_\cv} \\
=&~ \prod_{\av, \bv \in D, \cv\in A_0}
    1(r_\av = 0)
    (2W_\bv)^{s_\bv}
    Q_\cv^{m_\cv} = I(\{r_\av\}, \{ s_\bv\}, \{ m_\cv\}). 
\end{aligned}
\end{equation}

\noindent
{\bf Step 3. } Combining Step 1 and 2, we have proved Eq. \eqref{eqn:gen-multi-restate-equation} hold when taking $f_n = \Upsilon$ as defined in Eq. \eqref{eqn:function-f-choice}, for any choice of $\{ r_\av, s_\bv, m_\cv \}_{\av, \bv \in D, \cv \in A_0}$. 

Note that for any general polynomials $f_n(\{ \omega_\av \})$, it can be decomposed to a superposition of $\Upsilon$ functions as
\[
f_n(\{ \omega_\av \}) = \sum_{\{ r_\av, s_\bv, m_\cv \}_{\av, \bv \in D, \cv \in A_0}} \tilde F_n(\{ r_\av, s_\bv, m_\cv \}_{\av, \bv \in D, \cv \in A_0})  \Upsilon(\{\omega_\av\}_{\av \in A}; \{r_\av\}, \{ s_\bv\}, \{ m_\cv\}). 
\]
Since $f_n$ has uniformly bounded degree, the number of non-zero coefficients $|\{ \{ r_\av, s_\bv, m_\cv \}_{\av, \bv \in D, \cv \in A_0}: \tilde F_n(\{ r_\av, s_\bv, m_\cv \}_{\av, \bv \in D, \cv \in A_0}) \neq 0 \} |$ is uniformly bounded. Moreover, since $f_n$ converges to $f$ (in terms of polynomial coefficients), there exists $\{ \tilde F(\{ r_\av, s_\bv, m_\cv \}_{\av, \bv \in D, \cv \in A_0}) \}$ such that
\[
\begin{aligned}
&~f(\{ \omega_\av \}) = \sum_{\{ r_\av, s_\bv, m_\cv \}_{\av, \bv \in D, \cv \in A_0}} \tilde F(\{ r_\av, s_\bv, m_\cv \}_{\av, \bv \in D, \cv \in A_0})  \Upsilon(\{\omega_\av\}_{\av \in A}; \{r_\av\}, \{ s_\bv\}, \{ m_\cv\}), \\
&~\text{where} \qquad  \lim_{n \to \infty} \tilde F_n(\{ r_\av, s_\bv, m_\cv \}_{\av, \bv \in D, \cv \in A_0})  = \tilde F(\{ r_\av, s_\bv, m_\cv \}_{\av, \bv \in D, \cv \in A_0}). 
\end{aligned}
\]
By the linearity of Eq. (\ref{eqn:gen-multi-restate-equation}) in $\{ f_n\}_{n \ge 1}$ and $f$, this implies that Eq. (\ref{eqn:gen-multi-restate-equation}) holds for any converging polynomials with uniformly bounded degree. This finishes the proof of Proposition \hyperref[thm:gen-multinomial-restate]{4.1 (Formal)}.

\subsection{More on the well-played polynomials}

Before we give the proof of Proposition \ref{prop:well-played-again}, we first define some useful quantities related to well-played polynomials in this section of the appendix.

\begin{defn}[Region of non-zero coefficient of sequence of well-played polynomials]\label{def:region-non-zero-coefficient}
Let $A = A_0 \sqcup D \sqcup \barD$ be a proper set (c.f. Definition \ref{def:set_A}). Let $\{ P_n \}_{n \ge 1}$ be a sequence of converging uniformly well-played polynomials with canonical representation of form (\ref{eq:poly_t_d_n_def_sequence}). We denote the region of non-zero coefficient of $\{ P_n \}_{n \ge 1}$ by
\begin{equation} \label{eq:Region-def}
\mathcal{R} = \Big\{ (\avu, \bvu, \cvu) : \exists n, \Psi_{\avu, \bvu, \cvu,n} \neq 0  \Big\}.
\end{equation}
\end{defn}

\begin{defn}[Linear partner of well-played polynomials]\label{def:linear-partner}
Let $A = A_0 \sqcup D \sqcup \barD$ be a proper set (c.f. Definition \ref{def:set_A}) and assume that there is an ordering $\succ$ over the set $D$. For a well-played polynomial $P$ of form (\ref{eq:poly_t_d_n_def_redefinition}), we let the \emph{linear partner} of $P$ be a polynomial $P_{\lin}$ with canonical representation
\begin{equation}\label{eqn:Plin-definition}
\begin{aligned}
\cC[P_{\lin}](\{\tau_\av\}_{\av \in D}, \{\eta_\bv\}_{\bv \in D}, \{\nu_\cv\}_{\cv\in A_0}) 
    &= \sum_{\avu,\bvu, \cvu: |\avu| = 1}
    \Psi_{\avu,\bvu, \cvu}
            \prod_{\av \in \avu} \tau_\av
            \prod_{\bv \in \bvu} \eta_\bv
            \prod_{\cv \in \cvu} \nu_\cv \\
    &=
    \sum_{\av\in D}
    \tau_\av
    \sum_{\bvu\in D^*, \cvu \in A_0^*}
    \Psi_{\av,\bvu, \cvu}
            \prod_{\bv \in \bvu} \eta_\bv
            \prod_{\cv \in \cvu} \nu_\cv
\end{aligned}
\end{equation}
which is understood as $\cC[P]$ restricted to terms that are linear in the variables $\{\tau_\av\}$. We further denote 
\begin{equation}\label{eqn:P-av-definition}
P_\av(\{\eta_\bv\}_{\bv \in D}, \{\nu_\cv\}_{\cv\in A_0}) = \sum_{\bvu\in D^*, \cvu \in A_0^*}
    \Psi_{\av,\bvu, \cvu}
            \prod_{\bv \in \bvu} \eta_\bv
            \prod_{\cv \in \cvu} \nu_\cv
\end{equation}
so that we have 
\[
\cC[P_{\lin}](\{\tau_\av\}_{\av \in D}, \{\eta_\bv\}_{\bv \in D}, \{\nu_\cv\}_{\cv\in A_0}) = \sum_{\av \in D} \tau_{\av} \cdot P_\av(\{\eta_\bv\}_{\bv \in D}, \{\nu_\cv\}_{\cv\in A_0}) . 
\]
\end{defn}

We next show that the sequence of converging uniformly well-played polynomials has a uniform bound of a specific form. 
\begin{lemma}\label{lem:property-well-played-polynomial}
For any sequence of converging uniformly well-played polynomials $\{ P_n \}_{n \ge 1}$ with canonical representation
\[
\cC[P_n](\{\tau_\av\}_{\av \in D}, \{\eta_\bv\}_{\bv \in D}, \{\nu_\cv\}_{\cv\in A_0}) = \sum_{\avu,\bvu \in D^*, \cvu \in A_0^*}
    \Psi_{\avu,\bvu, \cvu,n}
            \prod_{\av \in \avu} \tau_\av
            \prod_{\bv \in \bvu} \eta_\bv
            \prod_{\cv \in \cvu} \nu_\cv,
\]
there exists polynomials $\{ \bar P_\av(\{ \eta_\bv \}_{\bv \prec \av}) \}_{\av \in D}$ whose coefficient is independent of $n$, such that for any $\Ub_\bv \ge 1$, $t_\av / n \in [0, 1]$, $n_\cv / n \in [0, 1]$ we have
\begin{equation} \label{eq:Pbar-condition}
n \times \sum_{\avu,\bvu \in D^*, \cvu \in A_0^*}
    |\Psi_{\avu,\bvu, \cvu,n} |
            \prod_{\av \in \avu} \frac{t_\av}{n}
            \prod_{\bv \in \bvu} \Ub_\bv
            \prod_{\cv \in \cvu} \frac{n_\cv}{n} 
\le \sum_{\av \in D} t_\av \bar P_\av(\{ \Ub_\bv\}_{\bv \prec \av}). 
\end{equation}
\end{lemma}

\begin{proof}[Proof of Lemma \ref{lem:property-well-played-polynomial}]
Since $\{ P_n \}_{n \ge 1}$ is a sequence of uniformly well-played polynomials, $\Psi_{\avu, \bvu, \cvu}^n$ is nonzero only if $\avu \neq \emptyset$ and $\max(\avu) \succ \max(\bvu)$. 
Let $\av_1 = \max(\avu)$, then we have $|n\prod_{\av\in \avu} (t_\av/n)| \le t_{\av_1}$.
This implies
\begin{align}\label{eqn:proof-lemma-property-well-played-intermediate}
n \times \sum_{\avu,\bvu \in D^*, \cvu \in A_0^*}
    |\Psi_{\avu,\bvu, \cvu,n} |
            \prod_{\av \in \avu} \frac{t_\av}{n}
            \prod_{\bv \in \bvu} \Ub_\bv
            \prod_{\cv \in \cvu} \frac{n_\cv}{n} 
\le 
\sum_{(\avu,\bvu,\cvu)\in \mathcal{R}}  |\Psi_{\avu,\bvu, \cvu,n} | t_{\av_1} \prod_{\bv \in \bvu} \Ub_\bv,
\end{align}
where $\cR$ is the region of non-zero coefficients given by Definition \ref{def:region-non-zero-coefficient}, and every $\bv \prec \av_1$ necessarily by Definition~\ref{def:well-played-polynomial-redefine}. Note that $\{ P_n \}_{n \ge 1}$ is converging, so for any fixed $(\avu,\bvu, \cvu)$, we have
\begin{equation}\label{eqn:def_overline_Psi}
\overline{\Psi}_{\avu,\bvu, \cvu} \equiv \sup_{n \ge 1} |\Psi_{\avu,\bvu, \cvu,n} | < \infty.
\end{equation}
So we can let
\begin{equation}
\bar{P}_\av(\{\Ub_\bv\}_{\bv \prec \av} ) = \sum_{(\avu,\bvu,\cvu)\in \mathcal{R}, \max(\avu) = \av} \overline{\Psi}_{\avu,\bvu, \cvu} \prod_{\bv \in \bvu} \Ub_\bv
\end{equation}
which satisfies the condition \eqref{eq:Pbar-condition} by Eqs. (\ref{eqn:proof-lemma-property-well-played-intermediate}) and (\ref{eqn:def_overline_Psi}). This proves the lemma. 
\end{proof}

\subsection{Proof of Proposition \ref{prop:well-played-again}}
\label{sec:proof-well-played}

Throughout the proof that follows, we take $\{r_\av\}_{\av \in D}, \{ s_\bv\}_{\bv \in D}, \{ m_\cv\}_{\cv \in A_0}$ to be fixed integers and suppress the dependence of $I_n$ and $I$ in them for simplicity of notations. We define
\begin{equation} \label{eq:en-def}
\begin{aligned}
e_n(t) :=&~  \binom{n}{t} \sum_{\substack{ t_\av \ge 0, \forall \av \in D \\ \sum_{\av \in D} t_\av = t}} \binom{t}{\{t_\av\}_{\av \in D}} 
	  \sum_{\substack{ n_\cv \ge 0, \forall \cv \in A_0 \\ \sum_{\cv \in A_0} n_\cv = n-t}} \binom{n-t}{\{n_\cv\}_{\cv \in A_0}} \Big(\prod_{\cv\in A_0} Q_{\cv}^{n_\cv} \Big) \Big(\prod_{\av\in D} \sqint_\av \Big)  \\
	  &~ \times \Big\{ \exp\Big[ n \cdot \cC[P_n](\{t_\av/n\}_{\av\in D}, \{d_\bv/n\}_{\bv\in D}, \{n_\cv/n\}_{\cv\in A_0})\Big] \prod_{\av, \bv \in D, \cv\in A_0}
    \Big(\frac{t_\av }{n}\Big)^{r_\av}
    \Big(\frac{d_\bv }{n}\Big)^{s_\bv}
    \Big(\frac{n_\cv }{n}\Big)^{m_\cv} \Big\},
\end{aligned}
\end{equation}
where 
\[
\cC[P_n](\{\tau_\av\}_{\av\in D}, \{\eta_\bv \}_{\bv\in D}, \{\nu_\cv\}_{\cv\in A_0}) = \sum_{(\avu, \bvu, \cvu) \in \mathcal{R}}
    \Psi_{\avu, \bvu, \cvu,n} \prod_{\av \in \avu} \tau_\av \prod_{\bv \in \bvu} \eta_\bv
            \prod_{\cv \in \cvu} \nu_\cv
\]
as in Definition \ref{def:well-played-polynomial-redefine}.
Looking back to $I_n$ defined in Eq.~\eqref{eqn:In-definition}, it is easy to see that 
\begin{equation}\label{eqn:In_decomposition_en}
I_n = \sum_{t = 0}^n e_n(t). 
\end{equation}
We further define
\begin{equation} \label{eq:e-def}
\begin{aligned}
e(t):=&~
 \sum_{\substack{t_\av \ge 0, \forall \av \in D \\ \sum_{\av \in D} t_\av = t}} 
 \Big[  \prod_{\av \in D} \frac{1}{2\pi i}\oint_{\Disk} \frac{d z_\av}{z_\av}  (2 Q_\av /z_\av)^{t_\av}\Big] 
 \exp\Big[ \cC[P_{\lin}](\{t_\av\}_{\av\in D}, \{z_\bv\}_{\bv\in D}, \{Q_\cv\}_{\cv\in A_0}) \Big] \\
 &\qquad \qquad \qquad\qquad\qquad \times 
 \prod_{\av, \bv \in D, \cv\in A_0}
    1\{r_\av = 0\}
    z_\bv^{s_\bv}
    Q_\cv^{m_\cv},
\end{aligned}
\end{equation}
where
\[
\cC[P_{\lin}](\{t_\av\}_{\av\in D}, \{z_\bv \}_{\bv\in D}, \{Q_\cv \}_{\cv\in A_0}) = \sum_{\av\in D}
    t_\av
    \sum_{\bvu\in D^*, \cvu \in A_0^*}
    \Psi_{\av,\bvu, \cvu}
            \prod_{\bv \in \bvu} z_\bv
            \prod_{\cv \in \cvu} Q_\cv
\]
is related to the limiting polynomial $P$ as in Definition \ref{def:linear-partner}.

We finally denote
\begin{equation} \label{eq:E-def}
\begin{aligned}
E(t) :=&~ \inf_{\Ub_\av \ge 1, \forall \av \in D}   \sum_{\substack{t_\av \ge 0, \forall \av \in D \\ \sum_{\av \in D} t_\av = t}} \Big( \prod_{\av \in D}  |2 e Q_\av / \Ub_\av |^{t_\av} \Ub_\av^{s_\av}\Big) \times \exp\Big[ \sum_{\av \in D} t_\av \bar P_\av(\{\Ub_\bv\}_{\bv \prec \av}) \Big],
\end{aligned}
\end{equation}
where $\{ \bar P_\av \}_{\av \in D}$ are given by Lemma \ref{lem:property-well-played-polynomial}.

Given these definitions, we have the following lemmas whose proofs will be deferred to the following subsections. 
\begin{lemma}\label{lem:en_convergence_e}
Under the conditions of Proposition \ref{prop:well-played-again}, for fixed $t \in \Z_{\ge 0}$, we have
\begin{equation}\label{eqn:en_convergence_e}
\lim_{n \to \infty} e_n(t) 
= e(t),
\end{equation}
where $e_n$ is as defined in Eq. (\ref{eq:en-def}) and $e$ is as defined in Eq. (\ref{eq:e-def}) . 
\end{lemma}

\begin{lemma}\label{lem:en_bounded_E}
Under the conditions of Proposition \ref{prop:well-played-again}, for any $0 \le t \le n$, we have 
\begin{equation}\label{eqn:en_bounded_E}
|e_n(t) |\le E(t),
\end{equation}
where $e_n$ is as defined in Eq. (\ref{eq:en-def}) and $E$ is as defined in Eq. (\ref{eq:E-def}). 
\end{lemma}

\begin{lemma}\label{lem:e_summation}
Under the conditions of Proposition \ref{prop:well-played-again}, we have
\begin{equation}\label{eqn:e_summation}
\sum_{t = 0}^\infty e(t) = I,
\end{equation}
where $e$ is as defined in Eq. (\ref{eq:e-def}) and $I$ is as defined in Eq. (\ref{eqn:I-definition}). 
\end{lemma}

\begin{lemma}\label{lem:E_summable}
Under the conditions of Proposition \ref{prop:well-played-again}, we have
\begin{equation}\label{eqn:e_summable}
\sum_{t = 0}^\infty E(t) < \infty, 
\end{equation}
where $E$ is as defined in Eq. (\ref{eq:E-def}). 
\end{lemma}

We now complete the proof of Proposition~\ref{prop:well-played-again} based on these lemmas. By Eq. (\ref{eqn:In_decomposition_en}) and Lemma \ref{lem:e_summation}, for any $T \in \Z_{> 0}$, we have 
\[
\begin{aligned}
\Big|I_n -  I \Big|\le  \underbrace{\sum_{t = 0}^T \Big|e_n(t) -  e(t) \Big|}_{\eps_1(n,T)} + \underbrace{\sum_{t = T+1}^n \Big|e_n(t) \Big|}_{\eps_2(n,T)} + \underbrace{ \Big|\sum_{t > T} e(t)  \Big|}_{\eps_3(T)}. 
\end{aligned}
\]
By Lemma \ref{lem:e_summation} and \ref{lem:E_summable}, for any $\eps > 0$, there exists $T_\eps \in \Z_{> 0}$ such that
\begin{equation}
\eps_3(T_\eps) = \Big|\sum_{t > T_\eps} e(t) \Big|\le \eps / 3 , 
\qquad \text{and} \qquad
\sum_{t > T_\eps} E(t) \le \eps / 3.    
\end{equation}
Then by Lemma \ref{lem:en_bounded_E} and the latter above, we have  for any $n$,
\begin{equation}
\eps_2(n,T_\eps) = \sum_{t = T_\eps + 1}^n \Big|e_n(t) \Big|\le \sum_{t > T_\eps} E(t) \le \eps / 3. 
\end{equation}
And for any such fixed $T_\eps$, Lemma \ref{lem:en_convergence_e} implies that there exists $n_\eps \in \Z_{> 0}$ such that for any $n \ge n_\eps$, we have 
\begin{equation}
\eps_1(n, T_\eps) = \sum_{t = 0}^{T_\eps} \Big|e_n(t) -  e(t) \Big|\le \eps / 3.
\end{equation}
This implies that for any $n \ge n_\eps$, we have $|I_n -  I |\le \eps$.
This proves Proposition \ref{prop:well-played-again}.

\subsubsection{Proof of Lemma \ref{lem:en_convergence_e}: $e_n(t) \to e(t)$}

\noindent
{\bf Step 1. Reformulate $e_n(t)$ as the application of operators $\TT_n^t, \SS_n^{\{ t_\av\}}, \UU_n^t$, and $\PP$.}~ 

Note that the exponential in the $e_n(t)$ as defined in Eq. (\ref{eq:en-def}) gives
\begin{align}
\exp\Big[ n \cdot \cC[P_n](\{t_\av/n\}_{\av\in D}, \{d_\bv/n\}_{\bv\in D}, \{n_\cv/n\}_{\cv\in A_0}) \Big] = \exp\Big[ n \sum_{(\avu, \bvu,\cvu) \in \mathcal{R}} \Psi_{\avu,\bvu, \cvu,n}
            \prod_{\av \in \avu} \frac{t_\av}{n}
            \prod_{\bv \in \bvu} \frac{d_\bv}{n}
            \prod_{\cv \in \cvu} \frac{n_\cv}{n} \Big]. 
\end{align}
Then, consider the following identity based on Taylor expansion:
\begin{align} \label{eq:Taylor-exp}
    \exp\Big[\sum_{i = 1}^N a_i \Big] = \sum_{k_i = 0, \forall i \in [N]}^\infty \prod_{i \in [N]} \frac{1}{k_i!}a_i^{k_i}
\end{align}
Applying the above identity \eqref{eq:Taylor-exp},
we get
\begin{equation}
\begin{aligned}
&~\exp\Big[ n \sum_{(\avu, \bvu,\cvu) \in \mathcal{R}} \Psi_{\avu,\bvu, \cvu,n}
            \prod_{\av \in \avu} \frac{t_\av}{n}
            \prod_{\bv \in \bvu} \frac{d_\bv}{n}
            \prod_{\cv \in \cvu} \frac{n_\cv}{n} \Big] \\
=&~ \sum_{\{\pois_{\avu, \bvu, \cvu} \ge 0:~ (\avu, \bvu, \cvu) \in \mathcal{R}\} } 
~~ \prod_{(\avu, \bvu, \cvu) \in \mathcal{R}}  \frac{1}{\pois_{\avu, \bvu, \cvu} ! } \Big[ n \Psi_{\avu,\bvu, \cvu,n}
            \prod_{\av \in \avu} \frac{t_\av}{n}
            \prod_{\bv \in \bvu} \frac{d_\bv}{n}
            \prod_{\cv \in \cvu} \frac{n_\cv}{n} \Big]^{\pois_{\avu, \bvu, \cvu}}\\
=&~  \PP \bigg[ \prod_{(\avu, \bvu, \cvu) \in \mathcal{R}}  \Big[ n \Psi_{\avu,\bvu, \cvu,n}
            \prod_{\av \in \avu} \frac{t_\av}{n}
            \prod_{\bv \in \bvu} \frac{d_\bv}{n}
            \prod_{\cv \in \cvu} \frac{n_\cv}{n} \Big]^{\pois_{\avu, \bvu, \cvu}} \bigg], 
\end{aligned}
\end{equation}
where $\PP$ is defined as the following operator,  upon an input function $F(\{ \pois_{\avu, \bvu, \cvu} : (\avu, \bvu, \cvu) \in \mathcal{R}\})$, yields
\begin{equation}\label{eqn:PP-definition}
\PP F = \sum_{\{ \pois_{\avu, \bvu, \cvu}\ge 0:~ (\avu, \bvu, \cvu) \in \mathcal{R}\} }
F(\{ \pois_{\avu, \bvu, \cvu} \}) 
\prod_{(\avu, \bvu, \cvu) \in \mathcal{R}} \frac{1}{\pois_{\avu, \bvu, \cvu} !} . 
\end{equation}

Furthermore, we define the $\TT_n^t$ operator acting on a function $f(\{ t_\av : \av \in D\})$ as
\begin{equation}\label{eqn:TTn-definition}
\TT_n^t f =  n^{-t} \binom{n}{t} \sum_{t_\av \ge 0, \forall \av \in D } \binom{t}{\{t_\av\}} f (\{ t_\av \} ).
\end{equation}

Next, we define the $\SS_n^{\{ t_\av\}}$ operator acting on a function $g(\{ d_\bv / n : \bv \in D\})$ as
\begin{equation}\label{eqn:SSn-definition}
\SS_n^{\{ t_\av\}} g = \Bigg(\prod_{\av \in D}  n^{t_\av} \sqint_\av\Bigg) g(\{ d_\bv / n\}) = \Bigg(\prod_{\av \in D}  n^{t_\av} \hspace{-5pt}  \sum_{n_\av + n_{\bar\av} = t_\av} \binom{t_\av}{n_\av, n_{\bar\av}} Q_\av^{n_\av} (- Q_\av)^{n_{\bar\av} }\Bigg)
g(\{ (n_\bv - n_{\bar \bv}) / n\}_{\bv \in D}) .
\end{equation}

Lastly, we define the $\UU_n^t$ operator acting on a function $h(\{ n_\cv / n: \cv\in A_0\})$ as
\begin{equation}
\UU_n^t h = \sum_{\{n_\cv: \cv \in A_0\}} \binom{n-t}{\{n_\cv\}} h(\{ n_\cv / n\}) \prod_{\cv\in A_0} Q_{\cv}^{n_\cv}. 
\end{equation}
With these operators defined, we can rewrite \eqref{eq:en-def} as
\begin{equation}
\begin{aligned}
e_n(t) = \UU_n^t \, \TT_n^t \,\SS_n^{\{ t_\xv\}} \,  &\PP \, \Bigg[
    \prod_{(\avu, \bvu, \cvu) \in \mathcal{R}}  \Big[ n \Psi_{\avu,\bvu, \cvu,n}
            \prod_{\av \in \avu} \frac{t_\av}{n}
            \prod_{\bv \in \bvu} \frac{d_\bv}{n}
            \prod_{\cv \in \cvu} \frac{n_\cv}{n} \Big]^{\pois_{\avu, \bvu, \cvu}} \\
            &\qquad
    \times  \prod_{\av, \bv \in D, \cv\in A_0}
    \Big(\frac{t_\av }{n}\Big)^{r_\av}
    \Big(\frac{d_\bv }{n}\Big)^{s_\bv}
    \Big(\frac{n_\cv }{n}\Big)^{m_\cv}
\Bigg]. 
\end{aligned}
\end{equation}
This equality used the fact that the $\prod_{\av \in D} n^{t_\av}$ factors in $\SS_n^{\{ t_\av \}}$ cancelled the factor of $n^{-t}$ in $\TT_n^t$.

\vspace{3pt}
\noindent
{\bf Step 2. Pointwise convergence. }  Before taking the $n\to\infty$ limit on $e_n$, we consider instead the following quantity which is the part of $e_n$ without the $\PP$ operator:
\begin{equation}\label{eqn:J_n_in_lemma}
\begin{aligned}
&~J_n(t; \{ \pois_{\avu, \bvu, \cvu} \}) \\
=&~ \UU_n^t \,\TT_n^t\, \SS_n^{\{ t_\xv\}} \, \Bigg[
    \prod_{(\avu, \bvu, \cvu) \in \mathcal{R}}  \Big[ n \Psi_{\avu,\bvu, \cvu,n}
            \prod_{\av \in \avu} \frac{t_\av}{n}
            \prod_{\bv \in \bvu} \frac{d_\bv}{n}
            \prod_{\cv \in \cvu} \frac{n_\cv}{n} \Big]^{\pois_{\avu, \bvu, \cvu}} 
    \times \prod_{\av, \bv \in D, \cv\in A_0}
    \Big(\frac{t_\av }{n}\Big)^{r_\av}
    \Big(\frac{d_\bv }{n}\Big)^{s_\bv}
    \Big(\frac{n_\cv }{n}\Big)^{m_\cv}
\Bigg]. 
\end{aligned}
\end{equation}
Then with $J_n$ defined as such, we have 
\begin{equation}\label{eqn:en-equal-PJ}
e_n(t) = \PP[J_n(t; \{ \pois_{\avu, \bvu, \cvu} \})]. 
\end{equation}
By rearranging the order of the operator and monomials in $J_n$, we have
\begin{equation}\label{eqn:Jn-reformulation}
\begin{aligned}
J_n(t; \{ \pois_{\avu, \bvu, \cvu} \}) =&~ \prod_{(\avu, \bvu, \cvu) \in \mathcal{R}}  (\Psi_{\avu,\bvu, \cvu,n})^{\pois_{\avu, \bvu, \cvu}} \times \UU_n^t\Big[  \Big(\prod_{\cv\in A_0} \Big(\frac{n_\cv }{n}\Big)^{m_\cv} \Big) \times \Big(\prod_{(\avu, \bvu, \cvu) \in \mathcal{R}} \Big(\prod_{\cv \in \cvu} \frac{n_\cv}{n}\Big)^{\pois_{\avu, \bvu, \cvu}} \Big) \Big] \\
&~  \times  \TT_n^t \Bigg\{ \Big(\prod_{\av \in D}
    \Big(\frac{t_\av }{n}\Big)^{r_\av}\Big)
    \times
    \Big(\prod_{(\avu, \bvu, \cvu) \in \mathcal{R}}  \Big( n
            \prod_{\av \in \avu} \frac{t_\av}{n}
     \Big)^{\pois_{\avu, \bvu, \cvu}} \Big) \\
     & \qquad \qquad 
    \times \SS_n^{\{ t_\av\}} \Big[ \Big( \prod_{\bv \in D}
    \Big(\frac{d_\bv }{n}\Big)^{s_\bv} \Big) \times
    \Big(\prod_{(\avu, \bvu, \cvu) \in \mathcal{R}}  \Big(\prod_{\bv \in \bvu} \frac{d_\bv}{n}
     \Big)^{\pois_{\avu, \bvu, \cvu}} \Big) \Big] 
\Bigg\}, 
\end{aligned}
\end{equation}
We now take the $n\to\infty$ limit of $J_n$. We decompose this into the following steps:
\begin{itemize}
\item First, since $\{ P_n \}_{n \ge 1}$ is a sequence of converging polynomials, we have $\lim_{n \to \infty} \Psi_{\avu,\bvu, \cvu,n} = \Psi_{\avu,\bvu, \cvu}$ where $\Psi_{\avu,\bvu, \cvu}$ is the coefficient of the canonical representation of the limiting polynomial $P$. This gives
\begin{equation}\label{eqn:Psin_limit}
\lim_{n \to \infty} \prod_{(\avu, \bvu, \cvu) \in \mathcal{R}}  (\Psi_{\avu,\bvu, \cvu,n})^{\pois_{\avu, \bvu, \cvu}} = \prod_{(\avu, \bvu, \cvu) \in \mathcal{R}}  (\Psi_{\avu,\bvu, \cvu})^{\pois_{\avu, \bvu, \cvu}}. 
\end{equation}
\item Secondly, note that for any fixed $t$, fixed $\{ m_\cv\}_{\cv \in A_0}$, and fixed $\{ \pois_{\avu, \bvu, \cvu} \}_{(\avu, \bvu, \cvu) \in \mathcal{R}}$, we have
\begin{equation}\label{eqn:Unt-limit}
\begin{aligned}
\lim_{n\to\infty} \UU_n^t \Big[ \Big(\prod_{\cv\in A_0} \Big(\frac{n_\cv }{n}\Big)^{m_\cv} \Big)\prod_{(\avu, \bvu, \cvu) \in \mathcal{R}} \Big(\prod_{\cv \in \cvu} \frac{n_\cv}{n}\Big)^{\pois_{\avu, \bvu, \cvu}} \Big] = \Big[ \Big(\prod_{\cv\in A_0} \Big(Q_\cv\Big)^{m_\cv}\Big) \prod_{(\avu, \bvu, \cvu) \in \mathcal{R}} \Big(\prod_{\cv \in \cvu} Q_\cv\Big)^{\pois_{\avu, \bvu, \cvu}} \Big],
    \end{aligned}
\end{equation}
which follows from the fact that $\sum_{\cv \in A_0} Q_{\cv} = 1$ and $Q_\cv \ge 0$, so we can think of $n_\cv$ as distributed according to a multinomial distribution, and by law of large numbers, as $n\to\infty$, the moments of $n_\cv/n$ converges to the powers of $Q_\cv$ (we use the fact that $t$ is fixed as $n \to \infty$). 
\item Thirdly, by Eq. \eqref{eqn:little_sum_monomial_limit_multi} in Lemma \ref{lem:little_sum_lemma_multi} (shown later in Section~\ref{sec:auxiliary}), we have
\begin{align}
&\lim_{n\to\infty} \SS_n^{\{ t_\av\}} \Big[ \Big(
    \prod_{\bv \in D}
    \Big(\frac{d_\bv }{n}\Big)^{s_\bv}\Big)
    \times
    \prod_{(\avu, \bvu, \cvu) \in \mathcal{R}}  \Big(\prod_{\bv \in \bvu} \frac{d_\bv}{n}
     \Big)^{\pois_{\avu, \bvu, \cvu}} 
    \Big]   \nonumber \\
&\qquad  = \SS^{\{t_\av\}} \Big[ 
    \Big(\prod_{\bv \in D} \Big(z_\bv\Big)^{s_\bv}\Big)
    \times
    \prod_{(\avu, \bvu, \cvu) \in \mathcal{R}}  \Big(\prod_{\bv \in \bvu} z_\bv
     \Big)^{\pois_{\avu, \bvu, \cvu}} 
    \Big],
\label{eqn:Snt-limit-St}
\end{align}
where for a monomial $g(\{ d_\bv /n\}_{\bv \in D})$, we have
\begin{equation}
\SS^{\{ t_\av\}} g(\{ z_\bv \}_{\bv \in D}) = 
\prod_{\av \in D} \Big[\frac{t_{\av} !}{2 \pi i} \oint_{\mathbb D}   (2 Q_\av / z_\av)^{t_\av} \frac{d z_\av}{z_\av} \Big] g(\{z_\bv\}_{\bv \in D}). 
\end{equation}
\item Fourthly, by the definition of well-played polynomial as in Definition \ref{def:well-played-polynomial-redefine}, for any $(\avu, \bvu, \cvu) \in \mathcal{R}$, we have $|\avu |\ge 1$.
Therefore, for any fixed $\{ t_\av \}_{\av \in D}$, we have
\begin{align}\label{eqn:tav-term-limit}
&
\lim_{n \to \infty} 
    \Big(\prod_{\av \in D} \Big(\frac{t_\av }{n}\Big)^{r_\av}\Big)
    \times 
    \prod_{(\avu, \bvu, \cvu) \in \mathcal{R}}  \Big( n
             \prod_{\av \in \avu} \frac{t_\av}{n}
     \Big)^{\pois_{\avu, \bvu, \cvu}} 
     \nonumber \\
    &\qquad\qquad
    =  \Big(\prod_{\av \in D} 1\{ r_\av = 0\} \Big)
    \times \Big(\prod_{(\av, \bvu, \cvu) \in \mathcal{R}}  t_\av ^{\pois_{\av, \bvu, \cvu}} \Big) \times \Big(\prod_{(\avu, \bvu, \cvu) \in \mathcal{R}: |\avu |\ge 2}  1\{ \pois_{\avu, \bvu, \cvu} = 0 \} \Big).
\end{align}
Here the $\prod_{\av \in D} 1\{ r_\av = 0\}$ factor is by the fact that $\{ t_\av \}$ are fixed as $n \to \infty$ so that $\lim_{n \to \infty} \prod_{\av \in D}(t_\av/n)^{r_\av} = \prod_{\av \in D} 1\{ r_\av = 0\}$. The $\prod_{(\avu, \bvu, \cvu) \in \mathcal{R}: |\avu |\ge 2}  1\{ \pois_{\avu, \bvu, \cvu} = 0 \}$ factor is due to the following observation: for any $\pois_{\avu, \bvu, \cvu} \ge 1$ with $|\avu |\ge 2$, we have $\lim_{n \to \infty} [n \prod_{\av \in \avu} (t_\av / n)]^{\pois_{\avu, \bvu, \cvu}} = 0$. 
\item Finally, for any sequence of functions $f_n(\{ t_\xv \})$, we have
\begin{equation}\label{eqn:TTn-limit}
\lim_{n \to \infty}\TT_n^t f_n = \lim_{n \to \infty} n^{-t} {n \choose t} \sum_{\substack{t_\av \ge 0, \forall \av \in D \\ \sum_{\av \in D} t_\av = t}} \binom{t}{\{t_\av\}} f_n (\{ t_\av \} ) = \TT^t \lim_{n\to\infty} f_n (\{ t_\av \}),  
\end{equation}
where $\TT^t$ is the operator such that for any $f(\{ t_\av \})$, we have
\begin{equation}\label{eqn:Tt-definition}
\TT^t f = \frac{1}{t!} \sum_{\substack{t_\av \ge 0, \forall \av \in D \\ \sum_{\av \in D} t_\av = t}} \binom{t}{\{t_\av\}} f (\{ t_\av \}). 
\end{equation}
\end{itemize}
Combining Eqs. \eqref{eqn:Psin_limit}, \eqref{eqn:Unt-limit}, \eqref{eqn:Snt-limit-St}, \eqref{eqn:tav-term-limit}, and \eqref{eqn:TTn-limit} above, we have for any fixed $\{ \pois_{\avu, \bvu, \cvu} \}$ the following:
\begin{equation}\label{eqn:Jn-to-J}
\lim_{n\to \infty}J_n (t; \{ \pois_{\avu, \bvu, \cvu} \}) = J(t; \{ \pois_{\avu, \bvu, \cvu} \}). 
\end{equation}
where 
\begin{equation}
\begin{aligned}
J(t; \{ \pois_{\avu, \bvu, \cvu} \}) =&~ \Big(\prod_{(\avu, \bvu, \cvu) \in \mathcal{R}}  (\Psi_{\avu,\bvu, \cvu})^{\pois_{\avu, \bvu, \cvu}}\Big) \times 
\Big(\prod_{\cv\in A_0} Q_\cv^{m_\cv}  \Big) \times
\Big(\prod_{(\avu, \bvu, \cvu) \in \mathcal{R}} \Big(\prod_{\cv \in \cvu} Q_\cv\Big)^{\pois_{\avu, \bvu, \cvu}} \Big)  \\
&~  \times  \TT^t \Bigg\{ 
    \bigg(\prod_{\av \in D} 1\{ r_\av = 0\} \bigg)
    \times \bigg(\prod_{(\av, \bvu, \cvu) \in \mathcal{R}}  t_\av ^{\pois_{\av, \bvu, \cvu}} \bigg) \times \bigg(\prod_{(\avu, \bvu, \cvu) \in \mathcal{R}: |\avu |\ge 2}  1\{ \pois_{\avu, \bvu, \cvu} = 0 \}\bigg)
    \\
    &~\qquad \qquad
    \times \SS^{\{ t_\av\}} \bigg[ 
    \Big(\prod_{\bv \in D} z_\bv^{s_\bv}\Big)
    \times
    \Big(\prod_{(\avu, \bvu, \cvu) \in \mathcal{R}}  \Big(\prod_{\bv \in \bvu} z_\bv
     \Big)^{\pois_{\avu, \bvu, \cvu}} \Big)
    \bigg]
\Bigg\}. 
    \end{aligned}
\end{equation}
With some rearrangement of the order of operators and monomials in $J$, we have
\begin{equation}
\begin{aligned}
J(t; \{ \pois_{\avu, \bvu, \cvu} \}) =&~ \TT^t \, \SS^{\{ t_\av\}} \Bigg[
    \bigg(\prod_{(\av, \bvu, \cvu) \in \mathcal{R}}  \Big( t_\av \Psi_{\av,\bvu, \cvu} \prod_{\bv \in \bvu} z_\bv
            \prod_{\cv \in \cvu} Q_\cv \Big)^{\pois_{\av, \bvu, \cvu}} \bigg)
    \times \bigg( \prod_{(\avu, \bvu, \cvu) \in \mathcal{R}: |\avu |\ge 2}  1\{ \pois_{\avu, \bvu, \cvu} = 0 \} \bigg)
    \\
    &~ \qquad \qquad \times \prod_{\av, \bv \in D, \cv\in A_0}
    1\{r_\av = 0\} z_{\bv}^{s_\bv} Q_\cv^{m_\cv}\Bigg],
    \end{aligned}
\end{equation}
here we used the fact that $(\Psi_{\av,\bvu, \cvu} \prod_{\cv \in \cvu} Q_\cv \prod_{\bv \in \bvu} z_\bv)^{\pois_{\avu, \bvu, \cvu}} = 1$ when $\pois_{\avu, \bvu, \cvu} = 0$. 

\noindent
{\bf Step 3. Apply the dominant convergence theorem. } A simple calculation of $\PP[J(t; \{ \pois_{\avu, \bvu, \cvu} \})]$ gives that
\begin{equation}\label{eqn:PJ-equal-e}
\begin{aligned}
& \hspace{-20pt}
\PP[J(t; \{ \pois_{\avu, \bvu, \cvu} \})] \\
&\stackrel{(i)}{=}  \TT^t \, \SS^{\{ t_\av\}}\Bigg[ \exp\Big[
   \sum_{(\av, \bvu, \cvu) \in \mathcal{R}} t_\av \Psi_{\av,\bvu, \cvu} \prod_{\bv \in \bvu} z_\bv
            \prod_{\cv \in \cvu} Q_\cv \Big] \times \prod_{\av, \bv \in D, \cv\in A_0}
    1\{r_\av = 0\}  z_{\bv}^{s_\bv}  Q_\cv^{m_\cv}
    \Bigg] \\
    &\stackrel{(ii)}{=}   
    \sum_{\substack{t_\av \ge 0, \forall \av \in D \\ \sum_{\av \in D} t_\av = t}} 
    \Big[ \prod_{\av \in D} \frac{1}{2 \pi i} \oint_{\mathbb D}   (2 Q_\av / z_\av)^{t_\av} \frac{d z_\av}{z_\av} \Big] 
    \times \exp\Big[
   \sum_{(\av, \bvu, \cvu) \in \mathcal{R}} t_\av\Psi_{\av,\bvu, \cvu} \prod_{\bv \in \bvu} z_\bv
            \prod_{\cv \in \cvu} Q_\cv \Big]
    \\
    & \qquad \hspace{165pt}\times 
    \prod_{\av, \bv \in D, \cv\in A_0}
        1\{r_\av = 0\}  z_{\bv}^{s_\bv}  Q_\cv^{m_\cv}
     \\
    &= e(t), 
\end{aligned}
\end{equation}
where $e(t)$ is as defined in Eq. (\ref{eq:e-def}). Here $(i)$ used the definition of $\PP$ as in Eq. (\ref{eqn:PP-definition}) and the Taylor expansion of exponential function as in Eq. (\ref{eq:Taylor-exp}), and $(ii)$ used the definition of $\TT^t$ and $\SS^{\{ t_\av\}}$ and the cancellation of $t!$ on the denominator and numerator, as well as the cancellation of $\prod_{\av \in D} t_\av !$. 

So far we have shown $e_n(t) = \PP[J_n(t)]$, and $\PP[\lim_{n\to\infty} J_n(t)] = e(t)$ in Eqs. \eqref{eqn:en-equal-PJ}, \eqref{eqn:Jn-to-J}, and \eqref{eqn:PJ-equal-e}.
Then to prove the Lemma~\ref{lem:en_convergence_e}'s statement that $\lim_{n \to \infty} e_n(t) = e(t)$, all we need to do is switch the order of the $n$-limit and $\PP$ operation.
This can be done by invoking the dominant convergence theorem,
which can be used under the condition that there exists an upper bound $\overline J(t; \{ \pois_{\avu, \bvu, \cvu} \})$ of $J_n$ with $\PP \overline J(t; \{ \pois_{\avu, \bvu, \cvu} \}) < \infty$. 

We now show this upper bound exists.
Recall the expression of $J_n$ as in Eq. (\ref{eqn:Jn-reformulation}). Note that for any $\{ n_\cv\}_{\cv\in A_0}$, $\{ t_\av \}_{\av\in D}$ and $\{\pois_{\avu, \bvu, \cvu} \}_{(\avu, \bvu, \cvu) \in \mathcal{R}}$ where $0 \le n_\cv \le n$ and $0 \le t_\av \le t \le n$, we have
\begin{align}
\UU_n^t\Big[ \Big( \prod_{\cv\in A_0} \Big(\frac{n_\cv }{n}\Big)^{m_\cv} \Big) \times \Big(\prod_{(\avu, \bvu, \cvu) \in \mathcal{R}} \Big[\prod_{\cv \in \cvu} \frac{n_\cv}{n}\Big]^{\pois_{\avu, \bvu, \cvu}} \Big) \Big]  
&\le 1, \\
\text{and} \qquad
\Big( \prod_{(\avu, \bvu, \cvu) \in \mathcal{R}}  \Big[ n
            \prod_{\av \in \avu} \frac{t_\av}{n}
     \Big]^{\pois_{\avu, \bvu, \cvu}} \Big)
    \times \Big( \prod_{\av \in D}
    \Big(\frac{t_\av }{n}\Big)^{r_\av} \Big)
    &\le \prod_{(\avu, \bvu, \cvu) \in \mathcal{R}} t^{\pois_{\avu, \bvu, \cvu}},
\end{align}
where we used the fact that $|\avu |\ge 1$ for any $(\avu, \bvu, \cvu) \in \mathcal{R}$ so that $\max_{(\avu, \bvu, \cvu) \in \mathcal{R}} [n \prod_{\av \in \avu} \frac{t_\av}{n}] \le t$.
As a consequence, we have (recall that $\overline{\Psi}_{\avu,\bvu, \cvu} \equiv \sup_{n \ge 1} |\Psi_{\avu,\bvu, \cvu,n} | < \infty$ as in Eq. (\ref{eqn:def_overline_Psi}))
\[
\begin{aligned}
&~|J_n(t; \{ \pois_{\avu, \bvu, \cvu} \})|\\
\le&~  \prod_{(\avu, \bvu, \cvu) \in \mathcal{R}}  |t \cdot \overline{\Psi}_{\avu,\bvu, \cvu} |^{\pois_{\avu, \bvu, \cvu}} \times \Bigg|\TT_n^t   \SS_n^{\{ t_\av\}} \Big[ 
    \prod_{(\avu, \bvu, \cvu) \in \mathcal{R}}  \Big[\prod_{\bv \in \bvu} \frac{d_\bv}{n}
     \Big]^{\pois_{\avu, \bvu, \cvu}} 
    \times \prod_{\bv \in D}
    \Big(\frac{d_\bv }{n}\Big)^{s_\bv}\Big] \Bigg|\\
\stackrel{(i)}{\le}&~  \prod_{(\avu, \bvu, \cvu) \in \mathcal{R}}  |t \cdot \overline{\Psi}_{\avu,\bvu, \cvu} |^{\pois_{\avu, \bvu, \cvu}} \times \sum_{\sum_\av t_\av = t} \frac{1}{\prod_{\av \in D} t_\av !}\Bigg|  \SS_n^{\{ t_\av\}} \Big[ 
    \prod_{(\avu, \bvu, \cvu) \in \mathcal{R}}  \Big[\prod_{\bv \in \bvu} \frac{d_\bv}{n}
     \Big]^{\pois_{\avu, \bvu, \cvu}} 
    \times \prod_{\bv \in D}
    \Big(\frac{d_\bv }{n}\Big)^{s_\bv}\Big] \Bigg|\\
\stackrel{(ii)}{\le}&~ \prod_{(\avu, \bvu, \cvu) \in \mathcal{R}}  |t \cdot \overline{\Psi}_{\avu,\bvu, \cvu} |^{\pois_{\avu, \bvu, \cvu}} \times  \sum_{\sum_\av t_\av = t} \prod_{\av \in D} \frac{|2 Q_\av t_\av |^{t_\av}}{t_\av !} \equiv \overline J(t, \{ \pois_{\avu, \bvu, \cvu}\}), 
\end{aligned}
\]
where $(i)$ used the definition of $\TT_n^t$ as in Eq. (\ref{eqn:TTn-definition}), and $(ii)$ used Eq.~\eqref{eqn:bound-Sn_multi} in Lemma \ref{lem:little_sum_lemma_multi}. 

For the quantity $\overline J$, note that it is only exponential in $\{ \pois_{\avu, \bvu, \cvu} \}$, so we have $\PP[\overline J(t; \{ \pois_{\avu, \bvu, \cvu} \})] < \infty$.
Thus the conditions of the dominant convergence theorem are satisfied, and applying it along with Eqs.~\eqref{eqn:en-equal-PJ}, \eqref{eqn:Jn-to-J}, and \eqref{eqn:PJ-equal-e} shows
\begin{equation}
\lim_{n \to \infty} e_n(t) = \lim_{n \to \infty} \PP[J_n (t; \{ \pois_{\avu, \bvu, \cvu} \})] = \PP[J (t; \{ \pois_{\avu, \bvu, \cvu} \})] = e(t). 
\end{equation}
This proves the lemma.

\subsubsection{Proof of Lemma \ref{lem:en_bounded_E}: $|e_n(t)| \le E(t)$}

We define 
\begin{align}
&b_n(t; \{ t_\av \}, \{n_\cv\} ) \nonumber \\
&\qquad =  \Big( \prod_{\av\in D}  \frac{n^{t_\av}}{t_\av !}  \Big) \Big(\prod_{\av\in D} \sqint_\av \Big) 
 \exp\Big[ n \cdot \cC[P_n](\{t_\av/n\}_{\av\in D}, \{d_\bv/n\}_{\bv\in D}, \{n_\cv/n\}_{\cv\in A_0})\Big] \prod_{\bv \in D}
    \Big(\frac{d_\bv }{n}\Big)^{s_\bv}. 
    \label{eq:bn-def}
\end{align}
Then by the definition $e_n(t)$ in Eq. (\ref{eq:en-def}), we have 
\begin{align}
|e_n(t) |=&~ \Big|\binom{n}{t} \sum_{\sum_{\av} t_\av = t} \binom{t}{\{t_\av\}} \sum_{\sum_{\cv} n_\cv = n-t} \binom{n-t}{\{n_\cv\}} \Big(\prod_{\cv\in A_0} Q_{\cv}^{n_\cv}\Big) \times \Big( \prod_{\av\in D}  \frac{t_\av !}{n^{t_\av}}  \Big) \nonumber \\
&~ \times b_n( t; \{ t_\av \}, \{n_\cv\} )  \prod_{\cv\in A_0} \Big(\frac{n_\cv}{n}\Big)^{m_\cv} \prod_{\av\in D} \Big(\frac{t_\av}{n}\Big)^{r_\av} \Big| \nonumber \\
=&~  \Big|\frac{n!}{(n - t)! n^{t}} \sum_{\sum_{\av} t_\av = t} \sum_{\sum_{\cv} n_\cv = n-t} \binom{n-t}{\{n_\cv\}}\prod_{\cv\in A_0} Q_{\cv}^{n_\cv} \cdot b_n( t; \{ t_\av \}, \{n_\cv\} ) \prod_{\cv\in A_0} \Big(\frac{n_\cv}{n}\Big)^{m_\cv} \prod_{\av\in D}  \Big(\frac{t_\av}{n}\Big)^{r_\av} \Big|  \nonumber \\
\le&~  \sum_{\sum_{\av} t_\av = t} \sup_{\{ n_\cv\}, n_\cv \le n} \Big|b_n( t; \{ t_\av \}, \{n_\cv\}) \Big|
\label{eq:en-bound1}
\end{align}
where the last inequality used the facts that $t\le n$, $n_\cv / n \in [0, 1]$ for $\cv \in A_0$, $t_\av / n \in [0, 1]$ for $\av \in D$, $0\le Q_\cv \le 1$ for $\cv \in A_0$, and $\sum_{\cv \in A_0} Q_\cv = 1$.

To upper bound $|b_n|$, we plug the form of $\cC[P_n]$ from Eq.~\eqref{eq:poly_t_d_n_def_sequence} into Eq.~\eqref{eq:bn-def}, then apply Lemma \ref{lem:upper_bound_little_sum_exponential} (given later in Section~\ref{sec:auxiliary}) to bound it with an infinimum, and get
\[
\begin{aligned}
&~|b_n(t; \{ t_\av \}, \{n_\cv\} ) |\\
= &~\Big|\Big( \prod_{\av\in D}  \frac{n^{t_\av}}{t_\av !}  \Big) \Big( \prod_{\av\in D} \sqint_\av \Big)  \exp\Big[ n \sum_{\avu \in D^*, \bvu \in D^*, \cvu \in A_0^*} \Psi_{\avu,\bvu, \cvu,n}
            \prod_{\av \in \avu} \frac{t_\av}{n}
            \prod_{\bv \in \bvu} \frac{d_\bv}{n}
            \prod_{\cv \in \cvu} \frac{n_\cv}{n} \Big] \prod_{\bv \in D} \Big(\frac{d_\bv}{n}\Big)^{s_\bv} \Big|\\
\le&~ \inf_{\Ub_\av \ge 1, \forall \av \in D} \Big( \prod_{\av \in D} | 2 e Q_\av / \Ub_\av |^{t_\av} \Big) 
        \exp\Big[ n \sum_{\avu \in D^*, \bvu \in D^*, \cvu \in A_0^*} | \Psi_{\avu,\bvu, \cvu,n} |
            \prod_{\av \in \avu} \frac{t_\av}{n}
            \prod_{\bv \in \bvu} \Ub_\bv
            \prod_{\cv \in \cvu} \frac{n_\cv}{n} \Big]
        \prod_{\bv \in D} \Ub_\bv^{s_\bv}. 
\end{aligned}
\]
Then applying Lemma \ref{lem:property-well-played-polynomial} to the exponential, we get an upper bound of $|b_n|$ in terms of the some polynomials $\{\bar{P}_\av\}_{\av\in D}$ as
\begin{equation}
|b_n(t; \{ t_\av \}, \{n_\cv\} ) | 
\le \inf_{\Ub_\av \ge 1, \forall \av \in D} \Big( \prod_{\av \in D} | 2 e Q_\av / \Ub_\av |^{t_\av} \Big) \exp\Big[  \sum_{\av \in D} t_\av \bar{P}_\av ( \{ \Ub_\bv \}_{\bv \prec \av})  \Big] \prod_{\bv \in D} \Ub_\bv^{s_\bv}.
\end{equation}
Plugging this back into \eqref{eq:en-bound1}, we have 
\begin{align}
| e_n(t)| \le&~ \sum_{\sum_{\av} t_\av = t} \inf_{\Ub_\av \ge 1, \forall \av \in D} \Big( \prod_{\av \in D} |2 e Q_\av / \Ub_\av |^{t_\av} \Big) \exp\Big\{ \sum_{\av \in D} t_\av \bar{P}_\av ( \{ \Ub_\bv \}_{\bv \prec \av}) \Big\} \prod_{\bv \in D} \Ub_\bv^{s_\bv}.
\end{align}
Now since $\sum_{s} \inf_\eta f(s, \eta) \le \inf_{\eta} \sum_s f(s, \eta) $, we have
\begin{align}
| e_n(t)| &\le \inf_{\Ub_\av \ge 1, \forall \av \in D} \sum_{\sum_{\av} t_\av = t}  \Big( \prod_{\av \in D} |2 e Q_\av / \Ub_\av |^{t_\av} \Big) \exp\Big\{  \sum_{\av \in D} t_\av \bar{P}_\av ( \{ \Ub_\bv \}_{\bv \prec \av}) \Big\} \prod_{\bv \in D} \Ub_\bv^{s_\bv}, 
\end{align}
where the right hand side gives $E(t)$ as defined in Eq. (\ref{eq:E-def}). This proves the lemma.

\subsubsection{Proof of Lemma \ref{lem:e_summation}: $e(t)$ sums to $I$}

Our goal here is to prove
$\sum_t e(t) = I$.
Recall the definition \eqref{eq:e-def} of $e(t)$, which we reproduce here for convenience:
\begin{equation}
\begin{aligned}
e(t) =&~
 \sum_{\substack{t_\av \ge 0, \forall \av \in D \\ \sum_{\av \in D} t_\av = t}} 
 \Big[\prod_{\av \in D}  \frac{1}{2\pi i}\oint_{\Disk} \frac{d z_\av}{z_\av} z_\av^{-t_\av}  (2 Q_\av )^{t_\av}\Big] 
 \exp\Big[ \cC[P_{\lin}](\{t_\av\}_{\av\in D}, \{z_\bv\}_{\bv\in D}, \{Q_\cv\}_{\cv\in A_0}) \Big] \\
 &\qquad \qquad \qquad\qquad\qquad \times 
 \prod_{\av, \bv \in D, \cv\in A_0}
    1\{r_\av = 0\}
    z_\bv^{s_\bv}
    Q_\cv^{m_\cv}.
\end{aligned}
\end{equation}
Also recall definition \eqref{eqn:I-definition} of $I$:
\begin{equation}
I(\{r_\av\}, \{ s_\bv\}, \{ m_\cv\}) = \prod_{\av, \bv \in D, \cv\in A_0} 1(r_\av = 0) (2 W_{\bv})^{s_\bv}Q_\cv^{m_\cv}, 
\end{equation}
where $\{ W_\bv \}_{\bv \in D}$ is given as the unique solution to the following equation (c.f. Lemma \ref{lem:unique_solution_SCE} for the existence and uniqueness of the solution):
\begin{align}
    W_{\xv} =  Q_{\xv} \exp \Big[ \partial_{\tau_\xv} \cC[P_\lin](\{\tau_\av = 0\}_{\av \in D}, \{2W_\bv\}_{\bv \in D},\{Q_\cv\}_{\cv \in A_0}) \Big],~~~~ \forall \xv \in D. 
\end{align}

We now state a lemma that we will recursively apply to prove Lemma~\ref{lem:e_summation}. 
\begin{lemma}\label{lem:simplification_lemma_new}
Let $f: \C \to \C$ be a fixed polynomial. For any $s \in \Z_{\ge 0}$, $Y \in \C$, we have
\begin{equation}  \label{eq:sum-and-integral-ident}
\sum_{t = 0}^\infty \frac{1}{2 \pi i}\oint_{\mathbb D}  \exp[f(z) + Y t ] (2 Q)^t z^{s - t} \frac{d z}{z} = (2W)^s \exp[f(2 W)],
\end{equation}
where $W = Q \exp(Y)$ and $\Disk$ is the unit circle in the complex plane. 
\end{lemma}

\begin{proof}[Proof of Lemma \ref{lem:simplification_lemma_new}]
First, let us rewrite the LHS of \eqref{eq:sum-and-integral-ident} using $W = Q \exp(Y)$ as
\begin{equation}
\LHS =\sum_{t = 0}^\infty \frac{1}{2 \pi i}\oint_{\Disk}  \exp[f(z) ]  (2 W / z)^t  z^s \frac{d z}{z}. 
\end{equation}
We want to exchange the order of operation so that we first do the infinite sum on $t$, but this sum only converges if $|z|>|2W|$.
To do this, we choose any $K > |2 W| = |2 Q \exp(Y) |$, and denote $K \Disk$ as the circle with radius $K$ centered at 0 in the complex plane.
Since the only pole in the integrand above occurs at $z=0$, then by the Cauchy integral theorem, we have
\begin{align}
\LHS &= \sum_{t = 0}^\infty \frac{1}{2 \pi i}\oint_{K \Disk}  \exp[f(z) ]  (2 W / z)^t z^s \frac{d z}{z} = \frac{1}{2 \pi i}\oint_{K \Disk}  \exp[f(z) ]  \frac{z^s}{z-2W} dz \nonumber \\
&= (2W)^s \exp[ f(2 W)]
\end{align}
where the infinite sum on $t$ now converges since $|z|=K > |2W|$, and applying the Cauchy integral formula yields the last equality.
\end{proof}

Now we are ready to prove Lemma~\ref{lem:e_summation}.
We want to evaluate the following explicitly
\begin{align}
 \sum_{t=0}^\infty e(t) &=  \prod_{\av \in D} \Big[ \sum_{t_\av=0}^\infty \frac{1}{2\pi i}\oint_{\Disk} \frac{d z_\av}{z_\av} z_\av^{-t_\av}  (2 Q_\av )^{t_\av}\Big] 
 \exp\Big[ \cC[P_{\lin}](\{t_\av\}_{\av\in D}, \{z_\bv\}_{\bv\in D}, \{Q_\cv\}_{\cv\in A_0}) \Big] \nonumber \\
 &\qquad \qquad \qquad\qquad\qquad \times 
 \prod_{\av, \bv \in D, \cv\in A_0}
    1\{r_\av = 0\}
    z_\bv^{s_\bv}
    Q_\cv^{m_\cv} .
\end{align}
where
\begin{align}
    \cC[P_\lin] (\{t_\av\}, \{z_\bv\}, \{Q_\cv\})
    =  \sum_{\av\in D} t_\av P_\av(\{z_\bv: \bv \prec \av\}, \{Q_\cv\}),
\end{align}
and
\begin{align}
    P_\xv(\{z_\bv\}, \{Q_\cv\}) = \partial_{\tau_\xv} \cC[P](\{\tau_\av=0\}, \{z_\bv\}, \{Q_\cv\}).
\end{align}
To do this, we note by the assumption that $P$ is a well-played polynomial that $P_\av$ does not depend on $z_\bv$ when $\bv \succeq \av$. So we can restrict its argument to $\{z_\bv: \bv \prec \av\}$.

Let us label the elements of $D=\{1,2,\ldots, |D|\}$ according to the increasing order defined on the set $D$ as in Definition~\ref{def:well-played-polynomial-redefine}.
Note due to this ordering, we can perform the sums and integrals of the form  $\sum_{t_\xv} \oint d z_\xv(\cdots)$ sequentially in the ordering of $\xv \in D$.
More precisely, this can be seen via the following.
Let us define for $K = 1,2,\ldots, |D|+1$ the following intermediate expression:
\begin{align}
    S_K &:= \prod_{\av \in D, \cv\in A_0} 1\{r_\av = 0\} Q_\cv^{m_\cv}\prod_{j \preceq K-1} (2{W}_j)^{s_j}
     \nonumber \\
 & \qquad \qquad \times 
     \prod_{j \succeq K}\bigg[ \sum_{t_j=0}^\infty  \frac{1}{2 \pi i} \oint_{\Disk} \frac{d z_j}{z_j}z_j^{s_j-t_j} (2 Q_j)^{t_j} \bigg]  
 \exp\big[G_K(z_K,\ldots, z_{|D|}) \big],
\end{align}
where
\begin{align}
    G_K := t_K \underbrace{P_K(\{z_j = 2W_j\}_{j\prec K}, \{Q_\cv\})}_{Y_K}
        + \underbrace{
            \sum_{L= K+1}^{|D|} t_L P_L(\{z_j\}_{j \prec L}, \{Q_\cv\})
            \Big|_{z_{j}= 2{W}_j ~\forall j \prec K}
            }_{f_K(z_K,\ldots, z_{|D|})}
\end{align}
and 
\begin{align}
    W_K &= Q_K \exp(Y_K)  
    = Q_K \exp[P_j(\{z_j = 2W_j\}_{j\prec K}, \{Q_\cv\})] \nonumber \\
    &= Q_K \exp[\partial_{\tau_K} P(\{0\}, \{z_j = 2W_j\}_{j\prec K}, \{Q_\cv\})].
    \label{eq:WK}
\end{align}
We will show inductively the following line of equalities:
\begin{align}
    \sum_{t=0}^\infty e(t) =S_1 = S_2 = \cdots = S_{|D|} = S_{|D|+1} = \prod_{\av, \bv \in D, \cv\in A_0} 1(r_\av = 0) (2 W_{\bv})^{s_\bv}Q_\cv^{m_\cv} = I,
\end{align}
which would prove the lemma.

Indeed, one can easily check that the base case $\sum_{t=0}^\infty e(t)=S_1$ is true since $G_1 = P_\lin$ from its definition.
Next, we will show that $S_K = S_{K+1}$ for any $K = 1, \ldots, \vert D\vert$.
Using the expression of $G_K(z_K,\ldots, z_{|D|})=t_K Y_K + f_K(z_K,\ldots, z_{|D|})$ above and applying Lemma~\ref{lem:simplification_lemma_new}, we get
\begin{align}
    \sum_{t_K=0}^\infty  \frac{1}{2 \pi i} \oint_{\Disk} \frac{d z_K}{z_K}z_K^{s_K-t_K} (2 Q_K)^{t_K} 
 \exp\big[G_K(\{z_j: j \succeq K\}) \big]
 &= (2{W}_K)^{s_K} \exp[f_K(z_K = 2{W}_K)]
\end{align}
with $W_K$ given in Eq.~\eqref{eq:WK}.
And since $f_K(z_K=2{W}_K, z_{K+1},\ldots) = G_{K+1}(z_{K+1},\ldots)$, we indeed have 
\begin{align}
    S_K &= \prod_{\av \in D, \cv\in A_0} 1\{r_\av = 0\} Q_\cv^{m_\cv} \prod_{j \preceq K} (2{W}_j)^{s_j}\prod_{j \succeq K+1}\bigg[ \sum_{t_j=0}^\infty  \frac{1}{2 \pi i} \oint_{\Disk} \frac{d z_j}{z_j}z_j^{s_j-t_j} (2 Q_j)^{t_j} \bigg]  
 \exp\big[G_{K+1} \big]  \nonumber \\
 &= S_{K+1}.
\end{align}
This completes the proof.

\subsubsection{Proof of Lemma \ref{lem:E_summable}: the sum on $E(t)$ converges}

Our goal is to show that the series $\sum_{t = 0}^\infty E(t)$ converges. Recall \eqref{eq:E-def} where $E(t)$ is defined as
\begin{align}
E(t) &=   \inf_{\Ub_\av \ge 1, \forall \av \in D} \sum_{\sum_{\av} t_\av = t}  \prod_{\av \in D} \Big( |2 e Q_\av / \Ub_\av |^{t_\av} \Ub_\bv^{s_\bv}\Big) \times \exp\Big[ \sum_{\av \in D} t_\av \bar P_\av(\{\Ub_\bv\}_{\bv \prec \av}) \Big].
\end{align}

We first state a lemma that is recursively used in the proof. 

\begin{lemma}\label{lem:upper_bound_recursion}
For any $X \ge 0$, integer $s\in \mathbb{Z}_{\ge 0}$, and polynomial $g$ with non-negative coefficients, we have
\begin{align}
\inf_{\eta \ge 1} \sum_{t = 0}^\infty ( X / \eta)^t \eta^s \exp[ g(\eta) ] \le 2  (2 X + 1)^s \exp[ g(2 X + 1)]. 
\end{align}
\end{lemma}

\begin{proof}[Proof of Lemma \ref{lem:upper_bound_recursion}]
The lemma follows by taking $\eta = 2 X + 1$ and note that $\sum_{t \ge 0} (X/(2X+1))^t \le 2$. 
\end{proof}

To prove Lemma \ref{lem:E_summable}, we just need to recursively apply Lemma \ref{lem:upper_bound_recursion}. The recursion process is similar to the proof of Lemma \ref{lem:e_summation} above.
First, we define (which will be shown to be finite)
\begin{equation}
R = \inf_{\Ub_\av \ge 1, \forall \av \in D} 
\sum_{t_\av \ge 0, \forall \av \in D}
   \prod_{\av \in D} \Big(|2 e Q_\av / \Ub_\av |^{t_\av} \Ub_\av^{s_\av}\Big) \exp\Big\{\sum_{\av \in D} t_\av \bar P_\av(\{\Ub_\bv\}_{\bv \prec \av}) \Big\}. 
\end{equation}
Since $\sum_{s} \inf_\eta f(s, \eta) \le \inf_{\eta} \sum_s f(s, \eta) $, we bave
\begin{equation}
    \sum_{t=0}^\infty E(t) \le R.
\end{equation}
We further define non-negative real numbers $\tilde W_\bv \ge 0$ recursively in the ascending order of $\prec$ (defining $\tilde W_\bv$ before $\tilde W_\av$ if $\bv \prec \av$) as follows: For any $\av\in D$, we define
\begin{equation}\label{eqn:recursive_tilde_W}
\tilde W_{\av} = 2 e |Q_\av|\exp\Big\{  \bar{P}_\av(\{ 2 \tilde W_{{\bv}} + 1 \}_{\bv \prec \av})  \Big\}.
\end{equation}
In what follows, we will show that 
\begin{equation}
R \le \prod_{\av \in D} 2 (2\tilde W_\av + 1)^{s_\av} < \infty    
\end{equation}
which is sufficient to prove Lemma~\ref{lem:E_summable}.

To accomplish this, let us begin by labelling the elements of $D=\{1,2,\ldots, |D|\}$ in the ascending order of $\prec$ defined on the set $D$ as mentioned in Definition~\ref{def:well-played-polynomial-redefine}. 
Let us define for $k = 1,2,\ldots, |D|+1$ the following intermediate expression:
\begin{equation}
   R_k := \prod_{j \preceq k-1} 2 (2\tilde{W}_j + 1)^{s_j} \times \inf_{\Ub_l \ge 1, \forall l \succeq k} \prod_{l \succeq k}\bigg[  \sum_{t_l=0}^\infty \Big(  |2 e Q_l / \Ub_l |^{t_l} \Ub_l^{s_l} \Big) \bigg] \times  
 \exp\big[ H_k(\Ub_k,\ldots, \Ub_{|D|}) \big]
\end{equation}
where
\begin{equation}
    H_k = t_k \underbrace{ \bar{P}_k(\{ 2 \tilde W_j + 1 \}_{j \prec k}) }_{Z_k}
        + \underbrace{\sum_{l= k+1}^{|D|} t_l \bar{P}_l(\{ \Ub_j \}_{j \prec l}) \Big|_{\Ub_{j}= (2\tilde{W}_j + 1), \forall j \prec k} }_{g_k(\Ub_k,\ldots, \Ub_{|D|})}.
\end{equation}
We will show the following line of inequalities:
\begin{equation}
    R = R_1 \le  R_2 \le \cdots \le R_{|D|} \le R_{|D|+1} = \prod_{j\in D} 2 (2\tilde{W}_j + 1)^{s_j},
\end{equation}
which proves the lemma.

Indeed, one can easily check that the base case $R= R_1$ is true from its definition.
Next, we want to show $R_k \le R_{k+1}$ for each $k = 1, 2, \ldots, | D |$.
Using the expression of $H_k(\Ub_k,\ldots, \Ub_{|D|})=t_k Z_k + g_k(\Ub_k,\ldots, \Ub_{|D|})$ above and applying Lemma~\ref{lem:upper_bound_recursion}, 
we get
\begin{equation}
    \inf_{\Ub_k \ge 1} \sum_{t_k=0}^\infty  |2 e Q_k / \Ub_k |^{t_k}  \Ub_k^{s_k}
 \exp\big[ H_k(\{\Ub_j: j \succeq k\}) \big]
 \le 2 (2\tilde{W}_k + 1)^{s_k} \exp[ g_k(\Ub_k = 2\tilde{W}_k + 1, \{ \Ub_j \}_{j \succ k})]
\end{equation}
where the role of $X$ in Lemma~\ref{lem:upper_bound_recursion} is played by
\begin{equation}
\tilde{W}_k = 2 e |Q_k|\exp( Z_k) 
    = 2 e |Q_k |\exp\Big[\bar{P}_k(\{ 2 \tilde W_j + 1 \}_{j \prec k}) \Big]
\end{equation}
which agrees with \eqref{eqn:recursive_tilde_W}.
And since $g_k(\Ub_k=2\tilde{W}_k + 1, \{ \Ub_j\}_{j \ge k+1}) = H_{k+1}(\{ \Ub_j\}_{j \ge k+1})$, we indeed have 
\begin{align}
   R_k \le \prod_{j \preceq k} 2 (2\tilde{W}_j + 1)^{s_j} \times \inf_{\Ub_l \ge 1, \forall l \succeq k + 1} \prod_{l \succeq k+1}\bigg[ 
    \sum_{t_l=0}^\infty \Big( |2 e Q_l / \Ub_l |^{t_l}  \Ub_l^{s_l}\Big)
   \bigg]  
 \exp\big[ H_{k+1} \big] 
 = R_{k+1}.
\end{align}
This concludes the proof of Lemma~\ref{lem:E_summable}.

\subsubsection{Auxiliary lemmas}
\label{sec:auxiliary}
In this subsection, we state and prove three technical lemmas that are used in the proofs of Lemma~\ref{lem:en_convergence_e} and \ref{lem:en_bounded_E} above.

The first two lemmas below study the property of the operator $\SS_n^{\{ t_\av \}}$ as defined in Eq.~\eqref{eqn:SSn-definition}.
We first work in Lemma \ref{lem:little_sum_lemma} with the quantity $S_n$ which can be viewed as $\SS_n^{\{t_\av\}}$ acting on a monomial of form $(d_\av / n)^\pois$ for a single $\av \in D$.
This immediately gives Lemma \ref{lem:little_sum_lemma_multi}, which concerns the operator $\SS_n^{\{ t_\av \}}$ acting on monomials of form $\prod_{\av \in D} ( d_\av / n)^{\pois_\av}$ and is used in the proof of Lemma~\ref{lem:en_convergence_e}.

\begin{lemma}\label{lem:little_sum_lemma}
Let $t, \pois \in \Z_{\ge 0}$, $n \in \Z_{> 0}$, and $Q \in \C$. Denote 
\begin{equation}
S_{n}(Q, t, \pois) := n^{t} \sum_{n_0 +n_1= t} \binom{t}{n_0, n_1} Q^{n_0} (- Q)^{n_1} \Big(\frac{ n_0 - n_1}{n}\Big)^\pois. 
\end{equation}
Then we have 
\begin{equation}\label{eqn:little_sum_monomial_equivalent_formula}
S_{n}(Q, t, \pois) = (2Q n)^{t} \frac{\partial_\alpha^\pois [(\sinh(\alpha))^t] |_{\alpha = 0} }{n^\pois} = (2 Q )^t \frac{\pois!}{2 \pi i} \oint_{\Disk} \frac{(n \sinh(z/n) )^t}{z^\pois} \frac{d z}{z}. 
\end{equation}
Furthermore, we have 
\begin{equation}\label{eqn:little_sum_monomial_limit}
\lim_{n \to \infty} S_{n}(Q, t, \pois) = (2Q)^{t} t! 1_{\pois = t} = \frac{t!}{2 \pi i}\oint_{\Disk}  z^\pois (2 Q / z)^t \frac{d z}{z}, 
\end{equation}
and for all $\pois \in \Z_{\ge 0}$, $n \in \Z_{> 0}$, $0 \le t \le n$, and $Q \in \C$, we have
\begin{equation}\label{eqn:bound-Sn}
|S_n(Q, t, \pois) |\le |2Q |^t \frac{t^\pois }{ n^{\pois - t}} \cdot 1_{\pois \ge t} \le |2Q |^t t^t 1_{\pois \ge t}. 
\end{equation}
\end{lemma}

\begin{proof}[Proof of Lemma \ref{lem:little_sum_lemma}]

By the definition of $S_n$, we have 
\begin{equation}
\begin{aligned}
S_{n}(Q, t, \pois) =&~ n^{t} \sum_{n_0 = 0}^t \binom{t}{n_0} Q^{n_0} (- Q)^{t - n_0} \Big(\frac{2 n_0 - t}{n}\Big)^\pois \\
=&~ n^{t - \pois} \sum_{n_0 = 0}^t \binom{t}{n_0} Q^{n_0} (- Q)^{t - n_0} \partial_\alpha^\pois \exp( \alpha (2 n_0 - t )) |_{\alpha = 0} \\
=&~ n^{t - \pois} \partial_\alpha^\pois \Big[ \sum_{n_0 = 0}^t \binom{t}{n_0} Q^{n_0} (- Q)^{t - n_0}  \exp( \alpha (2 n_0 - t )) \Big] |_{\alpha = 0}\\
=&~ n^{t - \pois} \partial_\alpha^\pois \Big[ \sum_{n_0 = 0}^t \binom{t}{n_0} (Q e^\alpha)^{n_0} (- Q e^{-\alpha})^{t - n_0} \Big] |_{\alpha = 0}\\
=&~ n^{t - \pois} \partial_\alpha^\pois (Q e^\alpha - Qe^{-\alpha})^{t} |_{\alpha = 0} = (2 Q n)^t \frac{\partial_\alpha^\pois (\sinh \alpha)^{t} |_{\alpha = 0}}{n^\pois}. 
\end{aligned}
\end{equation}
Now, the Cauchy's integral formula allows us to write the $n$-th derivative of a function $f$ as
\begin{equation}
    f^{(n)}(0) = \frac{n!}{2\pi i} \oint_{\Disk} \frac{f(z)}{z^{n+1}} dz
\end{equation}
where $\Disk$ is the unit circle in the complex plane. Applying this to $S_n(Q,t,\pois)$ yield 
\begin{align}
S_{n}(Q, t, \pois) 
=&~  (2 Q )^t \frac{\pois!}{2 \pi i} \oint_{\Disk} \frac{(n \sinh z)^t}{(nz)^{\pois+1}} d (n z) =  (2 Q )^t \frac{\pois!}{2 \pi i} \oint_{n\Disk} \frac{(n \sinh z)^t}{(nz)^{\pois+1}} d (n z) \nonumber \\
=&~ (2 Q )^t \frac{\pois!}{2 \pi i} \oint_{\Disk} \frac{(n \sinh(z/n) )^t}{z^\pois} \frac{d z}{z}. 
\end{align}
This proves Eq. (\ref{eqn:little_sum_monomial_equivalent_formula}). Furthermore, note that we have 
\[
\lim_{n \to \infty} \sup_{z \in \Disk} \Big|(n \sinh(z / n))^t - z^t \Big|= 0. 
\]
This gives 
\[
\begin{aligned}
\lim_{n \to \infty}S_{n}(Q, t, \pois) =&~ \lim_{n \to \infty} (2 Q )^t \frac{t!}{2 \pi i} \oint_{\Disk} \frac{(n \sinh(z/n) )^t}{z^\pois} \frac{d z}{z} \\
=&~ (2 Q )^t \frac{t!}{2 \pi i} \oint_{\Disk} z^{t - \pois} \frac{d z}{z} = (2 Q)^t t! 1_{\pois = t} =  \frac{t!}{2 \pi i} \oint_{\Disk} z^\pois (2 Q / z)^t \frac{d z}{z}.
\end{aligned}
\]
This proves Eq. (\ref{eqn:little_sum_monomial_limit}). 

Finally, note that for $\pois < t$, we have $\partial_\alpha^\pois [(\sinh(\alpha))^t]  |_{\alpha = 0} = 0$, and for $\pois \ge t$, we have
\[
\begin{aligned}
&~\Big|\partial_\alpha^\pois [(\sinh(\alpha))^t]  |_{\alpha = 0} \Big|= \Big|\partial_\alpha^\pois \Big[ \sum_{s = 0}^t {t \choose s} ((1/2) e^{\alpha})^s (- (1/2) e^{- \alpha })^{t-s}\Big]  |_{\alpha = 0} \Big|\\
\le&~ \Big|\partial_\alpha^\pois \Big[ \sum_{s = 0}^t {t \choose s} ((1/2) e^{\alpha })^s ( (1/2) e^{\alpha })^{t-s}\Big]  |_{\alpha = 0} \Big|= \partial_\alpha^\pois e^{\alpha t} |_{\alpha = 0} = t^\pois.
\end{aligned}
\]
As a consequence, we have 
\[
|S_n(Q, t, \pois)|\le |2 Q n |^t  (t / n)^\pois 1_{\pois \ge t} \le |2 Q |^t t^t 1_{\pois \ge t},
\]
where the last inequality used $t \le n$. This concludes the proof. 
\end{proof}

By Lemma \ref{lem:little_sum_lemma}, we immediately have the following lemma which studies the property of $\SS_n^{\{ t_\av \}}$ acting on monomials. 
\begin{lemma}\label{lem:little_sum_lemma_multi}
Consider the $\SS_n^{\{ t_\av\}}$ operator as defined in Eq. (\ref{eqn:SSn-definition}), where $t \in \Z_{\ge 0}$, $\sum_{\av \in D} t_\av = t$, $t \le n \in \Z_{> 0}$, and $\{ Q_\av \}_{\av \in D} \subseteq \C$. Consider a monomial 
\[
g(\{ d_\av /n \}; \{ \pois_\av \}) = \prod_{\av \in D} \Big( \frac{d_\av}{n} \Big)^{\pois_\av}.
\]
Then we have
\begin{equation}\label{eqn:little_sum_monomial_equivalent_formula_multi}
\SS_n^{\{ t_\av\}} g = \prod_{\av}(2Q_\av n)^{t_\av} \frac{\partial_\alpha^{\pois_\av} [(\sinh(\alpha_\av))^{t_\av}] |_{\alpha_\av = 0} }{n^\pois} = \oint_{\Disk^{|D|}} \prod_{\av \in D} (2 Q_\av )^{t_\av} \frac{\pois_\av!}{2 \pi i}  \frac{(n \sinh(z_\av/n) )^{t_\av}}{z_\av^{\pois_\av}} \frac{d z_\av}{z_\av}.
\end{equation}
Furthermore, we have 
\begin{equation}\label{eqn:little_sum_monomial_limit_multi}
\lim_{n \to \infty} \SS_n^{\{ t_\av\}} g = \prod_{\av \in D} (2Q_\av)^{t_\av} t_\av! 1_{\pois_\av = t_\av} = \SS^{\{ t_\av\}} g,
\end{equation}
where
\begin{equation}
\SS^{\{ t_\av\}} g = \frac{\prod_{\av \in D} t_\av!}{(2 \pi i)^{|D |}}\oint_{\Disk^{|D|}} g(\{ d_\av / n\}) \prod_{\av \in D} (2 Q_\av / z_\av)^{t_\av} \frac{d z_\av}{z_\av}.
\end{equation}
Finally, we have
\begin{equation}\label{eqn:bound-Sn_multi}
|\SS_n^{\{ t_\av\}} g |\le \prod_{\av \in D} |2Q_\av |^t \frac{t_\av^{\pois_\av} }{ n^{\pois_\av - t_\av}} \cdot 1_{\pois_\av \ge t_\av} \le \prod_{\av \in D} |2Q_\av |^{t_\av} t_\av^{t_\av} 1_{\pois_\av \ge t_\av}. 
\end{equation}
\end{lemma}

We next state a lemma which is used in the proof of Lemma \ref{lem:en_bounded_E} above. 

\begin{lemma} \label{lem:upper_bound_little_sum_exponential}
For any complex coefficient monomials $M_0, M_1, \ldots, M_K: \R^D \to \C$, denote
\[
b_n(\{ t_\av\}_{\av \in D}) = \Big(\prod_{\av \in D} \frac{n^{t_\av}}{t_\av !}\Big) \Big(\prod_{\av \in D}\sqint_\av \Big) \exp\Big\{ \sum_{k = 1}^K M_k( \{ d_\bv / n \}_{\bv \in D}) \Big\} M_0(\{ d_\bv / n \}_{\bv \in D}). 
\]
Further assume that $t_\av \le n$ for any $\av \in D$. Then we have
\[
\Big|b_n(\{ t_\av\}_{\av \in D}) \Big|\le \inf_{\Ub_\bv \ge 1, \forall \bv \in D} \Big( \prod_{\av \in D} |2 e Q_\av / \Ub_\av |^{t_\av} \Big) \exp\Big\{ \sum_{k = 1}^K |M_k( \{ \Ub_\bv \}_{\bv \in D} ) | \Big\} |M_0( \{ \Ub_\bv \}_{\bv \in D}) |. 
\]
\end{lemma}

\begin{proof}
Let us illustrate the proof with $D = \{1, 2\}$, $K = 2$, $M_k(u_1, u_2) = m_k u_1^{s_{k1}} u_2^{s_{k2}}$ for $k = 0, 1, 2$. It is easy to see that the proof will hold for general $D$ and $K$. 

In this case, we have 
\[
\begin{aligned}
\Big|b_n(\{ t_\av\}_{\av \in D}) \Big|=&~ \Big|\Big(\prod_{z \in [2]} \frac{n^{t_z}}{t_z !}\Big) \Big(\prod_{z \in [2]}\sqint_z \Big) \sum_{\pois_1 \ge 0}\frac{1}{\pois_1!} \Big( m_1 \Big(\frac{d_1}{n}\Big)^{s_{11}} \Big(\frac{d_2}{n}\Big)^{s_{12}} \Big)^{\pois_1} \\
&~ \times \sum_{\pois_2 \ge 0}\frac{1}{\pois_2!} \Big( m_2 \Big(\frac{d_1}{n}\Big)^{s_{21}} \Big(\frac{d_2}{n}\Big)^{s_{22}} \Big)^{\pois_2} m_0 \Big( \frac{d_1}{n}\Big)^{s_{01}} \Big( \frac{d_2}{n}\Big)^{s_{02}} \Big|\\
=&~ \Big|\Big(\prod_{z \in [2]} \frac{1}{t_z !}\Big) \sum_{\pois_1 \ge 0}\frac{|m_1 |^{\pois_1}}{\pois_1!} \sum_{\pois_2 \ge 0}\frac{|m_2 |^{\pois_2}}{\pois_2!} \Big(\prod_{z \in [2]}n^{t_z} \sqint_z \Big( \frac{d_z}{n} \Big)^{\pois_1 s_{1z} + \pois_2 s_{2z} + s_{0z}} \Big) |m_0 |\Big|.  \\
\end{aligned}
\]
Now we use Lemma \ref{lem:little_sum_lemma} so that $\vert n^{t_z} \sqint_z (d_z / n)^\pois \vert \le |2 Q_z |^{t_z} t_z^{t_z} 1_{\pois \ge t_z}$ (when $t_z \le n$). Then we have 
\[
\begin{aligned}
\Big|b_n(\{ t_\av\}_{\av \in D}) \Big|\le&~ \Big(\prod_{z \in [2]} \frac{1}{t_z !}\Big) \sum_{\pois_1 \ge 0}\frac{|m_1 |^{\pois_1}}{\pois_1!} \sum_{\pois_2 \ge 0}\frac{|m_2 |^{\pois_2}}{\pois_2!} \Big(\prod_{z \in [2]} |2 Q_z |^{t_z} t_z^{t_z} 1_{\pois_1 s_{1z} + \pois_2 s_{2z} + s_{0z} \ge t_z} \Big) |m_0 |  \\
\le&~ \inf_{\Ub_1, \Ub_2 \ge 1}  \sum_{\pois_1 \ge 0}\frac{|m_1 |^{\pois_1}}{\pois_1!} \sum_{\pois_2 \ge 0}\frac{|m_2 |^{\pois_2}}{\pois_2!} \Big(\prod_{z \in [2]} |2 e Q_z / \Ub_z |^{t_z}  \Ub_z^{\pois_1 s_{1z} + \pois_2 s_{2z} + s_{0z}} \Big) |m_0 |\\
=&~ \inf_{\Ub_1, \Ub_2 \ge 1} \Big(\prod_{z \in [2]} |2 e Q_z / \Ub_z |^{t_z} \Big) \sum_{\pois_1 \ge 0}\frac{|m_1 \Ub_1^{s_{11}} \Ub_2^{s_{12}} |^{\pois_1}}{\pois_1!} \sum_{\pois_2 \ge 0}\frac{|m_2 \Ub_1^{s_{21}} \Ub_2^{s_{22}} |^{\pois_2}}{\pois_2!}   |m_0  \Ub_1^{s_{01}} \Ub_2^{s_{0 2}} |\\
=&~ \inf_{\Ub_1, \Ub_2 \ge 1} \Big(\prod_{z \in [2]} |2 e Q_z / \Ub_z |^{t_z} \Big) \exp\Big\{ |M_1(\Ub_1, \Ub_2) |+ |M_2(\Ub_1, \Ub_2) |\Big\} |M_0(\Ub_1, \Ub_2) |, \\
\end{aligned}
\]
where the second inequality used the fact that $t^t / t! \le e^t$ and $1_{\pois \ge t} \le \eta^\pois / \eta^t$ for any $\eta \ge 1$. This proves the lemma. 
\end{proof}


\section{Remark: A non-rigorous ``proof'' using the saddle-point method}
\label{sec:nonrigorous}

As an aside, we informally describe a potential way to show our generalized multinomial theorem, Proposition~\hyperref[thm:gen-multinomial-restate]{4.1 (Formal)}, using a non-rigorous application of the saddle-point method.
This provides a much quicker way to get what ultimately turns out to be the correct answer, and perhaps will provide some intuition.
We consider it an interesting open challenge to make this saddle-point method approach rigorous.

Recall that our goal is to evaluate the following quantity in the $n\to\infty$ limit:
\begin{equation} 
\label{eq:F-saddle}
\mathfrak{F} := 
\sum_{\{n_\av\}} \binom{n}{\{n_\av\}}
\Big(\prod_{\bv\in A} Q_\bv^{n_\bv} \Big)
	\exp\Big[n P_n(\{n_\av/n\}) \Big]
f_n(\{n_\av/n\}).
\end{equation}
Let $\omega_\av = n_\av / n$ for all $\av\in A$.
Then
\begin{equation}
    \prod_{\av\in A} Q_\av^{n_\av} = \exp\Big[\sum_{\av\in A} n\omega_\av \log Q_\av\Big].
\end{equation}
And using Stirling's approximation and the fact that $\sum_{\av \in A} n_\av = n$, we have (crudely)
\begin{equation} \label{eq:stirling}
    \binom{n}{\{n_\av\}} \approx \frac{n^n}{\prod_{\av\in A} n_\av ^{n_\av}} = \exp\Big[{-} \sum_{\av \in A} n \omega_\av \log \omega_\av\Big].
\end{equation}
In the $n\to\infty$ limit, we approximate each $\omega_\av \in [0,1]$ as a continuous variable, and write \eqref{eq:F-saddle} as
\begin{equation} \label{eq:En-approx}
    \mathfrak{F} \approx \int \Big({\textstyle \prod_{\av \in A} d \omega_\av}\Big)  e^{n S(\{\omega_\av\})}  ~ f(\{\omega_\av\}),
\end{equation}
where
\begin{equation}
    S(\{\omega_\av\}) = -\sum_{\av\in A} \omega_\av \log \frac{\omega_\av }{ Q_\av}
    + P(\{\omega_\av\})
\end{equation}
with $P = \lim_{n\to\infty}P_n$ and $f = \lim_{n\to\infty} f$.
Suppose somehow we can apply the saddle-point approximation to \eqref{eq:En-approx}, then roughly
\begin{align}
    \mathfrak{F}  \approx  e^{n S(\{W_\av\})} ~ f(\{W_\av\})
\end{align}
where $\{W_\av\}$ is the saddle point of $S(\{\omega_\av\})$, subject to the constraint that $\sum_\av \omega_\av = \sum_\av n_\av / n = 1$.

To obtain the saddle point explicitly, we introduce a Lagrange multiplier for the constraint:
\begin{align}
    S_\lambda(\{\omega_\av\}) = S(\{\omega_\av\}) + \lambda \Big(\sum_{\av\in A} \omega_\av - 1\Big).
\end{align}
Then we look at stationary points of $S_\lambda$, given by
\begin{align}
    0=\frac{\partial S_\lambda}{\partial \omega_\av} = -1 - \log \frac{\omega_\av}{Q_\av} + \frac{\partial P}{\partial \omega_\av} + \lambda.
\end{align}
As a result, the saddle point of $S(\{\omega_\av\})$ is given as the solution to
\begin{equation} \label{eq:saddle-SCE}
    W_\av = \mathcal{N} Q_\av\exp\Big[ \partial P(\{W_\bv\}) / \partial \omega_\av\Big], \qquad \forall \av \in A,
\end{equation}
where $\mathcal{N}=e^{\lambda-1}$ is a normalization constant chosen to ensure that $\sum_\av W_\av=1$.
For a well-played polynomial $P$, 
we note this self-consistent equation \eqref{eq:saddle-SCE} has a unique solution if we choose $\mathcal{N}=1$, due to Lemma~\ref{lem:unique_solution_SCE} shown earlier.
This solution has an ``anti-symmetric'' property that $W_\av + W_{\bar\av}=0$ for all $\av\in D$ and $W_\av=Q_\av$ for all $\av\in A_0$.
And happily we have $
    \sum_{\av\in A} W_\av  =\sum_{\av\in A_0} Q_\av + \sum_{\av\in D} (W_\av + W_{\bar\av}) = 1$,
yielding the correct normalization.

\vspace{5pt}
It remains to show $\mathfrak{F} \approx f(\{W_\av\})$, which follows from showing $S(\{W_\av\})=0$. 
Consider the canonical form of 
$P(\{\omega_\av\})$ in \eqref{eq:poly_t_d_n_def_redefinition}, and plug in $\omega_\av = W_\av$ for all $\av\in A$.
Under the well-playedness condition \eqref{eqn:Psi-well-played-condition}, every term in $P$ will have a factor of $\tau_\av= \omega_\av + \omega_{\bar\av} = W_\av + W_{\bar\av}= 0 $ for some $\av\in D$.
This means $P(\{W_\av\})=0$.
Thus, 
\begin{equation} \label{eq:S-Wa}
    S(\{W_\av\}) = -\sum_{\xv\in A} W_\xv \log \frac{W_\xv}{Q_\xv} + P(\{W_\av\}) = - \sum_{\xv\in A} W_\xv \frac{\partial P(\{W_\av\})}{\partial\omega_\xv}.
\end{equation}
Furthermore, from the proof of Lemma~\ref{lem:unique_solution_SCE} (b) in Section~\ref{sec:unique_solution_SCE}, we know that the derivatives of $P(\{\omega_\av\})$ after setting $\omega_\av = W_\av$ satisfy
\begin{equation}
    \frac{\partial P(\{W_\av\})}{\partial \omega_\xv} = \frac{\partial P(\{W_\av\})}{\partial \omega_{\bar\xv}} \quad \forall \xv \in D,
    \qquad \text{and} \qquad
    \frac{\partial P(\{W_\av\})}{\partial \omega_\xv} = 0 \quad \forall \xv \in A_0.
\end{equation}
Plugging this into \eqref{eq:S-Wa} we get
\begin{align}
S(\{W_\av\}) &= - \sum_{\xv\in D} \Big(W_\xv\frac{\partial P(\{W_\av\})}{\partial \omega_\xv} + W_{\bar\xv} \frac{\partial P(\{W_\av\})}{\partial \omega_{\bar\xv}}\Big) \nonumber \\
&= - \sum_{\xv\in D} (W_\xv + W_{\bar\xv}) \frac{\partial P(\{W_\av\})}{\partial \omega_\xv} = 0.
\end{align}
So in the end we have
\begin{align}
    \mathfrak{F} \approx e^{n S(\{W_\av\})} f(\{W_\av\}) = f(\{W_\av\})
    \qquad \text{as} \qquad n\to\infty,
\end{align}
just as we claimed in Proposition~\hyperref[thm:gen-multinomial-restate]{4.1 (Formal)}.

We find it remarkable that such a na\"ive application of the saddle-point method gives a simple, albeit non-rigorous proof of our generalized multinomial theorem.
Besides the validity of the saddle-point approximation, it is also questionable whether the crude Stirling's approximation in Eq.~\eqref{eq:stirling} is valid; for example, it was found in \cite{farhi2019quantum} that only terms with $n_{+-}+n_{-+}=1$ contributes to the final value of the sum when calculating the energy achieved by the $p=1$ QAOA.
Presently, we have failed to find a way to make this proof idea rigorous directly, which we leave as an interesting problem on its own.


\section{Proof of Theorem~\ref{thm:moments} (Performance of the QAOA as $n\to\infty$)}
\label{sec:QAOA-derivation}

In this appendix, we prove Theorem~\ref{thm:moments}. We start by stating Lemma~\ref{lem:characteristic_function_derivative} and \ref{lem:P-is-well-played} which are used in the proof of Theorem~\ref{thm:moments}. The proofs of the lemmas can be found in Section~\ref{sec:QAOA-organized-sum} and \ref{sec:proof_lemma_well_played}, respectively. Section~\ref{sec:prelim_results_for_proof_thm_moments} goes over useful definitions and results needed in Section~\ref{sec:proof_lemma_well_played}.

We first reformulate the first and second moments of the performance of QAOA into combinatorial sum forms, given in the following lemma. 
\begin{lemma}\label{lem:characteristic_function_derivative}
Suppose $C_J = \sum_{q=1}^{\qmax} c_q \sum_{i_1,\ldots,i_q=1}^n J_{i_1,i_2,\ldots,i_q} z_{i_1} z_{i_2}\cdots z_{i_q}$ is a random COP drawn from an ensemble $\G$ that satisfies Assumption~\ref{assum:iidJ}. Recall that we have defined that $g_{q, n}(\lambda) = n^{q-1} \log \EV[e^{i \lambda J_{1,2,\ldots,q}}]$. Let $p \in \Z_{> 0}$ and fix parameters $(\paramv) \in \R^{2p}$. Let $A$ be as defined in Eq. (\ref{eq:set_A_def}). Let $\{ Q_\av \}_{\av \in A}$ be as defined in Eq. (\ref{eq:Qdef}). Let $\{ \Phi_{\av} \}_{\av \in A}$ be as defined in Eq.~\eqref{eq:Phi_def}. Then we have
\begin{align}
\EV_J[\bgbbraket{C_J/n}] =&~ \sum_{\{n_\av\}} \binom{n}{\{n_\av\}} \prod_{\av\in A } Q_\av^{n_\av}
\exp\Big[ n \sum_{q=1}^\qmax \sum_{\av_1,\ldots,\av_q \in A} g_{q,n}\big(c_q \Phi_{\av_1\cdots\av_q}\big)\frac{n_{\av_1} \cdots n_{\av_q}}{n^q}
	\Big] \nonumber \\
& \qquad\quad\times\Big(-\sum_{q=1}^\qmax ic_q \sum_{\bv_1,\ldots,\bv_q \in A} g_{q,n}'(c_q\Phi_{\bv_1\cdots \bv_q})\frac{n_{\bv_1} \cdots n_{\bv_q}}{n^q}\Big), 
\label{eq:En_1}
\end{align}
and
\begin{align}
\EV_J[\braket{\paramv | (C_J/n)^2 | \paramv}] =&~ \sum_{\{n_\av\}} \binom{n}{\{n_\av\}} \prod_{\av\in A } Q_\av^{n_\av}
\exp\Big[
		n \sum_{q=1}^\qmax \sum_{\av_1,\ldots,\av_q \in A} g_{q,n}\big(c_q \Phi_{\av_1\cdots\av_q}\big)\frac{n_{\av_1} \cdots n_{\av_q}}{n^q}
	\Big] \nonumber \\
& \qquad \quad \times \Bigg[ \Big(-\sum_{q=1}^\qmax i c_q 
\sum_{\bv_1,\ldots,\bv_q \in A} g_{q,n}'(c_q\Phi_{\bv_1\cdots \bv_q})
\frac{n_{\bv_1} \cdots n_{\bv_q}}{n^{q}}\Big)^2 \nonumber \\
&\qquad\quad \quad + \Big(- \frac{1}{n}\sum_{q=1}^\qmax c_q^2 \sum_{\bv_1,\ldots,\bv_q \in A} g_{q,n}''(c_q\Phi_{\bv_1\cdots \bv_q})\frac{n_{\bv_1} \cdots n_{\bv_q}}{n^{q}}\Big) \Bigg]. \label{eq:En_2}
\end{align}
Here the subscript $\av_1\av_2\cdots\av_q$ of $\Phi_{\av_1\av_2\cdots\av_q}$ is the bit-wise product of $\av_1, \av_2, \cdots, \av_q$, which gives an element in $A$. 
\end{lemma}

To derive the limits of the right hand side of Eqs. (\ref{eq:En_1}) and (\ref{eq:En_2}), we will use the generalized multinomial theorem stated in Proposition \hyperref[thm:gen-multinomial-restate]{4.1 (Formal)}.
In order to use this general result, we need to show that the polynomial inside the exponential function satisfies the well-played property as defined in Definition~\ref{def:well-played-polynomial-redefine}.
This is shown by the following lemma. 
\begin{lemma}\label{lem:P-is-well-played}
Let $h: \R \to \R$ be an even function, i.e., $h(\lambda) = h(-\lambda)$ for any $\lambda \in \R$.
Let $A$ be as defined in Eq. (\ref{eq:set_A_def}). Let $\{ \Phi_{\av} \}_{\av \in A}$ be given as in Eq. (\ref{eq:Phi_def}). Then for any $q \in \Z_{> 0}$, the following polynomial
\begin{equation} \label{eq:q-poly}
    H_q(\{\omega_\av\}_{\av \in A})  = \sum_{\av_1, \ldots, \av_q \in A} h(\Phi_{\av_1\av_2\cdots\av_q}) \omega_{\av_1} \omega_{\av_2} \cdots \omega_{\av_q}
\end{equation}
is well-played (c.f. Definition \ref{def:well-played-polynomial-redefine}). Here the subscript $\av_1\av_2\cdots\av_q$ of $\Phi_{\av_1\av_2\cdots\av_q}$ is the bit-wise product of $\av_1, \av_2, \cdots, \av_q$, which gives an element in $A$. 
\end{lemma}

With these two lemmas, we are now ready to establish Theorem~\ref{thm:moments}. 

\hspace{10pt}

\noindent
{\bf Step 1. Proof of Eq. \eqref{eqn:first-moment-limit-thm}. }
To prove Eq. \eqref{eqn:first-moment-limit-thm}, we apply Proposition \hyperref[thm:gen-multinomial-restate]{4.1 (Formal)} to the right hand side of Eq. \eqref{eq:En_1}. To do so, we need to check the assumptions of Proposition \hyperref[thm:gen-multinomial-restate]{4.1 (Formal)}.
In Lemma \ref{lem:verify-proper-A-Q} later, we verify that the set $A$ as defined in Eq. (\ref{eq:set_A_def}) is a proper set (c.f. Definition \ref{def:set_A}) and that the set of complex numbers $\{Q_\av\}_{\av \in A}$ as defined in Eq. (\ref{eq:Qdef}) is a set of proper complex numbers (c.f. Definition \ref{def:complex_Q}). 

Furthermore, $g_{q,n}$ is an even function since $J_{i_1,\ldots,i_q}$ is symmetric around mean 0 according to Assumption \ref{assum:iidJ}.
By Lemma \ref{lem:P-is-well-played}, the polynomial 
\[
P_n^{(q)}(\{ \omega_{\av}\}) := \sum_{\av_1,\ldots,\av_q \in A} g_{q,n}(c_q \Phi_{\av_1\cdots\av_q})\omega_{\av_1} \cdots \omega_{\av_q}
\]
is well-played.
By the definition of well-played polynomial as in Definition \ref{def:well-played-polynomial-redefine}, it is easy to see that a sum of well-played polynomials is also well-played.
This implies that the polynomial
\[
P_n(\{ \omega_{\av}\}) := \sum_{q = 1}^{q_{\max}} P_n^{(q)}(\{ \omega_{\av}\})= \sum_{q = 1}^{q_{\max}} \sum_{\av_1,\ldots,\av_q \in A} g_{q,n}(c_q \Phi_{\av_1\cdots\av_q})\omega_{\av_1} \cdots \omega_{\av_q}
\]
is well-played. Moreover, it is easy to see that the degree of $P_n$ is $q_{\max}$ which is independent of $n$, and by the fact that $\lim_{n \to \infty} g_{q, n}(\lambda) = g_q(\lambda)$, defining 
\[
P(\{ \omega_{\av}\}) := \sum_{q = 1}^{q_{\max}} \sum_{\av_1,\ldots,\av_q \in A} g_{q}(c_q \Phi_{\av_1\cdots\av_q})\omega_{\av_1} \cdots \omega_{\av_q}, 
\]
we have $\{ P_n \}_{n \ge 1}$ is a sequence of converging well-played polynomials with uniformly bounded degree and with limit $P$. 

Finally, we define
\[
\begin{aligned}
f_n(\{ \omega_\av\}) =&~ -\sum_{q=1}^\qmax ic_q \sum_{\bv_1,\ldots,\bv_q \in A} g_{q,n}'(c_q\Phi_{\bv_1\cdots \bv_q})\omega_{\bv_1} \cdots \omega_{\bv_q},\\
f(\{ \omega_\av\}) =&~ -\sum_{q=1}^\qmax ic_q \sum_{\bv_1,\ldots,\bv_q \in A} g_{q}'(c_q\Phi_{\bv_1\cdots \bv_q})\omega_{\bv_1} \cdots \omega_{\bv_q}.
\end{aligned}
\]
Then it is easy to see that $\{ f_n\}_{n \ge 1}$ is a sequence of converging polynomials with uniformly bounded degree and with limit $f$. 

As a consequence, all the assumptions of Proposition \hyperref[thm:gen-multinomial-restate]{4.1 (Formal)} are satisfied, and applying Proposition \hyperref[thm:gen-multinomial-restate]{4.1 (Formal)} to Eq. (\ref{eq:En_1}), we conclude that
\begin{equation}\label{eqn:first-moment-proof}
    \lim_{n \to \infty} \EV_J[\bgbbraket{C_J/n}] = f(\{ W_\av\})  = -\sum_{q=1}^\qmax ic_q \sum_{\av_1,\ldots,\av_q \in A} g_{q}'(c_q\Phi_{\av_1\cdots \av_q}) W_{\av_1} \cdots W_{\av_q},
\end{equation}
where $\{W_\av\}_{\av\in A}$ is the unique solution to
\begin{equation}\label{eqn:SCE-in-main-theorem-proof}
    W_\av = Q_\av \exp\Big[ \sum_{q=1}^\qmax q \sum_{\bv_1,\ldots,\bv_{q-1} \in A} g_{q}\big(c_q \Phi_{\av \bv_1\cdots\bv_{q-1}}\big) W_{\bv_1} \cdots W_{\bv_{q-1}}  \Big],~~~ \forall \av \in A.
\end{equation}
Note that the right hand side of Eq. \eqref{eqn:first-moment-proof} gives the formula $V_p(\G, \paramv)$ stated in Eq. \eqref{eq:Vp-formula}, and the self-consistent equation \eqref{eqn:SCE-in-main-theorem-proof} coincide with Eq. \eqref{eq:Ws_self_consistent}, which is guaranteed to have a unique solution by Lemma \ref{lem:unique_solution_SCE}.
This concludes the proof of Eq. \eqref{eqn:first-moment-limit-thm}. 

\noindent
{\bf Step 2. Proof of Eq. (\ref{eqn:second-moment-limit-thm}). } To prove Eq. \eqref{eqn:second-moment-limit-thm}, we apply Proposition \hyperref[thm:gen-multinomial-restate]{4.1 (Formal)} to the right hand side of Eq. (\ref{eq:En_2}). The assumptions for $A$, $\{ Q_\av \}_{\av \in A}$, and $\{ P_n\}_{n \ge 1}$ of Proposition \hyperref[thm:gen-multinomial-restate]{4.1 (Formal)} have been checked in Step 1. Furthermore, we define 
\[
\begin{aligned}
f_n(\{ \omega_\av\}) =&~ \Big(-\sum_{q=1}^\qmax i c_q \sum_{\bv_1,\ldots,\bv_q \in A} g_{q,n}'(c_q\Phi_{\bv_1\cdots \bv_q})\omega_{\bv_1} \cdots \omega_{\bv_q}\Big)^2 \\
&\quad + \Big(- \frac{1}{n}\sum_{q=1}^\qmax c_q^2 \sum_{\bv_1,\ldots,\bv_q \in A} g_{q,n}''(c_q\Phi_{\bv_1\cdots \bv_q}) \omega_{\bv_1} \cdots \omega_{\bv_q}\Big), \\
\text{and} \qquad
f(\{ \omega_\av\}) =&~ \Big(-\sum_{q=1}^\qmax i c_q \sum_{\bv_1,\ldots,\bv_q \in A} g_{q}'(c_q\Phi_{\bv_1\cdots \bv_q}) \omega_{\bv_1} \cdots \omega_{\bv_q}\Big)^2.
\end{aligned}
\]
Then it is easy to see that $\{ f_n\}_{n \ge 1}$ is a sequence of converging polynomials with uniformly bounded degree and with limit $f$. 

As a consequence, all the assumptions of Proposition \hyperref[thm:gen-multinomial-restate]{4.1 (Formal)} are satisfied, and applying Proposition \hyperref[thm:gen-multinomial-restate]{4.1 (Formal)} to Eq. (\ref{eq:En_2}), we conclude that
\begin{align}\label{eqn:second-moment-proof}
    \lim_{n \to \infty} \EV_J[\bgbbraket{(C_J/n)^2}] =  f(\{ W_\av\}) = \Big(-\sum_{q=1}^\qmax ic_q \sum_{\av_1,\ldots,\av_q \in A} g_{q}'(c_q\Phi_{\av_1\cdots \av_q}) W_{\av_1} \cdots W_{\av_q} \Big)^2,
\end{align}
where $\{W_\av\}_{\av\in A}$ is the unique solution to Eq. (\ref{eqn:SCE-in-main-theorem-proof}). Note that the right hand side of Eq. (\ref{eqn:second-moment-proof}) gives $V_p(\G, \paramv)^2$ (c.f. Eq. (\ref{eq:Vp-formula})). This concludes the proof of Eq. \eqref{eqn:second-moment-limit-thm} and thus Theorem~\ref{thm:moments}. 

\subsection{Structure of set $A$ and complex numbers $\{ Q_\av\}_{\av \in A}$} \label{sec:structure_A_Q}

In this section, we establish that the set $A$ as defined in Eq. \eqref{eq:set_A_def} is a proper set (c.f. Definition \ref{def:set_A}), and the set of complex numbers $\{ Q_\av \}_{\av \in A}$ as defined in Eq. (\ref{eq:Qdef}) is a set of proper complex numbers (c.f. Definition \ref{def:complex_Q}). We first define a rank function $\ell$ on the set $A$. 
\begin{defn}[Rank function $\ell$]
\label{def:ordering}
For any $\av = (a_1, a_2\ldots, a_p, a_{-p},\ldots, a_{-2}, a_{-1}) \in A$, we define its rank to be
\begin{align}
    \ell(\av) =  \max(\{ i: a_i \neq a_{-i}\} \cup \{0 \}).
\end{align}
In other word, $\ell(\av)$ gives the largest index $i$ where $a_i \neq a_{-i}$, or 0 otherwise. 
\end{defn}
For any $l \in [p]$, we define
\begin{align} \label{eq:Apart-def}
B_l := \{\av \in A: \ell(\av) = l \} = \{\av \in A:  a_{-k}  = a_{k} \text{ for } \ell+1 \le k \le p, \text{ and } a_{-l} = - a_{-l}\}, 
\end{align}
and we define
\begin{equation}\label{eq:A_0_def}
A_{0} := \{\av \in A: \ell(\av) = 0 \} = \{\av \in A: a_{-k} = a_{k}  \text{ for } 1\le k \le p\} \,. 
\end{equation}
Then $A_0, B_1, B_2, \ldots, B_p$ are disjoint and we have
\begin{equation}
A= B_{p} \sqcup B_{p-1} \sqcup \cdots \sqcup B_1 \sqcup A_0.
\end{equation}
We further define $B = \cup_{\ell = 1}^p B_\ell = A\setminus A_{0}$, and define
\begin{equation}\label{eqn:set_D}
\begin{aligned}
D = \Big\{ \av \in B: {\textstyle \prod_{j=1}^p a_j = + 1 }\Big\} 
\qquad \text{and} \qquad
\barD = \Big\{ \av \in B: {\textstyle \prod_{j=1}^p a_j = - 1} \Big\}. 
\end{aligned}
\end{equation}
With these definitions, we have $B = D \sqcup \barD$ and $A = A_0 \sqcup D \sqcup \barD$. 

Recall that we index the entries of $\av \in A$ as $(a_1,\ldots, a_p, a_{-p}, \ldots, a_{-1})$.
Since any element $\av$ in set $B_p$ satisfies $a_{-p} = - a_{p}$, the number of elements in $B_p$ is $2^{2p-1}$. Any element $\av$ in $B_{p-1}$ satisfies $a_{-p} = a_{p}$ and $a_{-p+1} = -a_{p-1}$, so the number of elements in $B_{p-1}$ is $2^{2p-2}$ elements. Similarly, the number of elements in $B_{l}$ is $2^{p + l - 1}$ for any $l \in [p]$. We illustrate this for
$p=3$,
\begin{equation} \label{eq:A-part-example}
\begin{array}{lclllrrrr}
B_p &: \quad     & (a_1, &a_2, &a_3, & -a_{3}, &a_{-2}, &a_{-1})   &\quad 32\text{ elements}, \\
B_2 &: \quad     & (a_1, &a_2, &a_3,  & a_{3}, &-a_{2}, &a_{-1})   & \quad 16\text{ elements}, \\
B_{1} &: \quad   & (a_1, &a_2, &a_3, & a_{3}, &a_{2}, &-a_{1})   & \quad 8\text{ elements}, \\
A_{0} &: \quad & (a_1, &a_2, &a_3, & a_{3}, &a_{2}, &a_{1})   &\quad 8\text{ elements}.
\end{array}
\end{equation}

We next define a full ordering $\succeq$ on $D$. 
\begin{defn}[Full ordering $\succeq$ on $D$]
\label{def:ordering_full}
For any two distinct element $\av_1, \av_2 \in D$, we define the $\prec$ relation as following: (1) If $\ell(\av_1) < \ell(\av_2)$, we let $\av_1 \prec \av_2$; (2) If $\ell(\av_1) > \ell(\av_2)$, we let $\av_2 \prec \av_1$; (3) If $\ell(\av_1) = \ell(\av_2)$ and if $\av_1$ is lexically less than $\av_2$, we let $\av_1 \prec \av_2$; (3) If $\ell(\av_1) = \ell(\av_2)$ and if $\av_1$ is lexically greater than $\av_2$, we let $\av_2 \prec \av_1$ (here lexical order means that, for example, $(-1, -1), (-1, 1), (1, -1), (1, 1)$ are in lexically increasing order). It is easy to see that such $\prec$ relation is a full order, so that we can also define $\preceq$, $\succeq$, and $\succ$ accordingly. 
\end{defn}

We now define a ``bar'' operation that takes configuration $\av \in B_\ell$ to $\bar\av \in B_\ell$ for $1 \le \ell \le p$ via
\begin{align} \label{eq:bar-op}
\bar{a}_{\pm r} = \begin{cases}
a_{\pm r}, & r \neq \ell\\
-a_{\pm r} & r = \ell
\end{cases}
\quad
\text{ for } \quad \av \in B_\ell, \quad 1\le \ell \le p\,.
\end{align}
Note the bar operation is its own inverse. For completeness we define $\bar\av = \av$ for $\av\in A_0$, although this is rarely used. 
For example, for the form of $\av$'s given in \eqref{eq:A-part-example}, the corresponding $\bar\av$'s are
\begin{equation}
\begin{array}{lclrrrrrrr}
B_p &: \quad     & (\hphantom{-}a_1, &a_2, &-a_3, & a_{3}, &a_{-2}, &a_{-1}),  & \phantom{\quad 32\text{ elements}} \\
B_2 &: \quad     & (\hphantom{-}a_1, &-a_2, &a_3,  & a_{3}, &a_{2}, &a_{-1}),  &\\
B_1 &: \quad   & (-a_1, &a_2, &a_3, & a_{3}, &a_{2}, &a_{1}),   & \\
A_{0} &: \quad & (\hphantom{-}a_1, &a_2, &a_3, & a_{3}, &a_{2}, &a_{1}).   &
\end{array}
\end{equation}
It is easy to see that the bar operation is a one-to-one mapping between $D$ and $\barD$.
Observe that $\ell(\xv)=\ell(\bar\xv)$, $\xv \in D \Longrightarrow \bar\xv \in \barD$, and $\xv \in \barD \Longrightarrow \bar\xv \in D$. 

Finally, by the definition of $\{ Q_\av \}_{\av \in A}$ as in \eqref{eq:Qdef} and by the definitions of $A_0$ and $B$ as in Eq. (\ref{eq:Apart-def}) and (\ref{eq:A_0_def}), we can directly verify that
\begin{align}
Q_{\bar\av} = -Q_\av, \quad \forall \av \in B;
\qquad\qquad 
Q_\av \in [0, 1], \quad \forall \av \in A_0; 
\qquad
\text{and}
\qquad
\sum_{\av \in A_0} Q_\av = 1. 
\end{align}

The line of reasoning above proves the following lemma. 
\begin{lemma}\label{lem:verify-proper-A-Q}
The set $A$ as defined in Eq. (\ref{eq:set_A_def}), endowed with the structure $A = A_0 \sqcup D \sqcup \barD$ (where $A_0$ is defined as in (\ref{eq:A_0_def}) and $D$ and $\barD$ is as defined in (\ref{eqn:set_D})) and the bar operation as defined in Eq. (\ref{eq:bar-op}), is a proper set (c.f. Definition \ref{def:set_A}). The set of complex numbers $\{ Q_\av \}_{\av \in A}$ as defined in Eq. (\ref{eq:Qdef}) is a set of proper complex numbers (c.f. Definition \ref{def:complex_Q}). 
\end{lemma}

\subsection{Proof of Lemma \ref{lem:characteristic_function_derivative}}
\label{sec:QAOA-organized-sum}

This proof follows the technique and convention introduced in \cite{farhi2019quantum}.
To begin, let us define a characteristic function $\varphi_n(\lambda)$ as follows:
\begin{equation}\label{eqn:MGF_definition}
\varphi_n(\lambda) = \EV_J[\bgbbraket{e^{i\lambda C_J}}]. 
\end{equation}
Then we have
\begin{equation}\label{eqn:relating_moment_MGF}
\begin{aligned}
    \EV_J[\braket{\paramv | C_J/n | \paramv}] &= \frac{1}{(in)} \frac{\partial \varphi_n(\lambda)}{\partial \lambda} \Big|_{\lambda = 0},\\
    \EV_J[\braket{\paramv | (C_J/n)^2 | \paramv}] &= \frac{1}{(in)^2} \frac{\partial^2 \varphi_n(\lambda)}{\partial^2 \lambda} \Big|_{\lambda = 0}. 
\end{aligned}
\end{equation}
As a consequence, to obtain the first and second moments, we can first calculate the characteristic function $\varphi_n(\lambda)$.

Recall that the QAOA state \eqref{eq:wavefunction} is
\begin{equation}
    \ket{\paramv} = e^{-i\beta_p B} e^{-i\gamma_p C} \cdots e^{-i\beta_1 B} e^{-i\gamma_1 C} \ket{s},
\end{equation}
where
\begin{equation}
\ket{s} = \Big({\textstyle \frac{\ket{+1} + \ket{-1}}{\sqrt{2}} }\Big)^{\otimes n} 
    = \frac{1}{\sqrt{2^n}} \sum_{\zv \in \{\pm1\}^n} \ket{\zv}.
\end{equation}
For a general cost function 
\begin{equation}
    C_J(\zv) = \sum_{q=1}^{\qmax} c_q \sum_{i_1,\ldots,i_q=1}^n J_{i_1,i_2,\ldots,i_q} z_{i_1} z_{i_2}\cdots z_{i_q},
\end{equation}
we have
\begin{align}
\bgbbraket{e^{i\lambda C_J}} &= \braket{s | e^{i\gamma_1 C_J} e^{i\beta_1 B}\cdots e^{i\gamma_p C_J} e^{i\beta_p B} e^{i\lambda C_J} e^{-i\beta_p B} e^{-i\gamma_p C_J}\cdots e^{-i\beta_1 B} e^{-i\gamma_1 C_J}|s}. 
\end{align}
Inserting $2p+1$ resolutions of identity $\Id=\sum_{\zv^j} \ketbra{\zv^j}$, we get
\begin{align}
&~\bgbbraket{e^{i\lambda C_J}} \nonumber \\
&= \sum_{\zv^\P{\pm1}, \ldots, \zv^\P{\pm p}, \zv^\M }
	\braket{s |\zv^\P{1}} e^{i\gamma_1 C_J(\zv^\P{1})} 
	\braket{\zv^\P{1}| e^{i\beta_1 B}|\zv^\P{2}} 
		\cdots e^{i\gamma_p C_J(\zv^\P{p})} \braket{\zv^\P{p}| e^{i\beta_p B} |\zv^\M}
		e^{i\lambda C_J(\zv^\M)} \nonumber \\
	&  \qquad \qquad \times \braket{\zv^\M |e^{-i\beta_p B} |\zv^\P{-p}} e^{-i\gamma_p C_J(\zv^\P{-p})}
	 \cdots \braket{\zv^\P{-2}| e^{-i\beta_1 B}|\zv^\P{-1}}  e^{-i\gamma_1 C_J(\zv^\P{-1})} \braket{\zv^\P{-1}|s}.
\end{align}
Here we label the $2p+1$ strings as $\zv^\P{1}, \zv^\P{2}, \ldots, \zv^\P{p}, \zv^\M, \zv^\P{-p},\ldots, \zv^\P{-2}, \zv^\P{-1}$.
This labelling is convenient because $\zv^\P{j}$ will often be paired with $\zv^\P{-j}$ in the calculations that follow.
Note each factor of the form $\braket{\zv^\P{1}|e^{i\beta_1}|\zv^\P{2}}$ only depends on the bitwise product $\zv^\P{1}\zv^\P{2}$, so we define
\begin{equation}
    f_j(\zv \zv') =\braket{\zv|e^{i\beta_j B}|\zv'}.
\end{equation}
Then we have
\begin{align}
\bgbbraket{e^{i\lambda C_J}} 
=  \frac{1}{2^n} \sum_{\zv^\P{\pm1}, \ldots, \zv^\P{\pm p}, \zv^\M }
	&\exp\Big[{i \sum_{r=1}^p \gamma_r [C_J(\zv^\P{r})-C_J(\zv^\P{-r})] + i\lambda C_J(\zv^\M)}\Big] \nonumber \\
&   \times f_1(\zv^\P{1} \zv^\P{2}) \cdots f_{p-1}(\zv^\P{p-1} \zv^\P{p}) f_{p}(\zv^\P{p} \zv^\M)
 \nonumber \\
&   \times f_{1}^*(\zv^\P{-1} \zv^\P{-2}) \cdots f_{p-1}^*(\zv^\P{1-p} \zv^\P{-p} ) f_{p}^*(\zv^\P{-p}\zv^\M ).
\end{align}
Now we are going to transform the $\zv^\P{j}$'s to simplify this expression.
For every $r=1,2,\ldots, p$, we simultaneously perform the following transform 
\begin{align}
\begin{split}
\zv^\P{r} &\to \zv^\P{r} \zv^\P{r+1} \cdots \zv^\P{p} \zv^\M, \\
\zv^\P{-r} &\to  \zv^\P{-r} \zv^\P{-r-1} \cdots \zv^\P{-p} \zv^\M.
\end{split}
\end{align}
Then $\zv^\P{\pm r}\zv^\P{\pm(r+1)}\to \zv^\P{\pm r}$ for $1\le r\le p-1$, and $\zv^\P{\pm p} \zv^\M\to \zv^\P{\pm p}$ under this transformation.
Recall that in general
\begin{equation}
    C_J(\zv) = \sum_{q=1}^{\qmax} c_q \sum_{i_1,\ldots,i_q=1}^n J_{i_1,i_2,\ldots,i_q} z_{i_1} z_{i_2}\cdots z_{i_q}.
\end{equation}
This gives us
\begin{align}\label{eq:last-zm-sum}
\bgbbraket{e^{i\lambda C_J}}  &=
	\frac{1}{2^n} \sum_{\zv^\P{\pm1}, \ldots, \zv^\P{\pm p}, \zv^\M}
   \exp\Big[i \sum_{q=1}^{\qmax} c_q \sum_{i_1,\ldots,i_q=1}^n J_{i_1,\ldots, i_q} \big(\phi_{i_1,\ldots, i_q}(\Zv)  + \lambda\big) z^\M_{i_1}\cdots z^\M_{i_q}\Big]
		 \nonumber \\
& \quad \quad \times f_1(\zv^\P{1}) f_2(\zv^\P{2}) \cdots f_{p}(\zv^\P{p}) f_{p}^*( \zv^\P{-p}) \cdots f_{2}^*(\zv^\P{-2} ) f_{1}^*(\zv^\P{-1})
\end{align}
where we have denoted $\Zv = (\zv^\P{1}, \ldots, \zv^\P{p}, \zv^\P{-p}, \ldots, \zv^\P{-1}) \in \{ \pm 1\}^{n \times (2p)}$, and 
\begin{align}
     \phi_{i_1,\ldots, i_q}(\Zv) = \sum_{r=1}^p \gamma_r \Big( \prod_{s=1}^q(z_{i_s}^\P{r} z_{i_s}^\P{r+1}\cdots z_{i_s}^\P{p}) - \prod_{s=1}^q(z_{i_s}^\P{-r} z_{i_s}^\P{-r-1}\cdots z_{i_s}^\P{-p})\Big).
\end{align}
Observe that for any distribution of $J_{i_1 \ldots i_q}$ that is symmetric about $0$ as in Assumption~\ref{assum:iidJ}, we have $\EV_J[\rho(J_{i_1 \ldots i_q} z_{i_1}^\M \ldots z_{i_1}^\M)] = \EV_J[\rho(J_{i_1 \ldots i_q})]$ for  any function $\rho(x)$. 
Then taking expectation $\EV_J$ of Eq. \eqref{eq:last-zm-sum} we get the characteristic function as
\begin{align}\label{eq:moment_before_config_basis}
\varphi_n(\lambda) =&~  \sum_{\zv^\P{\pm1}, \ldots, \zv^\P{\pm p}}
   \EV_J\Big\{\exp\Big[i \sum_{q=1}^\qmax c_q \sum_{i_1,\ldots,i_q=1}^n J_{i_1,\ldots, i_q} (\phi_{i_1,\ldots, i_q}(\Zv)  + \lambda)\Big]\Big\}
		 \nonumber \\
& \qquad \times f_1(\zv^\P{1}) f_2(\zv^\P{2}) \cdots f_{p}(\zv^\P{p}) f_{p}^*( \zv^\P{-p}) \cdots f_{2}^*(\zv^\P{-2} ) f_{1}^*(\zv^\P{-1})
\end{align}
where the sum over $\zv^\M$ killed the $1/2^n$ factor in front of \eqref{eq:last-zm-sum}.

Now we calculate the expectation over $J\sim\G(n)$ explicitly.
Under Assumption~\ref{assum:iidJ} where the tensor collection $J_{i_1,\ldots,i_q}$ are i.i.d. with log characteristic function $g_{q, n}$, we have
\begin{align}\label{eq:moment_after_averaging}
\varphi_n(\lambda) &= \sum_{\Zv \in \{ \pm 1\}^{n \times (2p)}}
   \exp\Big[ \sum_{q=1}^\qmax \frac{1}{n^{q-1}} \sum_{i_1,\ldots,i_q=1}^n g_{q,n}\big(c_q(\phi_{i_1,\ldots, i_q}(\Zv)  + \lambda)\big)\Big]
		 \nonumber \\
& \quad \quad \times f_1(\zv^\P{1}) f_2(\zv^\P{2}) \cdots f_{p}(\zv^\P{p}) f_{p}^*( \zv^\P{-p}) \cdots f_{2}^*(\zv^\P{-2} ) f_{1}^*(\zv^\P{-1}).
\end{align}

Following the techniques in \cite{farhi2019quantum}, we perform a change of variables so that the index of the summation in (\ref{eq:moment_after_averaging}) is changed to what the authors called the ``configuration basis.''
This is done via the following argument. Let us fix $\Zv \in \{ \pm 1\}^{n \times (2p)}$ with $(k, l)$-th element $z_k^\P{l} \in \{ \pm 1\}$ (for $k \in [n]$ and $l \in \{ \pm 1. \dots, \pm p\}$), and let us look at the $k$-th row of $\Zv$, which gives
\begin{align}
    (z_k^\P{1}, z_k^\P{2}, \ldots, z_k^\P{p}, z_k^\P{-p}, \ldots, z_k^\P{-2}, z_k^{\P{-1}}) \in A,
\end{align}
where $A = \{\pm 1 \}^{2p}$ as defined in \eqref{eq:set_A_def}.
We denote by $n_\av(\Zv)$ the number of times that the configuration $\av \in A$ occurs among the rows of $\Zv$, i.e., 
\begin{equation}\label{eqn:n_av_Z_definition}
n_\av(\Zv) := \sum_{k = 1}^n 1 \Big\{ (z_k^\P{1}, \ldots, z_k^\P{p}, z_k^\P{-p}, \ldots, z_k^\P{-1} ) = \av \Big\}. 
\end{equation}
By definition, we have $\sum_{\av \in A} n_\av(\Zv) = n$ for any $\Zv \in \{ \pm 1\}^{n \times (2p)}$. 

By this definition, a key observation is that, every term in the summation in Eq. (\ref{eq:moment_after_averaging}) only depend on $\Zv$ through $\{ n_\av(\Zv)\}_{\av \in A}$. Indeed, note we can write every $f_j(\zv^\P{j}) \equiv \braket{\zv^\P{j}|e^{i\beta_j B}|\vect{1}}$ in Eq. \eqref{eq:moment_after_averaging} as a function of $\{ n_\av(\Zv)\}_{\av \in A}$:
\begin{align}
    f_j(\zv^\P{j}) &= (\cos\beta_j)^{\#\text{ of $+1$'s in } \zv^\P{j}} (i\sin\beta_j)^{\#\text{ of $-1$'s in } \zv^\P{j}} \nonumber \\
    &=(\cos\beta_j)^{\sum_\av n_\av(\Zv) \cdot ( 1+a_j)/2} (i\sin\beta_j)^{\sum_\av n_\av(\Zv) \cdot ( 1-a_j)/2 },
\end{align}
which gives
\begin{equation}\label{eqn:f_1_f_p_reformulation}
    f_1(\zv^\P{1}) f_2(\zv^\P{2}) \cdots f_{p}(\zv^\P{p}) f_{p}^*( \zv^\P{-p}) \cdots f_{2}^*(\zv^\P{-2} ) f_{1}^*(\zv^\P{-1}) = \prod_{\av \in A} Q_\av^{n_\av(\Zv)}
\end{equation}
where $Q_\av$ is as defined in Eq.~\eqref{eq:Qdef}.
Similarly, we can write the sum over $i_1, \ldots, i_q \in [n]$ terms involving $\phi_{i_1,\ldots,i_q}(\Zv)$ in Eq.~\eqref{eq:moment_after_averaging} as a function of $\{ n_\av(\Zv)\}_{\av \in A}$:
\begin{align}\label{eqn:phi_to-Phi}
    & \sum_{i_1,\ldots,i_q=1}^n g_{q,n}\big(c_q(\phi_{i_1\ldots i_q}(\Zv)  + \lambda)\big) = \sum_{\av_1,\ldots,\av_q \in A} g_{q,n}\big(c_q(\Phi_{\av_1\ldots \av_q}  + \lambda)\big) n_{\av_1}(\Zv) 
    \cdots n_{\av_q}(\Zv),
\end{align}
where $\Phi_\av$ is as defined in Eq.~\eqref{eq:Phi_def}.

As a consequence, instead of summing over all $2^{n \times (2p)}$ possible bit strings $\Zv = (\zv^\P{\pm1}, \ldots, \zv^\P{\pm p})$ in Eq. \eqref{eq:moment_after_averaging}, we can instead sum over all possible combinations of configurations $\{ n_\av(\Zv) \}_{\av \in A}$. In other words, we change the basis from
\begin{equation}
\{\Zv : \Zv \in \{\pm1\}^{2pn}\} 
    \quad\longrightarrow\quad
\{ n_\av \ge 0: \av \in A,~ {\textstyle \sum_{\av\in A} n_\av = n} 
    \},
\end{equation}
which we call the configuration basis. This can be done by the following equation: for any $f(\{ \omega_\av\}_{\av \in A})$ as a function of $\{ \omega_\av\}_{\av \in A}$, we have
\begin{equation}\label{eqn:change_of_basis}
\sum_{\Zv \in \{ \pm 1\}^{n \times (2p)}} f(\{ n_\av (\Zv)\}_{\av \in A} )  =     \sum_{\{n_\av\ge 0:~ \av\in A, \sum_{\av\in A} n_\av =n \}} \binom{n}{\{n_\av\}} f(\{ n_\av \}_{\av \in A}).
\end{equation}
Here we have used some abuse of notations: in the left hand side of the equation above, $\{ n_\av \}$ should be understood as functions of $\Zv$ (as defined in Eq. (\ref{eqn:n_av_Z_definition})); in the right hand side of the equation above, $\{ n_\av\}$ should be understood as dummy variables that is summed over. 

Hence, applying Eqs. (\ref{eqn:change_of_basis}), (\ref{eqn:f_1_f_p_reformulation}), (\ref{eqn:phi_to-Phi}) to Eq. \eqref{eq:moment_after_averaging}, we have
\begin{align}
    \varphi_n(\lambda) = \sum_{\{n_\av\}} \binom{n}{\{n_\av\}} \prod_{\av\in A} Q_\av^{n_\av}
\exp\Big[
		n \sum_{q=1}^\qmax \sum_{\av_1,\ldots,\av_q \in A} g_{q,n}\Big(c_q (\Phi_{\av_1\cdots\av_q}+ \lambda)\Big)
		\frac{n_{\av_1}}{n} \cdots \frac{n_{\av_q}}{n}
	\Big].
\end{align}
Differentiating the above equation and using Eq.~\eqref{eqn:relating_moment_MGF}, one arrives at Eqs.~\eqref{eq:En_1} and \eqref{eq:En_2}. 
This proves Lemma~\ref{lem:characteristic_function_derivative}.

\subsection{Auxilliary lemmas for Proof of Lemma~\ref{lem:P-is-well-played}}\label{sec:prelim_results_for_proof_thm_moments}

Recall the bar operation defined in Eq. \eqref{eq:bar-op}.
Given an even function $h : \R \to \R$, we define the function $\chi_h: A \times A^* \to \mathbb{R}$ as
\begin{equation}\label{eq:chi_def}
\begin{aligned}
\chi_h(\av;  \cv_1,\cv_2,&~\ldots, \cv_{k}) := \sum_{\dv_s \in \{ \cv_s, \bar \cv_s\}, \forall 1 \le s \le k} h(\Phi_{\av \dv_1 \cdots \dv_{k}}) (-1)^{|\{j :~ \dv_j = \bar \cv_j \}|}, \\
&~ \forall \av \in A,~~~~~~ \forall k \ge 0,~ ( \cv_1,\cv_2,\ldots, \cv_{k}) \in A^k.
\end{aligned}
\end{equation}
For example, we have
\begin{equation}
\begin{split}
\chi_h(\av) &= h(\Phi_\av) , \\
    \chi_h(\av; \cv) &= h(\Phi_{\av \cv}) - h(\Phi_{\av\bar\cv}), \\
    \chi_h(\av; \cv_1, \cv_2) &= h(\Phi_{\av \cv_1 \cv_2}) - h(\Phi_{\av \bar\cv_1 \cv_2}) - h(\Phi_{\av \cv_1 \bar\cv_2}) + h(\Phi_{\av \bar\cv_1 \bar\cv_2}) .
\end{split}
\end{equation}
Here the subscript $\av\bv\cv$ of $\Phi_{\av\bv\cv}$ is the bit-wise product of $\av, \bv, \cv \in A$, which gives an element in $A$.

For any $\av,\bv\in A$, we define
\begin{equation}\label{eqn:definition_Delta}
\Delta_{\av,\bv} := \frac{1}{2} \Big[ h(\Phi_{\bar{\av}\bv}) - h(\Phi_{\av\bv})\Big]. 
\end{equation}

We will next show some properties of $\Delta_{\av,\bv}$ and function $\chi_h$. We start by showing some properties of the rank function $\ell$. 
\begin{lemma}\label{lem:partition_of_product}
Recall the rank function $\ell(\cdot)$ given in Definition \ref{def:ordering}. Recall that $\av\bv$ is the bit-wise product of $\av, \bv \in A$ which gives an element in $A$.
For any $\av, \bv \in A$, we have
\begin{align}\label{eqn:ell_inequality_in_lemma}
    \begin{cases}
        \ell(\av\bv) < \ell(\av), \ell(\bv), & \quad \text{ if } \ell(\av) = \ell(\bv) \neq 0, \\
        \ell(\av\bv) = \ell(\av)=\ell(\bv), & \quad \text{ if } \ell(\av) = \ell(\bv) = 0, \\
        \ell(\av\bv)= \max\{\ell(\av), \ell(\bv)\}, & \quad \text{ if } \ell(\av) \neq \ell(\bv).
    \end{cases}
\end{align}
In particular, this implies that for any $\av,\bv\in A$, we have
\begin{equation} \label{eq:ell_inequality_simple}
    \ell(\av\bv) \le \max\{\ell(\av),\ell(\bv)\}.
\end{equation}
\end{lemma}

\begin{proof}
First, let $\av, \bv \in A$ be such that $\ell(\av) = \ell(\bv) = l\neq 0$.
Then by the definition of rank function $\ell$, we have $(\av\bv)_j = (\av\bv)_{-j}$ for any $j > l$. Furthermore, we have
\begin{align}
    (\av\bv)_l = a_l b_l = (-a_{-l})(-b_{-l}) = a_{-l} b_{-l} = (\av\bv)_{-l}.
\end{align}
Hence we have $\ell(\av\bv)< l$. This proves the first case in Eq. (\ref{eqn:ell_inequality_in_lemma}). 

Now suppose $\ell(\av)=\ell(\bv)=0$. Then $(\av\bv)_j = a_j b_j = a_{-j} b_{-j} = (\av\bv)_{-j}$, so $\ell(\av\bv)=0$ as well.
This proves the second case in \eqref{eqn:ell_inequality_in_lemma}.

Lastly, let $\av,\bv \in A$ be such that $\ell(\av)\neq \ell(\bv)$.
Without loss of generality, we assume that $\ell(\av) > \ell(\bv)$.
Then for any $j$ such that $\ell(\av) + 1 \le j \le p$, we have
\begin{equation}
    (\av\bv)_j = a_j b_j = a_{-j}b_{-j} = (\av\bv)_{-j}.
\end{equation}
Moreover, we have
\begin{align}
    (\av\bv)_{\ell(\av)} = a_{\ell(\av)} b_{\ell(\av)} = (-a_{-\ell(\av)})(b_{-\ell(\av)}) = -(\av\bv)_{-\ell(\av)}.
\end{align}
This implies that $\ell(\av\bv) = \ell(\av)$, which proves the last case in Eq. (\ref{eqn:ell_inequality_in_lemma}). 
\end{proof}

\begin{lemma}\label{lem:Delta_iden_separate_bar}
For any $\av, \bv \in A$ such that $\ell(\av) < \ell(\bv)$, we have
\begin{align}\label{eq:bar_ab_separates}
    \overline{\av\bv} = \av \bar\bv,
\end{align}
where the bar operation is as defined in Eq. (\ref{eq:bar-op}).
\end{lemma}

\begin{proof}
Since $\ell(\av) < \ell(\bv)$, Lemma \ref{lem:partition_of_product} implies that $\ell(\av\bv)= \ell(\bv)$.
Now let us look at the $i$-th element of $\overline{\av\bv}$ for $i \in \{ \pm 1, \pm 2, \ldots, \pm p \}$. If $i \not \in \{ \pm \ell(\bv)\} = \{\pm \ell(\av\bv)\}$, we have
\begin{align}
    (\overline{\av\bv})_i = (\av\bv)_i = a_i b_i = a_i \bar{b}_i = (\av\bar\bv)_i.
\end{align}
Alternatively, if $i \in \{ \pm \ell(\bv)\} = \{\pm \ell(\av\bv)\}$, we have
\begin{align}
    (\overline{\av\bv})_i = -(\av\bv)_i = -a_i b_i = a_i \bar{b}_i = (\av\bar\bv)_i.
\end{align}
Thus Eq.~\eqref{eq:bar_ab_separates} holds.
\end{proof}

\begin{lemma}\label{lem:delta_zero}
For any $\av, \bv \in A$ such that $\ell(\av) \ge \ell(\bv)$, we have
\begin{align}
    \Delta_{\av,\bv} = 0,
\end{align}
where $\Delta_{\av,\bv}$ is as defined in Eq. (\ref{eqn:definition_Delta}) with any even function $h : \R \to \R$.
\end{lemma}
\begin{proof}
This is a generalization of \cite[Lemma 1]{farhi2019quantum}. In the proof of that lemma, it is shown that $\Phi_{\av\bv} = - \Phi_{\bar\av\bv}$ if $\ell(\av) \ge \ell(\bv)$. Since $h(\cdot)$ is even, it follows that $\Delta_{\av,\bv}=0$.
\end{proof}

\begin{lemma}\label{lem:Delta_iden_bar_second_elem}
For any $\av, \bv \in A$, we have
\begin{align}\label{eqn:Delta_iden_in_lemma}
\Delta_{\av,\bv} = \Delta_{\av,\bar\bv}, 
\end{align}
where $\Delta_{\av,\bv}$ is as defined in Eq. (\ref{eqn:definition_Delta})  with any even function $h : \R \to \R$.
\end{lemma}
\begin{proof}
If $\ell(\av) \ge \ell(\bv)$, Eq. (\ref{eqn:Delta_iden_in_lemma}) is trivially true by Lemma~\ref{lem:delta_zero} and noting $\ell(\bv)=\ell(\bar\bv)$.

If $\ell(\av) < \ell(\bv)$, by definition of $\Delta_{\av,\bv}$, we have
\[
\begin{aligned}
    \Delta_{\av,\bv} - \Delta_{\av,\bar\bv} &= \frac{1}{2} \Big[ h(\Phi_{\bar\av\bv}) - h(\Phi_{\av\bv}) - h(\Phi_{\bar\av\bar\bv}) + h(\Phi_{\av\bar\bv}) \Big]\\
    &= \frac{1}{2} \Big[h(\Phi_{\bar\bv\av}) - h(\Phi_{\bv\av}) - h(\Phi_{\bar\bv\bar\av}) + h(\Phi_{\bv\bar\av}) \Big] \\
    &= \Delta_{\bv,\av} - \Delta_{\bv,\bar\av} =0,
\end{aligned}
\]
where we used the fact that $\Phi_{\av \bv} = \Phi_{\bv \av}$ by commutativity of bit-wise products, and the last equality is by Lemma~\ref{lem:delta_zero} and by $\ell(\av)=\ell(\bar \av) < \ell(\bv)$. This completes the proof.
\end{proof}

\begin{lemma}\label{lem:chi-vanish}
Recall the definition of $\chi_h$ as defined in Eq. (\ref{eq:chi_def}) in which $h$ is an even function. For any $k \ge 1$ and $\av, \cv_1, \ldots, \cv_{k} \in A$ such that $\max\{\ell(\cv_1), \ldots, \ell(\cv_k)\} \ge \ell(\av)$, we have
\begin{align}
    \chi_{h}(\av;\cv_1,\ldots,\cv_k) = 0.
\end{align}
In other words, $\chi_{h}(\av;\cv_1,\ldots,\cv_k)\neq 0$ only if $\max\{\ell(\cv_1),\ldots,\ell(\cv_k) \} < \ell(\av)$ or $k = 0$.
\end{lemma}

\begin{proof}
Let $M \in [k]$ such that $\ell(\cv_M) = \max \{\ell(\cv_1), \ldots, \ell(\cv_k)\}$. Note that we have
\begin{align}
    \chi_{h}(\av;\cv_1,\ldots,\cv_k) =& \sum_{\substack{\dv_j \in \{\cv_j, \bar{\cv}_j\} \\ 1\leq j \leq k}} (-1)^{|\{j \colon \dv_j = \bar{\cv}_j\}|} h(\Phi_{\av \dv_1 \ldots \dv_{k}}) \nonumber \\
    =& \sum_{\substack{\dv_j \in \{\cv_j, \bar{\cv}_j\} \\ j \in [k] \setminus \{M\}}} (-1)^{|\{j \colon \dv_j = \bar{\cv}_j, j\neq M\}|} \big(h(\Phi_{\av \dv_1 \cdots \cv_M \cdots \dv_{k}}) - h(\Phi_{\av \dv_1 \cdots \bar{\cv}_M \cdots \dv_{k}}) \big) \nonumber \\
    =& -2 \sum_{\substack{\dv_j \in \{\cv_j, \bar{\cv}_j\} \\ j \in [k] \setminus \{M\}}} (-1)^{|\{j \colon \dv_j = \bar{\cv}_j, j\neq M\}|}  \Delta_{\cv_M,\av \dv_1 \cdots \dv_{M-1} \dv_{M+1} \cdots \dv_{k}}, \label{eq:sum_deltas_that_are_zero}
\end{align}
where $\Delta_{\av, \bv}$ is as defined in Eq. (\ref{eqn:definition_Delta}). By the assumption that $\ell(\cv_M) \ge \ell(\av)$, we have $\ell(\cv_M) \ge \max\{\ell(\av), \ell(\cv_1), \ldots, \ell(\cv_k)\}$. So by Lemma \ref{lem:partition_of_product} and by the fact that $\dv_j \in \{ \cv_j, \bar \cv_j\}$ so that $\ell(\dv_j) = \ell(\cv_j)$ for $j \in [k]$, we have
\begin{align}
    \ell(\cv_M) \geq \ell(\av \dv_1 \cdots \dv_{M-1} \dv_{M+1} \cdots \dv_{k}).
\end{align}
As a consequence of Lemma \ref{lem:delta_zero}, all terms in Eq.~\eqref{eq:sum_deltas_that_are_zero} are zero. This concludes the proof. 
\end{proof}

\subsection{Proof of Lemma \ref{lem:P-is-well-played}}
\label{sec:proof_lemma_well_played}
Define $\tau_\av = \omega_\av + \omega_{\bar \av}$ and $\eta_\av = \omega_\av - \omega_{\bar \av}$ for $\av \in B = D \sqcup \overline{D}$, and define $\nu_\av = \omega_\av$ for $\av \in A_0$. Note that by the definition of $H_q$ and using the transformation of variables, we have
    \begin{align}
        H_q(\{ \omega_\av\}_{\av}) =&~ \sum_{S \subseteq [q]} \sum_{\substack{\av_j \in B, j\in S \\ \av_j \in A_0, j\not\in S}} h(\Phi_{\av_1\cdots\av_q}) \omega_{\av_1} \omega_{\av_2} \cdots \omega_{\av_q} \nonumber \\
        =&~ \sum_{S \subseteq [q]} \sum_{\substack{\av_j \in A_0 \\ j \not\in S}} \Big(\prod_{j \not\in S} \nu_{\av_j}\Big)\sum_{\substack{\av_j \in B \\ j \in S}} h(\Phi_{\av_1\cdots\av_q}) \prod_{j \in S} \Big(\frac{\tau_{\av_j} + \eta_{\av_j}}{2}\Big) \nonumber \\
        =&~ \sum_{S \subseteq [q]} \frac{1}{2^{|S|}} \sum_{R \subseteq S} \sum_{\substack{\av_j \in A_0 \\ j \not\in S}} \Big(\prod_{j \not\in S} \nu_{\av_j}\Big)\sum_{\substack{\av_j \in B \\ j \in S}} h(\Phi_{\av_1\cdots\av_q}) \Big(\prod_{j \in R} \tau_{\av_j} \Big) \Big( \prod_{j \in S \setminus R} \eta_{\av_j}\Big). \label{eq:sum_ns_to_t_d}
    \end{align}
Moreover, since $\tau_{\av} = \tau_{\bar\av}$ and $\eta_{\av} = - \eta_{\bar\av}$ for $\av \in B = D \sqcup \overline{D}$, then for any subsets $R \subseteq S \subseteq [q]$, we have
    \begin{align} \label{eq:sum_tdn_with_signs}
        &~\sum_{\substack{\av_j \in A_0 \\ j \not\in S}}  \sum_{\substack{\av_j \in B \\ j \in S}}  h(\Phi_{\av_1\cdots\av_q}) \Big(\prod_{j \in R} \tau_{\av_j} \Big) \Big( \prod_{j \in S \setminus R} \eta_{\av_j}\Big) \Big(\prod_{j \not\in S} \nu_{\av_j}\Big) \nonumber \\
        =&~\sum_{\substack{\av_j \in A_0 \\ j \not\in S}} \sum_{\substack{\av_j \in D \\ j \in S}} \sum_{\substack{\bv_j \in \{\av_j, \bar{\av}_j\}, j\in S \\ \bv_j = \av_j, j\not\in S}} (-1)^{|\{ j \in S \setminus R \colon \bv_j = \bar{\av}_j\}|} h(\Phi_{\bv_1\cdots\bv_q}) \Big(\prod_{j \in R} \tau_{\av_j} \Big) \Big( \prod_{j \in S \setminus R} \eta_{\av_j}\Big) \Big(\prod_{j \not\in S} \nu_{\av_j}\Big).
\end{align}
By the definition of $\chi_h$ as in Eq. (\ref{eq:chi_def}), we have
\begin{align}\label{eq:sum_w_signs_to_chi}
    &~\sum_{\substack{\bv_j \in \{\av_j, \bar{\av}_j\}, j\in S \\ \bv_j = \av_j, j\not\in S}}  (-1)^{|\{ j \in S \setminus R \colon \bv_j = \bar{\av}_j\}|} h(\Phi_{\bv_1,\ldots,\bv_q}) \nonumber \\
    =&~\sum_{\bv_j \in \{\av_j, \bar{\av}_j\}, j\in R} \chi_{h}\bigg(\prod_{j\in R}\bv_j \prod_{j \not\in S} \av_j ~;~ \{\av_j \colon j \in S\setminus R\}\bigg),
\end{align}
where we use the convention that $\prod_{j \in R} \bv_j$ is an all-one vector with dimension $2 p $ if $R = \emptyset$. 

Since the labels in $R,S$ are dummy, we can combine Eqs.~\eqref{eq:sum_ns_to_t_d}, \eqref{eq:sum_tdn_with_signs} and \eqref{eq:sum_w_signs_to_chi} to see that
\begin{align}
&~ \cC[H_q] (\{\tau_\av\}_{\av \in D}, \{ \eta_\bv\}_{\bv \in D}, \{ \nu_\cv\}_{\cv \in A_0})  \nonumber \\
=&~ \sum_{0 \leq r \leq s \leq q} 2^{-s} \binom{q}{s} \binom{s}{r} \sum_{\substack{\av_j \in D \colon 1\leq j \leq r \\ \bv_j \in D \colon r+1 \leq j \leq s \\ \cv_j \in A_0 \colon s+1 \leq j\leq q}} \sum_{\dv_j \in \{\av_j, \bar{\av}_j\}} \chi_{h} \Big( \prod_{j=1}^r \dv_j \prod_{j=s+1}^q \cv_j ; \{\bv_j\} \Big) \prod_{j=1}^r \tau_{\av_j} \prod_{j=r+1}^s \eta_{\bv_j} \prod_{j=s+1}^q \nu_{\cv_j}. \label{eq:Phi_2_sum_expanded}
\end{align}

Note that, if $r=0$, then $\max\{\ell(\bv_j)\} \geq 0 = \ell(\prod_{j=s+1}^q \cv_j)$, so the summand is $0$ by Lemma~\ref{lem:chi-vanish}. Otherwise, by Lemma~\ref{lem:chi-vanish}, the summand in Eq.~\eqref{eq:Phi_2_sum_expanded} is nonzero only if
\begin{align}
    \max_{r+1 \le j \le s}\{\ell(\bv_j)\}  <  \ell\big({\textstyle \prod_{j=1}^r \dv_j \prod_{j=s+1}^q \cv_j}\big) \le \max_{1 \le j \le r}\{\ell(\av_j)\},
\end{align}
where the second inequality uses Eq.~\eqref{eq:ell_inequality_simple} in Lemma \ref{lem:partition_of_product}.
This allows us to write
\begin{align}
    & \cC[H_q] (\{\tau_\av\}_{\av \in D}, \{ \eta_\bv\}_{\bv \in D}, \{ \nu_\cv\}_{\cv \in A_0}) \nonumber \\
    =&~ \sum_{1 \leq r \leq s \leq q} 2^{-s} \binom{q}{s} \binom{s}{r}
    \hspace{-15pt}
    \sum_{\substack{\av_j \in D \colon 1\leq j \leq r \\ \bv_j \in D \colon r+1 \leq j \leq s \\ \cv_j \in A_0 \colon s+1 \leq j\leq q \\ \max\{\ell(\bv_j)\} < \max\{\ell(\av_j)\}}} 
    \hspace{-10pt}
    \sum_{\dv_j \in \{\av_j, \bar{\av}_j\}} \chi_{h} \Big( \prod_{j=1}^r \dv_j \prod_{j=s+1}^q \cv_j ; \{\bv_j\} \Big) \prod_{j=1}^r \tau_{\av_j} \prod_{j=r+1}^s \eta_{\bv_j} \prod_{j=s+1}^q \nu_{\cv_j}. \label{eq:Phi_2_sum_expanded_simplified}
\vspace{-5pt}
\end{align}
Note that the canonical representation of $H_q$ in the above equation verifies that $H_q$ satisfies the property of well-played polynomial (c.f. Definition \ref{def:well-played-polynomial-redefine}). This completes the proof of Lemma~\ref{lem:P-is-well-played}. 

\subsection{Computational complexity of evaluating $V_p$}
\label{apx:efficiency-4pqmax}

In this section, we examine the computational complexity of evaluating the formula $V_p(\G,\paramv)$ described in Theorem~\ref{thm:moments}.
This involves solving a self-consistent equation \eqref{eq:Ws_self_consistent},
where the existence of a unique solution is guaranteed by Lemma~\ref{lem:unique_solution_SCE}.
The same lemma also describes an explicit algorithm to solve this equation in time $O(|A|^{\qmax+1})=O(4^{p(\qmax+1)})$,
where $\qmax$ is the degree of the relevant well-played polynomial:
\vspace{-2pt}
\begin{equation}
P(\{\omega_\av\})
    =\sum_{q=1}^\qmax \sum_{\av_1,\ldots,\av_q \in A} g_{q}(c_q\Phi_{\av_1\cdots \av_q}) \omega_{\av_1} \cdots \omega_{\av_q}.
\vspace{-2pt}
\end{equation}
Nevertheless, we can in fact slightly improve the time-complexity bound to $O(4^{p\qmax})$ due to the explicit form of $P$ above.
We explain this in the following lemma:

\begin{lemma}
\label{lem:efficiency-4pqmax}
Let $\G$ be an ensemble describing $C_J$ of the form \eqref{eq:CJ} and satisfying Assumption~\ref{assum:iidJ}.
Then for any $(\paramv)\in\mathbb{R}^{2p}$, the formula in Theorem~\ref{thm:moments} for 
\begin{align}
    V_p(\G,\paramv) =\lim_{n\to\infty} \EV_{J\sim \G(n)}[\bgbbraket{C_J/n}]
\end{align}
can be explicitly evaluated with an $O(4^{p\qmax})$-time iterative procedure using $O(4^p)$ memory.
\end{lemma}
\begin{proof}
We first write down the self-consistent equation explicitly as (c.f. Eq.~\eqref{eq:Ws_self_consistent}) 
\begin{equation}
W_\av = Q_\av \exp[\partial_{\omega_\av} P(\{W_\bv\})] = Q_\av \exp\Big[\sum_{q=1}^\qmax q \sum_{\bv_1,\ldots,\bv_{q-1} \in A} g_{q}\big(c_q \Phi_{\av \bv_1\cdots\bv_{q-1}}\big) W_{\bv_1} \cdots W_{\bv_{q-1}} \Big].
\end{equation}
We then iteratively solve for the variables $\{W_\av: \av\in A\}$ in the ascending order of $A= \{\pm 1\}^{2p}$ given in Definition~\ref{def:ordering_full}.
The lowest ordered elements are those in $A_0$.
But for every $\av\in A_0$, we have $W_\av = Q_\av$ as shown in Lemma~\ref{lem:unique_solution_SCE}, with $Q_\av$ now explicitly given in Eq.~\eqref{eq:Qdef}.
So it remains to solve for $W_\av$ for every $\av\in D$ in the ascending order, and obtain $W_{\bar\av}=-W_\av$ for its partner $\bar\av\in\barD$.
As explained in Eq.~\eqref{eq:dPdWx} of the proof of Lemma~\ref{lem:unique_solution_SCE}(b), $\partial_{\omega_\av} P(\{W_\bv\})$ only depends on $\{W_\bv: \bv \prec \av\}$.
Hence, this allows us to write for every $\av \in D$
\begin{equation}
W_\av = Q_\av \exp\Big[\sum_{q=1}^\qmax q \sum_{\bv_1,\ldots,\bv_{q-1} \prec \av} g_{q}\big(c_q \Phi_{\av \bv_1\cdots\bv_{q-1}}\big) W_{\bv_1} \cdots W_{\bv_{q-1}} \Big].
\vspace{-3pt}
\end{equation}
Here, each $W_\av$ is obtained from evaluating a sum over $O(|A|^{\qmax-1})$ terms, regardless of $g_q(\lambda)$.
Repeating this $|D|=O(|A|)$ times, we obtain the full solution $\{W_\av\}_{\av \in D}$ after $O(|A|^{\qmax})$ time.

Finally, the formula \eqref{eq:Vp-formula} of $V_p(\G,\paramv)$ involves a sum over $O(|A|^{\qmax})$ terms.
Hence, the total time complexity to get $V_p$ is $O(|A|^{\qmax})$, and the total memory complexity is $O(|A|)$ for storing $\{W_\av\}$.
Plugging in $|A|=4^p$, we get the complexity stated in the lemma.
\end{proof}


\newcommand{\pure}{\textnormal{pure}}

\section{Proof of Theorem~\ref{thm:dense_sparse_agreement} (Universality)}
\label{apx:ER}

In this appendix, we prove Theorem~\ref{thm:dense_sparse_agreement}, which is the statement that the QAOA at constant levels have a certain universality property: its performance for many models of random COPs agrees with the one for $\G_q$ asymptotically.

Recall $\G_q$ is the pure $q$-spin model with all-to-all random Gaussian couplings.
Let $\G_{d,q}(n)$ be a generic ensemble of COPs as described in the statement of Theorem~\ref{thm:dense_sparse_agreement}, where $d$ is a generic parameter which may or may not be the graph vertex degree.
We first derive the expressions of $V_p(\G_q, \paramv)$ and $V_p(\G_{d,q},\paramv)$ using Theorem \ref{thm:moments}. 

\paragraph{The pure $q$-spin model $\G_q$. }
The cost function of the $\G_q$ model is 
\[
C_{J}(\zv) = \sum_{i_1,i_2,\cdots, i_q=1}^n J_{i_1, i_2,\ldots, i_q} z_{i_1}z_{i_2}\cdots z_{i_q},
\]
where $J_{i_1, i_2,\ldots, i_q} \sim_{iid} \mathcal{N}(0, 1/n^{q-1})$. Note that the log characteristic function of $J_{i_1, \ldots, i_q}$ gives
\begin{equation}
    h_{q,n}(\lambda) = n^{q-1} \log \EV_J[e^{iJ_{1,2,\ldots,q} \lambda}] = -\frac{\lambda^2}{2},
\end{equation}
so that we have
\begin{equation}
    h_q(\lambda) = \lim_{n \to \infty} h_{q,n}(\lambda) = -\lambda^2/2,  \qquad \text{and} \qquad h'_q(\lambda) = -\lambda.
\end{equation}
This verifies that Assumption~\ref{assum:iidJ} is satisfied, so we can apply Theorem~\ref{thm:moments}.
Using formula \eqref{eq:Vp-formula} with $c_q=1$ and $c_{q'}=0$ for $q'\neq q$, we get
\begin{equation}\label{eq:Vp_pure}
    V_p(\G_q, \paramv) = i \sum_{\av_1,\ldots,\av_q \in A} \Phi_{\av_1\cdots \av_q} W_{\av_1}^\pure \cdots W_{\av_q}^\pure,
\end{equation}
where $\{W_\av^\pure\}_{\av\in A}$ is given as the unique solution to the following system of equations:
\begin{equation}\label{eq:W_pure_self_consistent}
W_\av^\pure = Q_\av \exp\Big[{-}\frac{q}{2} \sum_{\bv_1,\ldots,\bv_{q-1} \in A} \Phi_{\av \bv_1\cdots\bv_{q-1}}^2 W_{\bv_1}^\pure \cdots W_{\bv_{q-1}}^\pure \Big]. 
\end{equation}

\paragraph{A generic model $\G_{d,q}$. }
Since $\G_{d,q}$ satisfies Assumption~\ref{assum:iidJ}, we may apply Theorem~\ref{thm:moments}.
Using the formula \eqref{eq:Vp-formula} with $c_q=1$ and $c_{q'}=0$ for $q'\neq q$, we have
\begin{equation}
    V_p(\G_{d,q}, \paramv) = -i \sum_{\av_1,\ldots,\av_q \in A} g_{q}^{(d)\prime}
(\Phi_{\av_1\cdots \av_q}) W_{\av_1,d} \cdots W_{\av_q,d},
\end{equation}
where we defined
\begin{equation}
    g_{q}^{(d)}(\lambda) = \lim_{n \to \infty} g_{q,n}^{(d)}(\lambda),
\end{equation}
and $\{W_{\av,d}\}_{\av\in A}$ is given as the unique solution to the following system of equations:
\begin{equation} \label{eq:SCE-ER-d}
W_{\av,d} = Q_\av \exp\Big[\sum_{\bv_1,\ldots,\bv_{q-1} \in A} q g_{q}^{(d)}(\Phi_{\av \bv_1\cdots\bv_{q-1}})
W_{\bv_1,d} \cdots W_{\bv_{q-1},d} \Big].
\end{equation}

\paragraph{Proof that $\lim_{d \to \infty}V_p(\G_{d,q}, \paramv) = V_p(\G_{q}, \paramv)$. } We first assert that the following limit exists:
\begin{align}
\lim_{d\to\infty} W_{\av,d} = \tilde W_\av,
\end{align}
where $\{ W_{\av, d} \}_{\av \in A}$ is the solution to Eq. (\ref{eq:SCE-ER-d}), and $\{ \tilde W_\av \}_{\av\in A}$ are some complex numbers. 

This can be shown inductively in the ordering on $D$.
Let $k=1,2,\ldots, |D|$ be elements of $D$ in its order.
As explained in Eq.~\eqref{eq:dPdWx} of the proof of Lemma~\ref{lem:unique_solution_SCE}(b), the polynomial in the exponential of \eqref{eq:SCE-ER-d}, which corresponds to $\partial_{\omega_\av} P(\{W_{\bv,d}\}_{\bv\in A})$, in fact only depends on $\{W_{\bv,d}: \bv \prec \av\}$.
So let us rewrite \eqref{eq:SCE-ER-d} as
\begin{equation}
W_{k,d} = Q_k \exp\Big[f_{k, d}(\{W_{j,d}: j \prec k\}) \Big],
\end{equation}
where $f_{k, d}$ can be taken to be the polynomial upstairs in \eqref{eq:SCE-ER-d} with $W_{j,d}=0$ for $j\succeq k$.
Explicitly, this is
\begin{equation}
    f_{k, d}(\{W_{j,d}: j \prec k\}) = \sum_{\bv_1,\ldots,\bv_{q-1} \in A, ~ \bv_j \prec k} q g_{q,d}(\Phi_{\av \bv_1\cdots\bv_{q-1}})
    W_{\bv_1,d} \cdots W_{\bv_{q-1},d} .
\end{equation}
In particular $f_1 = 0$, so $W_{1,d} = Q_1$ and the limit trivially exists.
Then for any $k\ge 2$ we have
\begin{align}
\lim_{d\to\infty} W_{k,d} = \lim_{d\to\infty} Q_k \exp\Big[f_{k, d}(\{W_{j,d}: j \prec k\}) \Big] = Q_k \exp\Big[f_k(\{\tilde W_{j}: j \prec k\}) \Big] 
\end{align}
which exists by the fact that 
\begin{equation}
\lim_{d \to \infty} f_{k,d} =
-\frac{q}{2}\sum_{{\bv_1,\ldots,\bv_{q-1} \in A,~ \bv_j \prec k}} \Phi_{\av \bv_1\cdots\bv_{q-1}}^2 \tilde{W}_{\bv_1} \cdots \tilde{W}_{\bv_{q-1}} =: f_k.
\end{equation}
This follows from the assumption that $\lim_{d\to\infty} g_{q}^{(d)}(\lambda) = -\lambda^2/2$ in Eq.~\eqref{eq:universality_assumption} and the inductive hypothesis that $\lim_{d\to\infty} W_{j,d} = \tilde W_j$ for $j\prec k$. 

Knowing that $\lim_{d\to\infty} W_{\av,d} = \tilde W_\av$ exists, we can take the limit on both sides of \eqref{eq:SCE-ER-d} to get
\begin{align}
\tilde W_{\av} 
 &= Q_\av \exp\Big[{-}\frac{q}{2}\sum_{\bv_1,\ldots,\bv_{q-1} \in A} \Phi_{\av \bv_1\cdots\bv_{q-1}}^2  \tilde W_{\bv_1} \cdots \tilde W_{\bv_{q-1}} \Big].
\end{align}
And since $\lim_{d\to\infty} g_{q}^{(d)\prime}(\lambda) = -\lambda$ due to the assumption in Eq.~\eqref{eq:universality_assumption}, 
we get
\begin{align}
    \lim_{d\to\infty} V_p(\G_{d,q}, \paramv) = i\sum_{\av_1,\ldots,\av_q \in A} \Phi_{\av_1\cdots \av_q}
    	 \tilde W_{\av_1} \cdots \tilde W_{\av_q}.
\end{align}
Since $\tilde W_\av$ and $W_\av^\pure$ are given as the unique solution to the same system of equations, we have
\begin{equation} \label{eq:equivalence}
   \tilde W_\av = W_\av^\pure 
    \qquad \Longrightarrow \qquad
    \lim_{d\to\infty} V_p(\G_{d,q}, \paramv) = V_p(\G_q, \paramv).
\end{equation}

\paragraph{The sparse \ER model $\G^\ERsub_{d,q}$. } For the ensemble $\G^\ERsub_{d,q}$ described in Section~\ref{sec:review}, we can map it to an equivalent i.i.d. description as follows.
In the original description of the ensemble, we select $m\sim \Poisson(dn)$ directed hyperedges uniformly from the set $[n]^q$ of all ordered $q$-tuples, and every hyperedge (can possibly be identical) is associated an independent random weight $\Unif(\{ \pm 1/\sqrt{d}\})$.
To give an alternative description of the ensemble, we use the Poisson splitting lemma.
Intuitively, this refers to the fact that any Poisson random variable can be decomposed into a sum of independent Poisson variables. More formally,
\begin{lemma}[Poisson splitting lemma]\label{lem:Poisson-splitting}
Let $\lambda > 0$ and $(p_1, \ldots, p_K) \in [0, 1]^K$ with $\sum_{k = 1}^K p_k = 1$. Let $J \sim \Poisson(\lambda)$, and conditional on $J$ we let $w_1, w_2, \ldots, w_J \sim_{iid} {\rm Categorical}(p_1, \ldots, p_K)$. Denote $N_k = \sum_{j = 1}^J 1\{ w_j = k \}$ for each $k \in [K]$. Then we have $(N_1, \ldots, N_K)$ are mutually independent, and $N_k \sim \Poisson(\lambda p_k)$ for $k \in [K]$. 
\end{lemma}
By the Poisson splitting lemma, for a fixed $q$-tuple $(i_1,\ldots,i_q)$, the number of times it is selected as a hyperedge, which we denote as $m_{i_1 \cdots i_q}$, follows the $\Poisson(d/n^{q-1})$ distribution independently for every $(i_1,\ldots,i_q)$. And each time it is selected as an hyperedge, it is equipped with a weight $w^{l}_{i_1 ,\ldots, i_q} \sim_{iid} \Unif(\{\pm 1/\sqrt{d}\})$ for $1 \le l \le m_{i_1, \ldots, i_q}$. Since the weights chosen every time the $q$-tuple $(i_1,\ldots,i_q)$ is selected add up in the cost function, the coupling constant $J_{i_1,\ldots, i_q}$, in the $\G^\ERsub_{d,q}$ model cost function $C_J = \sum_{i_1,i_2,\cdots, i_q=1}^n J_{i_1, i_2,\ldots, i_q} z_{i_1}z_{i_2}\cdots z_{i_q}$, is distributed as $J_{i_1, \ldots, i_q} \sim_{iid} \sum_{l = 1}^{m_{i_1, \ldots, i_q}} w^{l}_{i_1, \ldots, i_q}$ with $m_{i_1, \ldots, i_q} \sim_{iid} \Poisson(d/n^{q-1})$ and $w^{l}_{i_1, \ldots, i_q} \sim_{iid} \Unif(\{ \pm 1/\sqrt{d}\})$.
Here it is understood that $J_{i_1,\ldots,i_q}=0$ if $m_{i_1,\ldots,i_q}=0$. Furthermore, denote $J^+_{i_1, \ldots, i_q} = \sum_{l = 1}^{m_{i_1,\ldots, i_q}} 1\{ w^{l}_{i_1,\ldots, i_q} = + 1/\sqrt{d}\}, J^-_{i_1, \ldots, i_q} = \sum_{l = 1}^{m_{i_1,\ldots, i_q}} 1\{ w^{l}_{i_1, \ldots, i_q} = - 1/\sqrt{d}\}$. Then we have $J_{i_1, \ldots, i_q} = ( J^+_{i_1, \ldots, i_q} - J^-_{i_1, \ldots, i_q}) / \sqrt{d} $. By the Poisson splitting lemma again, we have $J^+_{i_1, \ldots, i_q} = \sum$ $J^+_{i_1, \ldots, i_q}, J^-_{i_1, \ldots, i_q} \sim_{iid} \Poisson(d/(2n^{q-1}))$. 

\vspace{3pt}

To check that this ensemble satisfies Assumption~\ref{assum:iidJ}, we compute the characteristic function of each $q$-body coupling. Let $J^+, J^- \sim \Poisson(d/(2 n^{q-1}))$ and $J = (J^+ - J^-)/ \sqrt{d}$, then
\begin{equation}
\begin{aligned}
    \EV_{J}[e^{iJ \lambda}]  =&~ \EV_{J^+}[e^{i J^+ \lambda/\sqrt{d}}] \cdot \EV_{J^-}[e^{-i J^- \lambda/\sqrt{d}}] \\
    =&~ \exp\Big[ \frac{d}{2 n^{q-1}}(e^{i \lambda/\sqrt{d}} - 1) \Big] \times \exp\Big[ \frac{d}{2 n^{q-1}} (e^{- i\lambda/\sqrt{d}} - 1) \Big]\\
    =&~ \exp\Big[\frac{d}{n^{q-1}}
    \Big( \cos \big(\frac{\lambda}{\sqrt{d}} \big) - 1 \Big)\Big].
\end{aligned}
\end{equation}
Then we have
\begin{equation} \label{eq:ER-gqn}
    g_{q,n}^{(d)}(\lambda) = n^{q-1} \log \EV_{J}[e^{i J \lambda}] = d [\cos(\lambda/\sqrt{d})-1].
\end{equation}
Note $g_{q,n}^{(d)}(\lambda)$ is independent of $q$ and $n$, and is infinitely differentiable.
So the ensemble $\G^\ERsub_{d,q}$ satisfies Assumption~\ref{assum:iidJ}. Moreover, it can be directly checked that
\begin{align}
    \lim_{d \to \infty} \lim_{n \to \infty} \Big(g_{q,n}^{(d)}(\lambda), g_{q,n}^{(d)\prime}(\lambda)\Big) = \Big({-} \frac{\lambda^2}{2}, - \lambda \Big),
\end{align}
so the equivalence in \eqref{eq:equivalence} holds in particular for $\G_{d,q}^\ERsub$. This concludes the proof of Theorem~\ref{thm:dense_sparse_agreement}.

\begin{remark}
We can consider an alternate version of the \ER ensemble, where
\begin{align}
    J_{i_1,\ldots,i_q} = \begin{cases}
        0 & \text{w.p.}\quad 1-d/n^{q-1} \\
        +1/\sqrt{d}& \text{w.p.}\quad d/2n^{q-1} \\
        -1/\sqrt{d} & \text{w.p.}\quad d/2n^{q-1} \\
    \end{cases}.
\end{align}
Then
\begin{equation}
    g_{q,n}^{(d)}(x) = n^{q-1} \log \EV_J [e^{i J_{1,2,\ldots,q} x}] = n^{q-1} \log \Big[1-\frac{d}{n^{q-1}} + \frac{d}{n^{q-1}}\cos(\frac{x}{\sqrt{d}})\Big]
\end{equation}
and
\begin{equation}
    g_q^{(d)}(x) = \lim_{n\to\infty} g_{q,n}^{(d)}(x) = d\big[\cos(\frac{x}{\sqrt{d}})-1\big].
\end{equation}
Since this equals \eqref{eq:ER-gqn} and Assumption~\ref{assum:iidJ} is satisfied, the QAOA yields identical performance on this ensemble as on $\G^\ERsub_{d,q}$ in the $n\to\infty$ limit.
\end{remark}


\section{Proof of Theorem~\ref{thm:agree-regular} (Agreement of $\G_q$ with regular hypergraphs)}
\label{apx:regular}

To prove Theorem~\ref{thm:agree-regular}, we need to match $V_p(\G_q,\paramv)$ derived in \eqref{eq:Vp_pure} for the pure $q$-spin model to the performance of the QAOA for the Max-$q$-XORSAT problem on large-girth regular hypergraphs.
Specifically, we consider a $d$-regular $q$-uniform hypergraph with a set $\cE$ of hyperedges and girth $>2p+1$.
The notion of girth used here is the minimum length of Berge cycles in the hypergraph.
Each hyperedge $(i_1,\ldots,i_q)\in \cE$ is also associated with a weight $J_{i_1,\ldots,i_q} \in \{+1, -1\}$, and a XORSAT clause $z_{i_1} z_{i_2}\cdots z_{i_q} = J_{i_1,\ldots,i_q}$.
The cost function of the Max-$q$-XORSAT problem on this hypergraph can be written as
\begin{equation}
    C_J(\zv) =  \frac{1}{\sqrt{d}} \sum_{(i_1,\ldots,i_q) \in \cE} J_{i_1,\ldots,i_q} z_{i_1} \cdots z_{i_q}.
\end{equation}
This problem is studied in \cite{basso2021quantum}, where an iterative formula is derived to evaluate
\begin{equation}
    \nu_p^{[q]}(\paramv) = \frac{1}{\sqrt{2q}}\lim_{d\to \infty} \frac{d}{|E|} \braket{\gammav, \betav | C_J | \gammav,\betav}.
\end{equation}
This value is shown by the authors to be independent of the graph as well as the signs of the chosen $J_{i_1,\ldots,i_q}$, as a consequence of the fact that the QAOA at level $p$ only sees a hypertree neighborhood on these graphs.

Recall that
Theorem~\ref{thm:agree-regular} is the statement that for any such graph, $\sqrt{2}\nu_p^{[q]}(\sqrt{q}\paramv)=V_p(\G_q,\paramv)$.

\paragraph{Formula of $\nu_p^{[q]}(\paramv)$.}---
We now describe the formula in~\cite[Section 8]{basso2021quantum}.
Let
\begin{equation}
\mathcal{B} = \{(a_1,\ldots a_p, a_0, a_{-p}, \ldots, a_{-1}) : a_i = \pm 1\}    
\end{equation}
be the set of $(2p+1)$-bit strings used in \cite{basso2021quantum}.
Note this set differs from the set $A$ of $2p$-bit string used in the rest of the current paper.
For any $(2p+1)$-bit string $\av \in \mathcal{B}$, let
\begin{align}\label{eq:f_definition}
f(\av) &= \frac{1}{2} \braket{a_1 | e^{i \beta_1 X} | a_2} \cdots \braket{ a_{p-1} | e^{i \beta_{p-1} X} | a_p} \braket{a_p | e^{i \beta_p X} | a_0} \nonumber \\
    &\quad \times \braket{a_{0} | e^{-i \beta_p X} | a_{-p}} \braket{a_{-p} | e^{-i \beta_{p-1} X} | a_{-(p-1)}} \cdots \braket{a_{-2} | e^{-i \beta_1 X} | a_{-1}},
\end{align}
where $\braket{a_1|e^{i\beta X}|a_2} = (\cos\beta)^{(1+a_1 a_2)/2} (i\sin\beta)^{(1-a_1 a_2)/2}$.

Also, let $\Gammav = (\Gamma_1, \ldots, \Gamma_p, \Gamma_0, \Gamma_{-p}, \ldots, \Gamma_{-1}) \in \R^{2p+1}$, where for $1 \leq r \leq p$, we have
\begin{equation} \label{eq:Gamma-def}
    \Gamma_{r} = \gamma_r, \qquad
    \Gamma_0 = 0, \qquad
    \Gamma_{-r} = -\gamma_r.
\end{equation}
Now we define $p+1$ functions $H^{(m)} \colon \mathcal{B} \to \mathbb{C}$ for $0 \leq m \leq p$ as follows. 
Let $H^{(0)} = 1$ and, for $1 \leq m \leq p$,
\begin{equation}\label{eq:Hm_def}
    H^{(m)}(\av)= \exp \bigg[ {-}\frac{1}{2} \sum_{\bv^1,\ldots,\bv^{q-1} \in \mathcal{B}} \big( \Gammav \cdot (\av \bv^1 \cdots \bv^{q-1}) \big)^2 \prod_{i=1}^{q-1} [f(\bv^i) H^{(m-1)}(\bv^i)] \bigg].
\end{equation}
Starting at $m=0$ and going up by $p$ steps we arrive at $H^{(p)}$ which is used to compute
\begin{align}
     \nu_p^{[q]}(\paramv)
     &= \frac{i}{\sqrt{2q}} 
     \sum_{j=-p}^p \Gamma_j 
        \prod_{i=1}^q \Big(
		        \sum_{\av^i \in \mathcal{B}} a^i_0 a^i_j f(\av^i) H^{(p)}(\av^i)
        	\Big) \nonumber \\
&=\frac{i}{\sqrt{2q}} 
     \sum_{\av^1,\cdots,\av^q \in \mathcal{B}} [\Gammav \cdot (\av \bv^1 \cdots \bv^{q-1})] \prod_{i=1}^q [ a_0^i f(\av^i) H^{(p)}(\av^i) ].
     \label{eq:infinite_D_iter_result_general}
\end{align}

In the rest of this appendix, we will rewrite this formula for $\nu_p^{[q]}(\paramv)$ in the language of the present paper and show that it matches (up to a rescaling) the formula of $V_p(\G_q,\paramv)$ for the pure $q$-spin model.
Note we derived the latter explicitly in Eq.~\eqref{eq:Vp_pure} of the previous appendix.
This proof idea is essentially the same as the proof of Theorem 1 in~\cite[Section 6]{basso2021quantum}.

\paragraph{Relating $\{H^{(p)}(\av): \av\in \mathcal{B}\}$ to $\{W_\av: \av\in A\}$.}---
Since $H^{(m)}(\av)$ does not depend on $a_0$ given that $\Gamma_0 = 0$, we will slightly abuse notation below, and let $H^{(m)}(\av)$ take $\av \in A = \{\pm 1\}^{2p}$ or $\av \in \mathcal{B} = \{\pm 1\}^{2p+1}$ as argument.

Furthermore, for any $(2p+1)$-bit string $\av \in \mathcal{B}$, we associate it with a $2p$-bit string $\lo\av \in A$ via $\lo{a}_{\pm {r}} = a_{\pm r} a_{\pm(r+1)}$ for $1\le r\le p-1$, and $\lo{a}_{\pm p} = a_{\pm p} a_0$. More explicitly,
\begin{align}  \label{eq:lo-a-def}
\lo{a}_1 &= a_1 a_2,  &\ldots,&  
&\lo{a}_{p-1} &= a_{p-1} a_p,   &\lo{a}_p &= a_p a_0, \nonumber \\
\lo{a}_{-1} &= a_{-1} a_{-2}, &\ldots,&  
&\lo{a}_{-(p-1)} &= a_{-(p-1)} a_{-p},   &\lo{a}_{-p} &= a_{-p} a_0 .
\end{align}

Using the fact that $\braket{a_1 |e^{i\beta X} |a_2} = \braket{a_1 a_2 | e^{i\beta X} |1}$, we can rewrite $f(\av)$ for any $\av\in \mathcal{B}$ as
\begin{align} \label{eq:f-and-Q}
f(\av) &= \frac{1}{2} \braket{\lo{a}_1 | e^{i \beta_1 X} | 1} \cdots \braket{ \lo{a}_{p-1} | e^{i \beta_{p-1} X} | 1} \braket{\lo{a}_p | e^{i \beta_p B } |1} \nonumber \\
    &\quad \times \braket{1 | e^{-i \beta_p X} | \lo{a}_{-p}} \braket{1 | e^{-i \beta_{p-1} X} | \lo{a}_{-(p-1)}} \cdots \braket{1 | e^{-i \beta_1 X} | \lo{a}_{-1}} \nonumber \\
    &= \frac12 Q_{\lo\av}
\end{align}
where $Q_{\lo\av}$ for $\lo\av \in A$ is as defined in Eq.~\eqref{eq:Qdef}.

We also define a ``star'' operation ($*$) on $\av \in A$ which gives $\av^*\in A $ whose entries are
\begin{align}\label{eq:star_def}
a^*_r=a_r a_{r+1}\cdots a_p
\qquad \text{and} \qquad
a^*_{-r} = a_{-r} a_{-r-1}\cdots a_{-p}
\qquad
\text{for } 1\le r\le p \,.
\end{align}
Take care to note that in the current proof, $*$ always refers to the above star operation and not complex conjugation nor the Kleene star. 

Observe that 
\begin{equation} \label{eq:a-hat-star}
    \lo{a}^*_{\pm r} = \lo{a}_{\pm r} \cdots \lo{a}_{\pm p}  = a_{\pm r} a_0
    \qquad \forall \av\in \mathcal{B}.
\end{equation}
Then, for any $\av\in \mathcal{B}$, we can write $\Phi_{\lo\av}$ defined in Eq.~\eqref{eq:Phi_def} as
\begin{equation} \label{eq:Phi-and-Gamma}
\Phi_{\lo\av} = \sum_{r=1}^p \gamma_r (\lo{a}^*_r - \lo{a}^*_{-r}) = \sum_{r=1}^p \gamma_r (a_r - a_{-r}) a_0 = (\Gammav \cdot \av) a_0 .
\end{equation}
Since $a_0^2=1$, we get $(\Gammav\cdot \av)^2 = \Phi_{\lo\av}^2$.
Also note $\widehat{\av\bv} = \lo\av\lo\bv$, so we can rewrite Eq.~\eqref{eq:Hm_def} as
\begin{equation} \label{eq:Hm-Q-with-half}
    H^{(m)}(\av) = \exp \Big[ -\frac{1}{2} \sum_{\bv^1,\ldots,\bv^{q-1} \in \mathcal{B}} \Phi^2_{\lo\av \lo\bv^1 \cdots \lo\bv^{q-1}} \prod_{i=1}^{q-1} [\frac{1}{2}Q_{\lo\bv^i} H^{(m-1)}(\bv^i)] \Big].
\end{equation}
Note from Eq.~\eqref{eq:a-hat-star} that
\begin{align} \label{eq:b-hat-star}
    H^{(m)}(\lo\bv^*) = H^{(m)}(\bv b_0)
\qquad \text{ for any } \bv\in \mathcal{B}.
\end{align}
It is clear from Eq.~\eqref{eq:Hm_def} that $H^{(m)}(\av) = H^{(m)}(-\av)$, thus we have
\begin{align} \label{eq:Hm-lo-star}
H^{(m)}(\lo\bv^*) =  H^{(m)}(\bv)
\qquad \text{ for any } \bv\in \mathcal{B}.
\end{align}
Hence, in Eq.~\eqref{eq:Hm-Q-with-half}, we can sum over $\lo\bv^i \in A$ instead of $\bv^i \in \mathcal{B}$, each time killing a 1/2 factor from the redundancy of the sum over $b_0^i$.
Also using the above equation we can replace $H^{(m)}(\av^i) = H^{(m)}(\lo\av^{i*})$, where $\lo\av^{i*}=  \text{star}(\text{hat}(\av^i))\in A$ for any $\av^i\in\mathcal{B}$.
Then we write Eq.~\eqref{eq:Hm-Q-with-half} as
\begin{align}
        H^{(m)}(\lo\av^*) = \exp \Big[ -\frac{1}{2} \sum_{\lo\bv^1,\ldots,\lo\bv^{q-1} \in A} \Phi^2_{\lo\av \lo\bv^1 \cdots \lo\bv^{q-1}} \prod_{i=1}^{q-1} [Q_{\lo\bv^i} H^{(m-1)}(\lo\bv^{i*})] \Big].
\end{align}
We can then drop the hats and rewrite this as (for any $\av \in A$)
\begin{align}\label{eq:Hm-Q-Phi}
        H^{(m)}(\av^*) = \exp \Big[ -\frac{1}{2} \sum_{\bv^1,\ldots,\bv^{q-1} \in A} \Phi^2_{\av \bv^1 \cdots \bv^{q-1}} \prod_{i=1}^{q-1} [Q_{\bv^i} H^{(m-1)}(\bv^{i*})] \Big].
\end{align}

Now let us define for $0\le m\le  p$ and any $\av\in A$
\begin{align}
R^{(m)}_{{}\av} := Q_{{}\av} H^{(m)}({}\av^*) .
\end{align}
Thus we have $R^{(0)}_{{}\av} = Q_{{}\av}$.
Plugging Eq.~\eqref{eq:Hm-Q-Phi} into the above yields
\begin{align} \label{eq:R-iter}
R^{(m)}_{{}\av} 
= Q_{\av} \exp \Big[ -\frac{1}{2} \sum_{\bv^1,\ldots,\bv^{q-1} \in A} \Phi^2_{\av \bv^1 \cdots \bv^{q-1}} \prod_{i=1}^{q-1} R^{(m-1)}_{\bv^i} \Big].
\end{align}

Note $R_\av^{(m)}$ has a similar description as $W_\av^\pure$ given in Eq.~\eqref{eq:W_pure_self_consistent}. We will prove that $R_\av^{(p)} = W_\av^\pure$. First, we state a lemma whose proof we defer to the end of this section.
\begin{lemma}\label{lem:Hp_fixed_point}
$H^{(p)}(\av)$ is a fixed-point of the iteration in Eq.~\eqref{eq:Hm_def}. 
\end{lemma}

By Lemma~\ref{lem:Hp_fixed_point}, $R^{(p)}_\av$ is a fixed point of the iteration in Eq.~\eqref{eq:R-iter} 
\begin{equation} \label{eq:R-SCE}
    R^{(p)}_{{}\av} 
= Q_{\av} \exp \Big[ -\frac{1}{2} \sum_{\bv^1,\ldots,\bv^{q-1} \in A} \Phi^2_{\av \bv^1 \cdots \bv^{q-1}} \prod_{i=1}^{q-1} R^{(p)}_{\bv^i} \Big].
\end{equation}
Note that this is almost the same self-consistent equation~\eqref{eq:W_pure_self_consistent} for the $W_\av^\pure$, except for a factor of $q$ in the exponential.
This factor can be fixed by a rescaling of $\gammav$.
More formally, let $R_\av^{(p)}(\paramv)$ be the solution to \eqref{eq:R-SCE} with $\Phi_\av(\gammav)$ upstairs.
Note from the definition of $\Phi_\av$ in \eqref{eq:Phi_def}, we have $\Phi_\av(\sqrt{q}\gammav) = \sqrt{q}\Phi_\av(\gammav)$.
Hence, $R_\av^{(p)}(\sqrt{q}\paramv)$ is the solution
\begin{equation}
    R^{(p)}_{{}\av}(\sqrt{q}\paramv) 
= Q_{\av} \exp \Big[ -\frac{q}{2} \sum_{\bv^1,\ldots,\bv^{q-1} \in A} \Phi^2_{\av \bv^1 \cdots \bv^{q-1}}(\gammav) \prod_{i=1}^{q-1} R^{(p)}_{\bv^i}(\sqrt{q}\paramv) \Big].
\end{equation}

Now the above equation agrees exactly with Eq.~\eqref{eq:W_pure_self_consistent} for $\{W_\av^\pure\}_{\av\in A}$.
Since the polynomial upstairs is well-played by Lemma~\ref{lem:P-is-well-played}, this equation has a unique solution due to Lemma~\ref{lem:unique_solution_SCE}, and so
\begin{align}\label{eq:R_is_Wpure}
    R^{(p)}_\av(\sqrt{q}\paramv) = W_\av^\pure(\paramv), \quad \text{for all } \av \in A.
\end{align}

\paragraph{Agreement of $V_p(\G_q,\paramv)$ with $\nu_p^{[q]}(\paramv)$.}---
The remaining task is to relate $\nu_p^{[q]}(\paramv)$ with $V_p(\G_q,\paramv)$.
Returning to Eq.~\eqref{eq:infinite_D_iter_result_general}, we can modify it using Eqs.~\eqref{eq:f-and-Q}, \eqref{eq:Phi-and-Gamma} and \eqref{eq:Hm-lo-star} to obtain
\begin{align}
\nu_p^{[q]}(\paramv)
     &= \frac{i}{\sqrt{2q}} \sum_{\av^1,\cdots,\av^q \in \mathcal{B}} \Phi_{\lo\av \lo\bv^1 \cdots \lo\bv^{q-1}} \prod_{i=1}^q [ \frac{1}{2} Q_{\lo\av} H^{(p)}(\lo\av^{i*}) ].
\end{align}
Since the summand is independent of $a_0^i$, we can sum over $\lo\av^i \in A$ instead of $\av^i \in \mathcal{B}$, killing factors of $1/2$ to get
\begin{equation}
    \nu_p^{[q]}(\paramv) = \frac{i}{\sqrt{2q}} \sum_{\lo\av^1,\cdots,\lo\av^q \in A} \Phi_{\lo\av \lo\bv^1 \cdots \lo\bv^{q-1}} \prod_{i=1}^q R^{(p)}_{\lo\av^i} = \frac{i}{\sqrt{2q}} \sum_{\av^1,\cdots,\av^q \in A} \Phi_{\av \bv^1 \cdots \bv^{q-1}} \prod_{i=1}^q R^{(p)}_{\av^i}
\end{equation}
where we dropped the hats in the second equality.
Thus, keeping track of explicit dependence on $(\paramv)$, we have
\begin{align}
    \nu_p^{[q]}(\sqrt{q}\paramv) &= \frac{i}{\sqrt{2}} \sum_{\av^1,\cdots,\av^q \in A} \Phi_{\av \bv^1 \cdots \bv^{q-1}}(\gammav) \prod_{i=1}^q R^{(p)}_{\av^i}(\sqrt{q}\paramv) \nonumber\\
    &=\frac{i}{\sqrt{2}} \sum_{\av^1,\cdots,\av^q \in A} \Phi_{\av \bv^1 \cdots \bv^{q-1}}(\gammav) \prod_{i=1}^q W_{\av^i}^\pure(\paramv) = \frac{1}{\sqrt{2}} V_p(\G_q,\paramv).
\end{align}
The last equality follows from Eq.~\eqref{eq:Vp_pure}.
This proves Theorem~\ref{thm:agree-regular}.

We now prove the lemma used above as promised.
\begin{proof}[{\bf Proof of Lemma~\ref{lem:Hp_fixed_point}}]
$ $\newline
This proof will use notation solely in the context of~\cite{basso2021quantum}. In that reference, combining Eqs. (8.20) and (8.21), for $0 \leq m \leq p$, $H^{(m)}(\av)$ can be written as
\begin{align}
    H^{(m)}(\av) = \exp \Big[ -\frac{1}{2} \sum_{j,k=-p}^p \Gamma_j \Gamma_k a_j a_k \big( G^{(m-1)}_{j,k} \big)^{q-1}  \Big]
\end{align}
where $\{G^{(k)}\}_{1\leq k \leq p}$ are matrices with entries $\{G^{(k)}\}_{j,k}$ and $j,k \in \{-p, \ldots, -1, 0, 1, \ldots, p\}$.

If the iteration were to continue for one more step, we would have
\begin{align}\label{eq:H_p_plus_one}
    H^{(p+1)}(\av) = \exp \Big[ -\frac{1}{2} \sum_{j,k=-p}^p \Gamma_j \Gamma_k a_j a_k \big( G^{(p)}_{j,k} \big)^{q-1}  \Big]
\end{align}
As discussed in~\cite[Section 3.2]{basso2021quantum}, $G^{(p-1)}_{j,k} = G^{(p)}_{j,k}$ except when either $j=0$ or $k=0$. Recalling that $\Gamma_0 = 0$, we can modify Eq.~\eqref{eq:H_p_plus_one} into
\begin{align}
    H^{(p+1)}(\av) &= \exp \Big[ -\frac{1}{2} \sum_{j,k=-p}^p \Gamma_j \Gamma_k a_j a_k \big( G^{(p-1)}_{j,k} \big)^{q-1}  \Big] \nonumber \\
    &= H^{(p)}(\av).
\end{align}
This proves the lemma.
\end{proof}

\bibliographystyle{halphalab}
\bibliography{refs}

\end{document}